\documentclass[pra,twocolumn]{revtex4-1}%
\usepackage{amsfonts}
\usepackage{amsmath}
\usepackage{amssymb}
\usepackage{graphicx}
\usepackage[colorlinks=true,linkcolor=blue,citecolor=red,plainpages=false,pdfpagelabels]{hyperref}%
\usepackage{enumitem}
\usepackage{cleveref}
\usepackage{float}
\usepackage{xcolor}
\usepackage{mathrsfs}
\DeclareMathOperator{\Tr}{Tr}

\DeclareMathOperator{\sq}{sq}

\newtheorem{theorem}{Theorem}
\newtheorem{lemma}[theorem]{Lemma}
\newtheorem{remark}{Remark}
\newtheorem{proposition}{Proposition}
\newtheorem{condition}{Condition}
\newtheorem{definition}{Definition}
\newtheorem{corollary}{Corollary}
\newenvironment{proof}[1][Proof]{\noindent\textbf{#1.} }{\ \rule{0.5em}{0.5em}}

\begin{document}
\title[ ]{Energy-constrained two-way assisted private and quantum capacities\\of quantum channels}
\author{Noah Davis}
\affiliation{Hearne Institute for Theoretical Physics, Department of Physics and Astronomy, Louisiana State University, Baton Rouge, Louisiana 70803, USA}

\author{Maksim E. Shirokov}
\affiliation{Steklov Mathematical Institute, Russian Academy of Sciences, Moscow, Russia}

\author{Mark M. Wilde}
\affiliation{Hearne Institute for Theoretical Physics, Department of Physics and Astronomy,
Center for Computation and Technology,
Louisiana State University, Baton Rouge, Louisiana 70803, USA}

\begin{abstract}
	With the rapid growth of quantum technologies, knowing the fundamental characteristics of quantum systems and protocols is essential for their effective implementation. A particular communication setting that has received increased focus is related to quantum key distribution and distributed quantum computation. In this setting, a quantum channel connects a sender to a receiver, and their goal is to distill  either a secret key or entanglement, along with the help of arbitrary local operations and classical communication (LOCC). In this work, we establish a general theory of energy-constrained, LOCC-assisted private and quantum capacities of quantum channels, which are the maximum rates at which an LOCC-assisted quantum channel can reliably establish secret key or entanglement, respectively, subject to an energy constraint on the channel input states. We prove that the energy-constrained squashed entanglement of a channel is an upper bound on these capacities. We also explicitly prove that a thermal state  maximizes a relaxation of the squashed entanglement of all phase-insensitive, single-mode input bosonic Gaussian channels, generalizing results from prior work. After doing so, we prove that a variation of the method introduced in [Goodenough \textit{et al.}, New J.~Phys.~18, 063005 (2016)] leads to improved upper bounds on the energy-constrained secret-key-agreement capacity of a bosonic thermal channel. We then consider a multipartite setting and prove that two known multipartite generalizations of the squashed entanglement are in fact equal. We finally  show that the energy-constrained, multipartite squashed entanglement plays a role in bounding the energy-constrained LOCC-assisted private and quantum capacity regions of quantum broadcast channels.
\end{abstract}

\date{\today}

\maketitle

\section{Introduction}
Modern communications from simple web browsing to high-security, governmental discussions rely on encryption protocols that use a private key to secure and interpret messages. The strength of the encryption is directly tied to the security of the key, and the security of most systems currently in use rests on computational assumptions. In contrast, quantum communication allows for generating an information-theoretically secure key, shared among trusted parties, via a method known as quantum key distribution (QKD) \cite{brass,E91,SBCDLP09}.

The rate (bits of key per channel use) at which QKD can be accomplished using a variety of protocols is known to fall off exponentially with distance \cite{brass,E91,GG02,SBCDLP09}. This observed rate-loss trade-off previously suggested the question of whether some other protocols could be designed to outperform the exponential fall-off. However, the exponential rate-loss trade-off has been established to be a fundamental limit for bosonic loss channels \cite{tgw}, and a number of works \cite{STW16,goodEW16,PLOB17,tsw,AML16,WTB17,Christandl2017,appxpriv,BA17,RGRKVRHWE17} have now considered this problem and generalizations of it after \cite{tgw} appeared. The tightest known non-asymptotic bounds for the pure-loss bosonic channels have been given in \cite{tgw,WTB17,KW17}.

An important notion in addressing this question is the capacity of a quantum channel, which is a fundamental characteristic of the channel and is independent of any  specific communication protocol. In the setting of quantum key distribution, it is natural to allow for an authenticated, public classical channel to assist the quantum channel connecting two parties, and so the quantum channel is said to be assisted by local operations and classical communication (LOCC). The secret-key-agreement capacity is the maximum rate at which classical bits can be privately and faithfully transmitted through many uses of a channel, while allowing for free classical communication \cite{tgw,tgwB}. Similarly, the LOCC-assisted quantum capacity is the maximum rate at which qubits can be transmitted faithfully through many uses of a channel and with free classical communication \cite{BBPSSW96EPP,BDSW96,Mul12}. Ultimately, the capacity of a quantum channel limits the usefulness of the channel, and so these LOCC-assisted private and quantum capacities of a quantum channel are important factors in determining any practical use of the channel.

Given that current communications (particularly quantum communication experiments) utilize photons, it is important to consider capacities of bosonic, Gaussian channels \cite{HW01}. Expressions for the unassisted quantum capacities of channels such  as the single-mode quantum-limited amplifier and attenuator have been presented in \cite{HW01,WPG07}, but the expressions therein suppose that the transmitters in question have no constraint on their energy consumption. While these bounds have been shown to depend only on fundamental characteristics of the channel, any real transmitter will not have unbounded energy available, and so these expressions may have limited applicability to practical scenarios, as argued in \cite{qi}.  More recently, strides have been made to bound capacities in energy-constrained scenarios \cite{tgw,qi,goodEW16}. 

In the effort to bound these capacities, several information measures of quantum channels have been proposed, each of which is based on correlation measures for bipartite quantum states. Among these, an entanglement measure \cite{H42007} known as the squashed entanglement  \cite{christ} has played a critical role, as shown in \cite{tgwB,tgw,STW16,goodEW16}. This is due in part to the fact that it possesses several desirable properties, such as additivity, monotonicity under LOCC, uniform (asymptotic) continuity, and faithfulness \cite{christ,AF04,BCY11}. Most recently, squashed entanglement has been shown to retain several of these attributes for infinite-dimensional states, which allows for its use in rather general scenarios \cite{maks}.

In this paper, we formally define the task of energy-constrained, LOCC-assisted private and quantum communication, and we show that the energy-constrained squashed entanglement of a channel is an upper bound on its corresponding capacities.
We prove this bound in a rather general, infinite-dimensional setting, allowing for applications to physical situations other than those specifically considered in this paper. In this sense, this paper is complementary to the developments from \cite{qi} and generalizes those from \cite{tgw,tgwB,STW16,goodEW16}.
We then prove that a thermal state is the optimal input to a relaxation of the energy-constrained squashed entanglement of phase-insensitive, single-mode input, bosonic Gaussian channels, which extends various statements from prior work.
After doing so, we prove that a variation of the method introduced in \cite{goodEW16} leads to improved upper bounds on the energy-constrained secret-key-agreement capacity of a bosonic thermal channel.
In particular, these improved upper bounds have the property that they converge to zero in the limit as the thermal channel becomes entanglement breaking.
We finally prove that the two most common multipartite generalizations  of the squashed entanglement from \cite{YHHHOS,AHS08} are in fact equal  to one another, and we show how the general framework developed in this paper applies to energy-constrained capacity regions of quantum broadcast channels.

We begin in Section~\ref{sec:background} by giving an introduction to  notation, tools, and terminology, as well as defining the important quantities used in the following sections. In Section~\ref{s:CQMIinInfDim} we prove two useful lemmas about the conditional quantum mutual information (CQMI) of infinite-dimensional states. Section~\ref{s:ECSKAC} formally defines the task of energy-constrained secret key agreement. We go on to show that the squashed entanglement of a channel is an upper bound on its energy-constrained LOCC-assisted private and quantum capacities in Section~\ref{s:EsqBound}. Section~\ref{s:realizing} shows that the bosonic thermal state is the optimal input to particular bosonic Gaussian channels in order to maximize relaxations of their squashed entanglement.
A subsection of Section~\ref{s:realizing}presents the improved upper bounds on the energy-constrained secret-key-agreement capacity of a bosonic thermal channel.
 Sections~\ref{s:mcqmi} and \ref{s:equivEsq} begin the multipartite segment of this paper by proving the duality of two different multipartite generalizations of conditional quantum mutual information, which implies the equivalence of two multipartite squashed entanglements that have appeared in the literature \cite{YHHHOS,AHS08,STW16} and were previously thought to be different. In Section~\ref{s:multiPrimer} more tools for working in a multipartite setting are defined. Broadcast channels are introduced, and the private communication protocol from Section~\ref{s:ECSKAC} is recast with multiple receivers in Section~\ref{s:broadcast}. The energy-constrained multipartite squashed entanglement is shown in Section~\ref{s:EsqBroad} to upper bound the energy-constrained LOCC-assisted capacities of broadcast channels with squashed entanglements depending on the partitions of systems. Finally, a calculation of the rate bounds is presented in Section~\ref{s:BoundCalcs} before closing thoughts are given in Section~\ref{s:concl}.

\section{Background: Quantum Information Preliminaries}

\label{sec:background}

\subsection{Quantum Systems, States, and Channels}

In order to study the quantum aspects of information and communication, we first review foundational aspects, consisting of terms and measures which serve to describe and quantify key features of the systems in question, as well as the operations performed on those systems. The reader can find  background other than that presented here by consulting \cite{H06book,H12,HZ12,W16}.

We denote some first Hilbert space as $\mathcal{H}_{A}$ and another one as $\mathcal{H}_{B}$. Throughout, the Hilbert spaces we consider are generally infinite-dimensional and separable, unless stated otherwise. The tensor product of $\mathcal{H}_{A}$ and $\mathcal{H}_{B}$ is itself a Hilbert space, represented as $\mathcal{H}_{A}\otimes\mathcal{H}_{B}=\mathcal{H}_{AB}$. 
Let $\mathcal{L}(\mathcal{H}_{A})$ denote the set of bounded linear  operators acting on $\mathcal{H}_{A}$, and 
let $\mathcal{L}_{+}(\mathcal{H}_{A})$ denote the subset of positive, semi-definite operators acting on $\mathcal{H}_{A}$.
Let $\mathcal{L}_1(\mathcal{H})$ denote the set of trace-class operators, those
operators $X$ for which the trace norm is finite:$\ \left\Vert X\right\Vert
_{1}\equiv\operatorname{Tr}\{\left\vert X\right\vert \}<\infty$, where
$\left\vert X\right\vert \equiv\sqrt{X^{\dag}X}$. 
The set of states (also called density operators) $\mathcal{D}(\mathcal{H}_{A})\subset\mathcal{L}_{+}(\mathcal{H}_{A})$ contains all operators  $\rho_A \in\mathcal{L}_{+}(\mathcal{H}_{A})$ such that $\Tr\{\rho_A\}=1$. The state $\rho_{AB}\in\mathcal{D}(\mathcal{H}_{AB})$ is called an extension of a state $\rho_{A}\in\mathcal{D}(\mathcal{H}_{A})$ if $\rho_{A}=\Tr_{B}\{\rho_{AB}\}$, where $\Tr_{B}$ denotes the partial trace over $\mathcal{H}_{B}$.

Every density operator $\rho \in \mathcal{D}(\mathcal{H})$ can be expressed in terms of a spectral decomposition of a countable number of eigenvectors and eigenvalues: \begin{equation} \rho=\sum_{i}p_{i}|\phi_{i} \rangle\langle\phi_{i}| ,\label{eq:specDecomp} \end{equation} where the probabilities $\{p_{i}\}_i$ are the eigenvalues and $\{|\phi_{i}\rangle\}_i$ are the eigenvectors. A state $\rho\in \mathcal{D}(\mathcal{H})$ is called a pure state if there exists a unit vector $|\psi\rangle\in \mathcal{H}$ such that $\rho=|\psi\rangle\langle\psi|$. When this is not the case, we say that the state is a mixed state, because a spectral decomposition indicates that any state can be interpreted as a probabilistic mixture of pure states. 

We can purify a state $\rho_{A}=\sum_{i}p_{i}|\phi_{i}\rangle\langle\phi_{i}|_{A}$ by introducing a set of orthonormal vectors $\{|i\rangle_{R}\}_i$ and extending it to a pure state in the tensor-product space $\mathcal{H}_{RA}$. Then \begin{equation}
|\psi\rangle_{RA}=\sum_{i}\sqrt{p_{i}}|\phi_{i}\rangle_{A}|i\rangle_{R} \label{eq:purification}
\end{equation}
is a unit vector in $\mathcal{H}_{RA}$, and $\rho_{RA}=|\psi\rangle \langle\psi|_{RA}$ is a pure state in $\mathcal{D}(\mathcal{H}_{RA})$. A state purification is a special kind of extension, given that $\rho_{A}=\Tr_{R}\{\rho_{RA}\}$.

A key feature of quantum systems is the phenomenon of entanglement \cite{H42007}. A state made up of multiple systems is said to be entangled if it cannot be written as a probabilistic mixture of product states. For example, $\rho_{AB}=\sum_{z}p_{Z}(z)|\psi^{z}\rangle\langle\psi^{z} |_{A}\otimes |\phi^{z}\rangle\langle\phi^{z}|_{B}$ represents an unentangled, separable state in $\mathcal{D}(\mathcal{H}_{AB})$ \cite{W89}, where $p_{Z}(z)$ is a probability distribution and $\{|\psi^{z}\rangle_A\}_z$ and $\{|\phi^{z}\rangle_B\}_z$ are sets of unit vectors.

The Schmidt decomposition theorem gives us a tool for simplifying the form of pure, two-party (bipartite) states and particularly for determining whether a pure, bipartite state is entangled. An arbitrary bipartite unit vector $|\psi\rangle_{AB}$ can be written as $|\psi\rangle_{AB}=\sum_{i}\sqrt{p_{i}}|i\rangle_{A}|i\rangle_{B}$ where $\{|i\rangle_A\}_i$ and $\{|i\rangle_B\}_i$ are orthonormal bases in the Hilbert spaces $\mathcal{H}_{A}$ and $\mathcal{H}_{B}$, respectively, and $\{p_{i}\}_i$ are strictly positive, real probabilities. The set $\{\sqrt{p_{i}}\}_{i}$ is the set of Schmidt coefficients. For finite-dimensional $|\psi\rangle_{AB}$, the number $d$ of Schmidt coefficients is called the Schmidt rank of the vector, and it satisfies the following inequality: $d\leq\min\left[\dim(\mathcal{H}_{A}),\dim(\mathcal{H}_{B})\right]$. For infinite-dimensional $|\psi\rangle_{AB}$, the Schmidt rank $d$ can clearly be equal to infinity. The state $|\psi\rangle_{AB}$ is an entangled state if and only if $d\geq 2$. For finite-dimensional $\mathcal{H}_{A}$ and $\mathcal{H}_B$,  such that
$\mathcal{H}_{A}$ is isomorphic to $\mathcal{H}_B$,
we define a maximally entangled state in terms of the following unit vector:
\begin{equation}
	|\Phi\rangle_{AB}=\frac{1}{\sqrt{d}} \sum_{i=1}^{d}|i\rangle_{A}|i\rangle_{B}. \label{eq:maximallyentangled}
\end{equation}

According to the Choi-Kraus theorem, a linear map $\mathcal{N}_{A\to B}$ from $\mathcal{L}_1(\mathcal{H}_{A})$ to $\mathcal{L}_1(\mathcal{H}_{B})$ is completely positive and trace preserving (CPTP) if and only if it can be written in the following way:
\begin{equation}
	\mathcal{N}_{A\to B}(X_{A})=\sum_{l}V_{l}X_{A}V_{l}^{\dagger}  ,\label{eq:kraus}
\end{equation}
where $X_{A} \in \mathcal{L}_1(\mathcal{H}_{A})$, $V_{l}$ is a bounded linear operator mapping $\mathcal{H}_{A}\to\mathcal{H}_{B}$, and
$
	\sum_{l}V_{l}^{\dagger}V_{l}=I_{A}$.
 This is called the Choi-Kraus representation, and $\{V_{l}\}_{l}$ is called the set of Kraus operators. Such a linear map is referred to as a quantum channel, and it takes quantum states to other quantum states. Quantum channels can be concatenated in a serial or parallel way, and such a combination is also a quantum channel.

An isometric extension $U_{A\to BE}^{\mathcal{N}}$ of a quantum channel
$\mathcal{N}_{A\to B}$
 is a linear isometry taking $\mathcal{H}_A$ to $\mathcal{H}_B \otimes \mathcal{H}_E$, satisfying
\begin{equation}
  \mathcal{N}_{A\to B}(X_A)
  = \Tr_{E}\{\mathcal{U}_{A\to BE}^{\mathcal{N}}(X_{A})\}, \label{eq:iso-extend-condition}
\end{equation}
for all  $X_A \in \mathcal{L}_1(\mathcal{H}_A)$, where the isometric channel $\mathcal{U}_{A\to BE}^{\mathcal{N}}$ is defined in terms of the isometry $U_{A\to BE}^{\mathcal{N}}$ as
\begin{equation} \mathcal{U}_{A\to BE}^{\mathcal{N}}(X_{A})=U_{A\to BE}^{\mathcal{N}}X_{A}(U_{A\to BE}^{\mathcal{N}})^\dagger \label{eq:extendedChannel}. \end{equation} 
 We can construct a canonical isometric extension of a quantum channel in the following way:
\begin{equation}
	U_{A\to BE}^{\mathcal{N}}=\sum_{l}V_{l}\otimes|l\rangle_{E},\label{eq:extendedKraus}
\end{equation}
where $\{|l\rangle_E\}_l$ is an orthonormal basis. One can check that \eqref{eq:iso-extend-condition} is satisfied for this choice.

An isometric extension of a quantum channel shows that we can think of a channel as involving not only a sender and receiver but also a passive environment represented by system $E$ above. In order to determine the output of the extended channel $\mathcal{U}_{A\to BE}^{\mathcal{N}}$ to the environment, we simply trace over the output system $B$ instead of the environment $E$. The resulting channel is known as a complementary channel \cite{cmp2005dev,H07,KMNR07} (sometimes ``conjugate channel''), with the following action on an input state~$\rho_A$: \begin{equation}
\hat{\mathcal{N}}_{A \to E}(\rho_{A})=\Tr_{B}\{\mathcal{U}_{A\to BE}^{\mathcal{N}}(\rho_{A})\}. \label{complementaryChannel} \end{equation}
A channel complementary to $\mathcal{N}_{A\to B}$ is a CPTP map from $\mathcal{L}_1(\mathcal{H}_{A})$ to $\mathcal{L}_1(\mathcal{H}_{E})$ and is unique up to an isometry acting on the space $\mathcal{H}_{E}$ (see, e.g., \cite{H12,W16}).

The quantum instrument formalism provides the most general description of a quantum measurement \cite{DL70}. A quantum instrument is a set of completely positive, trace non-increasing maps $\left\{\mathcal{M}^{x}_{A \to B}\right\}_x$ such that the sum map $\sum_x 
\mathcal{M}^{x}_{A \to B}$ is a quantum channel \cite{DL70}. One can equivalently think of it as a quantum channel that takes as input a quantum system and gives as output both a quantum system and a classical system: \begin{equation}
\mathcal{M}_{A\to BX}(\rho_A)=\sum_{x} \mathcal{M}^{x}_{A \to B}(\rho_A)\otimes |x\rangle\langle x|_{X}. \label{eq:quantumInstrument} \end{equation}
Here $\{|x\rangle\}_{x}$ is a classical orthonormal basis identified with the outcomes of the instrument. Throughout this paper, we consider only the case when the measurement has a finite or countable number of outcomes.

In discussing quantum systems corresponding to tensor-product Hilbert spaces, it is useful to consider which parties can influence which subsystems, and we give names to the parties corresponding to the label on their subsystem. For example, it is conventional to say that Alice has access to system $A$, Bob to system $B$, and Eve to system $E$, which we often refer to as the environment as well. Eve is so named because the third party is regarded as a passive adversary or eavesdropper in a cryptographic context. By taking system $E$ to encompass anything not in another specified system, we can consider the most general cases of Eve's participation.

In what follows, we consider the use of a quantum channel interleaved with rounds of local operations and classical communication (LOCC). These rounds of LOCC can be considered channels themselves as follows:
\begin{enumerate}
	\item Alice performs a quantum instrument on her system, resulting in both quantum and classical outputs.
	\item Alice sends a copy of the classical output to Bob.
	\item Bob performs a quantum channel on his system conditioned on the classical data that he receives from Alice.
	\item Bob then performs a quantum instrument on his system and forwards the classical output to Alice.
	\item Finally, Alice performs a quantum channel on her system conditioned on the classical data from Bob.
	\item Iterate the above steps an arbitrarily large, yet finite number of times.
\end{enumerate}

The sequence of actions in the first through third steps is called ``local operations and one-way classical communication,'' and they can be expressed as a quantum channel of the following form:
\begin{equation}
\mathcal{S}_{AB}\equiv\sum_{z}\mathcal{G}_{A}^{z}\otimes \mathcal{J}_{B}^{z}, \label{onewayLOCCchannel}
\end{equation}
where $\{\mathcal{G}_{A}^{z}\}_{z}$ is a countable set of completely positive, trace non-increasing maps, such that the sum map $\sum_z \mathcal{G}_{A}^{z}$ is trace preserving, and  $\{\mathcal{J}_{B}^{z}\}_{z}$ is a set of channels. These conditions imply that $\mathcal{S}_{AB}$ is a channel. The fourth and fifth  steps above can also take the form of \eqref{onewayLOCCchannel} with the system labels reversed.

As indicated above, a full round of LOCC consists of the concatenation of some number of these channels back and forth between Alice and Bob \cite{BDSW96,CLMOW14}. This concatenation is a particular kind of separable channel  and takes the form
\begin{equation}
\mathcal{L}_{AB}\equiv\sum_{y}\mathcal{E}_{A}^{y}\otimes \mathcal{F}_{B}^{y} \label{LOCCchannel},
\end{equation}
where $\{\mathcal{E}_{A}^{y}\}_{y}$ and $\{\mathcal{F}_{B}^{y}\}_{y}$ are countable sets of completely positive, trace non-increasing maps such that $\mathcal{L}_{AB}$ is CPTP. We stress again that we only consider LOCC channels with a finite or countable number of classical values, and we refer to them as countably decomposable LOCC channels.

\subsection{Trace Distance and Quantum Fidelity}

We defined the trace norm $\|X\|_{1}$ of an operator $X$ previously.
Being a norm, it is homogeneous,  non-negative definite, and obeys the triangle inequality. It is also convex and invariant under multiplication by isometries; i.e., 
 for $\lambda\in[0,1]$, we have that $\|\lambda X+(1-\lambda)Y\|_{1}\leq\lambda\|X\|_{1}+(1-\lambda)\|Y\|_{1}$, and for isometries $U$ and $V^{\dagger}$, we have that $\|UXV^{\dagger}\|_{1}=\|X\|_{1}$.

The trace norm of an operator leads to  the trace distance between two of them. The trace distance between two density operators $\rho$ and $\sigma$ quantifies the distinguishability of the two states \cite{H69,H73,Hel76} and satisfies the inequality: $ 0\leq\|\rho-\sigma\|_{1}\leq 2 $.
From the triangle inequality, we see that the trace distance is maximized  for orthogonal states; i.e., when $\rho \sigma = 0$, then $ \|\rho-\sigma\|_{1}=\|\rho\|_{1}+\|\sigma\|_{1}=2$. 
Note that sometimes we employ the normalized trace distance, which is equal to half the usual trace distance:  $0\leq\frac{1}{2}\|\rho-\sigma\|_{1}\leq 1$.

Another way to measure the closeness of quantum states is given by the quantum fidelity \cite{uhlmann}. The pure-state fidelity for pure-state vectors $|\psi\rangle_{A}$ and $|\phi\rangle_{A}$ is given by
\begin{equation} F(\psi_{A},\phi_{A})\equiv |\langle\psi|\phi\rangle_{A} |^{2}, \label{fidelity} \end{equation}
from which we conclude that $0\leq F(\psi_{A},\phi_{A})\leq 1$.
The general definition of the fidelity for arbitrary density operators 
$\rho_{A}$ and $\sigma_{A}$ is as follows:
\begin{equation} 
	F(\rho_{A},\sigma_{A})\equiv \|\sqrt{\rho_{A}}\sqrt{\sigma_{A}}\|_{1}^{2}. \label{generalizedFidelity} 
\end{equation} Uhlmann's theorem is the statement that the following equality holds \cite{uhlmann}:
\begin{equation}
F(\rho_{A},\sigma_{A})=\sup \limits_{U_{R}}\left|\langle\phi^{\rho}|_{RA}U_{R}\otimes I_{A}|\phi^{\sigma}\rangle_{RA} \right|^{2},
\end{equation}
where $|\phi^{\rho}\rangle_{RA}$ and $|\phi^{\sigma}\rangle_{RA}$ are purifications of $\rho_{A}$ and $\sigma_{A}$ with purifying system $R$ and $U_{R}$ is a unitary acting on system $R$.

\subsection{Entropy and Information} 

In order to study the information contained and transmitted in various systems and operations, we now recall a number of common measures used to quantify information. With these measures  defined below, we also focus on generalizations of the quantities as functions of operators acting on infinite-dimensional, separable Hilbert spaces, as considered in, e.g., \cite{maks}. The first and most common measure is the quantum entropy and is defined for a state $\rho\in\mathcal{D}(\mathcal{H})$ as%
\begin{equation}
H(\rho)\equiv\operatorname{Tr}\{\eta(\rho)\}, \label{eq:vonneumanne}
\end{equation}
where $\eta(x)=-x\log_{2}x$ if $x>0$ and $\eta(0)=0$. The trace in the above
equation can be taken with respect to any countable orthonormal basis of
$\mathcal{H}$ \cite[Definition~2]{AL70}. The quantum entropy is a
non-negative, concave, lower semicontinuous function on $\mathcal{D}%
(\mathcal{H})$ \cite{W76}. It is also not necessarily finite (see, e.g.,
\cite{BV13}). When $\rho_{A}$ is the state of a system $A$, we write
\begin{equation}
H(A)_{\rho}\equiv H(\rho_{A}).
\end{equation}
The entropy is a familiar thermodynamic quantity and is roughly a measure of the disorder in a system.
One property of quantum entropy that we use here is its duality: for a pure state $|\psi\rangle\langle \psi |_{RA}$, quantum entropy is such that $H(A)_{\psi}=H(R)_{\psi}$.

For a positive semi-definite, trace-class operator $\omega$ such that $\Tr\{\omega\} \neq 0$, we extend the definition of quantum entropy as 
\begin{equation}
H(\omega)\equiv \Tr\{\omega\} H\!\left(\frac{\omega}{\Tr\{\omega\}}\right). \label{eq:exvonneumanne}
\end{equation}
Observe that $H(\omega)$ reduces to the definition in \eqref{eq:vonneumanne} when $\omega$ is a state with
$\Tr\{\omega\}=1$.


The quantum relative entropy $D(\rho\Vert\sigma)$ of $\rho,\sigma
\in\mathcal{D}(\mathcal{H})$ is defined as \cite{F70,Lindblad1973}%
\begin{align}
&  D(\rho\Vert\sigma)\nonumber\\
&  \equiv\lbrack\ln2]^{-1}\sum_{i,j}|\langle\phi_{i}|\psi_{j}\rangle
|^{2}[p(i)\ln\!\left(  \frac{p(i)}{q(j)}\right)
+q(j)-p(i)],\label{eq:rel-ent-sep}%
\end{align}
where $\rho=\sum_{i}p(i)|\phi_{i}\rangle\langle\phi_{i}|$ and $\sigma=\sum
_{j}q(j)|\psi_{j}\rangle\langle\psi_{j}|$ are spectral decompositions of
$\rho$ and $\sigma$ with $\{|\phi_{i}\rangle\}_{i}$ and $\{|\psi_{j}%
\rangle\}_{j}$ orthonormal bases. The prefactor $[\ln2]^{-1}$ is there to
ensure that the units of the quantum relative entropy are bits. We take the
convention in \eqref{eq:rel-ent-sep} that $0\ln0=0\ln\!\left(  \frac{0}%
{0}\right)  =0$ but $\ln\!\left(  \frac{c}%
{0}\right) = +\infty$ for $c>0$. Each term in the sum in \eqref{eq:rel-ent-sep} is
non-negative due to the inequality
\begin{equation}
x\ln(x/y)+y-x\geq0
\end{equation}
holding for all $x,y\geq0$ \cite{F70}. Thus, by Tonelli's theorem, the sums in \eqref{eq:rel-ent-sep}
may be taken in either order as discussed in \cite{F70,Lindblad1973}, and it
follows that
\begin{equation}
D(\rho\Vert\sigma)\geq0 
\end{equation}
 for all $\rho,\sigma\in\mathcal{D}%
(\mathcal{H})$, with equality holding if and only if $\rho=\sigma$ \cite{F70}.
If the support of $\rho$ is not contained in the support of $\sigma$, then
$D(\rho\Vert\sigma)=+\infty$. The converse statement need not hold in general:
there exist $\rho,\sigma\in\mathcal{D}(\mathcal{H})$ with the support of
$\rho$ contained in the support of $\sigma$ such that $D(\rho\Vert
\sigma)=+\infty$.
Thus, for states $\rho$ and $\sigma$, we have that
\begin{equation}
D(\rho\Vert \sigma) \in [0,\infty].\label{eq:rel-ent-non-neg}
\end{equation}
It is also worth noting that relative entropy is not generally symmetric; i.e., there exist states $\rho$ and $\sigma$ for which
\begin{equation}
D(\rho\|\sigma)\neq D(\sigma\|\rho). \label{eq:notcom}
\end{equation}

One of the most important properties of the quantum relative entropy
$D(\rho\Vert\sigma)$ is that it is monotone with respect to a quantum channel
$\mathcal{N}:\mathcal{L}_1(\mathcal{H}_{A})\rightarrow\mathcal{L}_1(\mathcal{H}%
_{B})$ \cite{Lindblad1975}:%
\begin{equation}
D(\rho\Vert\sigma)\geq D(\mathcal{N}(\rho)\Vert\mathcal{N}(\sigma)).
\label{eq:mono-rel-ent}%
\end{equation}
The above inequality is often called the ``data processing inequality.''
This inequality implies that the quantum relative entropy is invariant under the action of an isometry $U$:
\begin{equation}
D(\rho\Vert\sigma)= D(U\rho U^\dag\Vert U \sigma U^\dag). \label{eq:iso-inv-rel-ent}
\end{equation}



The quantum mutual information of a bipartite state $\rho_{AB}$ is defined in terms of the relative entropy \cite{Lindblad1973} as 
\begin{equation}
I(A;B)_{\rho}\equiv D(\rho_{AB}\|\rho_{A}\otimes\rho_{B})  .
\label{eq:exmutuali}
\end{equation} 
Note that, with the definition in \eqref{eq:exmutuali}, we have that
\begin{equation}
I(A;B)_{\rho} \in [0,\infty]
\end{equation}
 as a consequence of \eqref{eq:rel-ent-non-neg}.
 The following inequality applies to quantum mutual information \cite{Lindblad1973}:
\begin{equation}
I(A;B)_{\rho}\leq 2\min\{H(A)_{\rho},H(B)_{\rho}\} \label{mutualInfoBound}
\end{equation}
and establishes that it is finite if one of the marginal entropies is finite.
For a general positive semi-definite  trace-class operator $\omega_{AB}$
such that $\Tr\{\omega_{AB}\} \neq 0$, we extend the definition of mutual information as in \cite{S15}
\begin{equation}
I(A;B)_\omega \equiv \operatorname{Tr}\{\omega\} I(A;B)_{\frac{\omega}{\operatorname{Tr}\{\omega\} }  }.
\label{eq:extended-MI}
\end{equation}
Note that, while the relative entropy is not generally symmetric,  mutual information is symmetric under the exchange of systems $A$ and $B$
\begin{equation} I(A;B)_{\rho}=I(B;A)_{\rho} \label{symmetricMutualInformation} ,\end{equation}
due to \eqref{eq:iso-inv-rel-ent} and by taking the isometry therein to be a unitary swap of the systems $A$ and $B$.  
For a state $\rho_{AB}$ such that the entropies $H(A)_{\rho}$ and $H(B)_{\rho}$ are finite, the mutual information reduces to
\begin{equation} I(A;B)_{\rho}=H(A)_{\rho}+H(B)_{\rho}-H(AB)_{\rho}. \label{eq:finitemutualinfo} \end{equation}

For a state $\rho_{AB}$ such that $H(A)_{\rho}<\infty$,  the conditional entropy is defined as
\cite{K11}
\begin{equation}
H(A|B)_{\rho}
 \equiv H(A)_{\rho}-I(A;B)_{\rho}, \label{eq:conditionale}
\end{equation}
and the same definition applies for 
a positive semi-definite  trace-class operator $\omega_{AB}$, by employing the extended definitions of entropy in \eqref{eq:exvonneumanne} and mutual information in \eqref{eq:extended-MI}. 
Thus, as a consequence of the definition and \eqref{mutualInfoBound}, we have that
\begin{equation}
H(A|B)_{\rho} \in [-H(A)_{\rho}, H(A)_{\rho}] \ .\label{welldefinedConditionale} 
\end{equation}
If $H(B)_{\rho}$ is also finite, then the conditional entropy simplifies to the following more familiar form:\begin{equation} H(A|B)_{\rho}=H(AB)_{\rho}-H(B)_{\rho}. \label{conditionalefamiliarform} \end{equation}
For a tripartite pure state $\psi_{ABC}$ such that $H(A)_{\psi}<\infty$, the conditional entropy satisfies the following duality relation \cite{K11}: \begin{equation} H(A|B)_{\psi}=-H(A|C)_{\psi} .\label{conditionalEntropyDuality} \end{equation} 
\cite[Proposition~1]{K11} states that conditional entropy is subadditive: for a four-party state $\rho_{ABCD}$, we have that
\begin{equation} H(AB|CD)_{\rho}\leq H(A|C)_{\rho} + H(B|D)_{\rho}. \label{subadditivityOfConditionalEntropy} \end{equation}
This in turn is a consequence of the strong subadditivity of quantum entropy \cite{LR73,PhysRevLett.30.434}.

The conditional quantum mutual information (CQMI) of  tripartite states $\omega_{ABE}\in\mathcal{D}(\mathcal{H}_{ABE})$, with $\mathcal{H}_{ABE}$  a separable Hilbert space, was defined only recently in  \cite{S15}, as a generalization of the information measure commonly used in the finite-dimensional setting. The definition from \cite{S15} involves taking a supremum over all finite-rank projections $P_{A}\in \mathcal{L}(\mathcal{H}_{A})$ or
$P_{B}\in \mathcal{L}(\mathcal{H}_{B})$, in order to write CQMI  in terms of the quantum mutual information in the following equivalent ways:
\begin{align}
& I(A;B|E)_{\omega} \nonumber \\
& =\sup_{P_{A}}I(A;BE)_{Q_{A}\omega Q_{A}}-I(A;E)_{Q_{A}\omega Q_{A}}  \label{exconmutualiA} \\
& =\sup_{P_{B}}I(AE;B)_{Q_{B}\omega Q_{B}}-I(E;B)_{Q_{B}\omega Q_{B}}  ,\label{exconmutualiB}
\end{align}
where $Q_{A}=P_{A}\otimes I_{BE}$ and $Q_{B}=P_{B}\otimes I_{AE}$. Due to the data-processing inequality in \eqref{eq:mono-rel-ent}, with the channel taken to be a partial trace, we have that
\begin{equation}
I(A;B|E)_{\omega} \in [0,\infty].
\end{equation}
The conditional mutual information, as defined above, is a lower semi-continuous function of tripartite quantum states \cite[Theorem~2]{S15}; i.e., for any sequence $\{ \omega^n_{ABE}\}_n$ of tripartite states converging to the state $\omega^0_{ABE}$, the following inequality holds
\begin{equation}
\liminf_{n\to \infty} I(A;B|E)_{\omega^n} \geq I(A;B|E)_{\omega^0}.
\end{equation}
If $I(A;BE)_{\omega}, I(A;E)_{\omega} < \infty$, as is the case if $H(A)_\omega < \infty$, then the definition reduces to the familiar one from the finite-dimensional case:
\begin{equation}
I(A;B|E)_{\omega} = I(A;BE)_{\omega}-  I(A;E)_{\omega}.
\end{equation}


\subsection{Squashed Entanglement}

The information measure of most concern in our paper is the squashed entanglement. Defined and analyzed in \cite{christ}, and extended to the infinite-dimensional case in \cite{maks}, the squashed entanglement of a state $\rho_{AB} \in\mathcal{D}(\mathcal{H}_{AB})$ is defined as
\begin{equation}
E_{\sq}(A;B)_{\rho} = \frac{1}{2}
\inf_{\omega_{ABE}}I(A;B|E)_{\omega} ,\label{eq:squashede}
\end{equation}
where $\omega_{ABE} \in \mathcal{D}(\mathcal{H}_{ABE})$ satisfies $\Tr_E\{\omega_{ABE}\} = \rho_{AB}$, with $\mathcal{H}_E$ taken to be an infinite-dimensional, separable Hilbert space.
(See \cite{Tucci2000,Tucci2002} for discussions related to squashed entanglement.)
An equivalent definition is given in terms of an optimization over squashing channels, as follows:
\begin{equation}
E_{\sq}(A;B)_{\rho} = \frac{1}{2}\inf\limits_{\mathcal{S}_{E\to E'}}I(A;B|E')_{\tau} \label{eq:squashed-squash-chnl} ,
\end{equation}
where $\tau_{ABE'}=\mathcal{S}_{E\to E'}(\phi_{ABE}^{\rho})$, with $\phi_{ABE}^{\rho}$ a purification of $\rho_{AB}$. The infimum is with respect to all squashing channels $\mathcal{S}_{E\to E'}$ from system $E$ to a system $E'$ corresponding to an infinite-dimensional, separable Hilbert space. The reasoning for this equivalence is the same as that given in \cite{christ}. Due to the expression in \eqref{eq:squashed-squash-chnl}, squashed entanglement can be interpreted as the leftover correlation after an adversary attempts to ``squash down'' the correlations in  $\rho_{AB}$.
Squashed entanglement obeys many of the properties considered important for an entanglement measure, such as LOCC monotonicity, additivity for product states, and convexity \cite{christ}. These properties are discussed in the next section.

Suppose that Alice, in possession of the systems $RA$ of a pure state $\phi_{RA}$, wishes to construct a shared state with Bob. If Alice and Bob are connected by a quantum channel $\mathcal{N}_{A\to B}$ mapping system $A$ to $B$, then they can  establish the shared state
\begin{equation}
\omega_{RB}=\mathcal{N}_{A\to B}(\phi_{RA}). \label{eq:omega-1}
\end{equation}
Going to the purified picture, an isometric channel $\mathcal{U}^{\mathcal{N}}_{A\to BE}$ extends $\mathcal{N}_{A\to B}$, so that the output state of the extended channel is
$\phi_{RBE}=\mathcal{U}^{\mathcal{N}}_{A\to BE}(\phi_{RA})$ when the input is
$\phi_{RA}$. Suppose that a third party Eve has access to the system $E$, such that she could then perform a squashing channel $\mathcal{S}_{E\to E^{\prime}}$, bringing system $E$ to system $E'$. In this way, she could attempt to thwart the correlation between Alice and Bob's systems, as measured by conditional mutual information.  Related to the above physical picture, the squashed entanglement of the channel
$\mathcal{N}_{A\to B}$ is defined
as the largest possible squashed entanglement that can be realized between systems $R$ and $B$ \cite{tgwB,tgw}: \begin{equation}
E_{\sq}(\mathcal{N})\equiv\sup\limits_{\phi_{RA}}E_{\sq}(R;B)_{\omega} \label{eq:squashedeN},
\end{equation}
where the supremum is with respect to all possible pure bipartite input states $\phi_{RA}$, with system $R$ isomorphic to system $A$, and $\omega_{RB}$ is defined in \eqref{eq:omega-1}.

If specific requirements are placed on the channel input states, such as an energy constraint as discussed in Section~\ref{ss:LOCCcapacityWithAvgEnergyConstraint} below, the optimization should reflect those stipulations, leading to the energy-constrained squashed entanglement of a channel $\mathcal{N}$:
\begin{equation} E_{\sq}(\mathcal{N},G,P)\equiv\sup\limits_{\phi_{RA}:\Tr\{G\phi_{A}\}\leq P}E_{\sq}(R;B)_{\omega}. \label{eq:constrainedSquashedeN}
\end{equation}
Here $G$ is an energy observable acting on the channel input system $A$, the positive real $P\in [0,\infty)$ is a constraint on the expected value of that observable such that $\Tr\{G\phi_{A}\}\leq P$, and the supremum is with respect to all pure input states $\phi_{RA}$ to the channel that obey the given constraint. It suffices to optimize the quantity in \eqref{eq:constrainedSquashedeN} with respect to pure, bipartite input states, following from purification, the Schmidt decomposition theorem, and LOCC monotonicity of squashed entanglement. These notions are discussed in more detail in Section~\ref{s:ECSKAC}.

As discussed in \cite{tgwB}, the squashed entanglement of a channel can be written in a different way by considering an isometric channel
$\mathcal{V}^{\mathcal{S}}_{E\to E'F}$
extending the squashing channel
$\mathcal{S}_{E\to E'}$. Let $\varphi_{RBE'F}$ denote the following pure output state when the pure state $\phi_{RA}$ is input:
\begin{equation}
\varphi_{RBE'F}=\big(\mathcal{V}^{\mathcal{S}}_{E\to E'F}\circ\mathcal{U}^{\mathcal{N}}_{A\to BE}\big)(\phi_{RA}) \label{extendedSquashingChannel} .
\end{equation}
By taking advantage of the duality of conditional entropy and in the case that the entropy $H(B)_\varphi$ is finite, the alternate way of writing follows from
the equality
\begin{align}
I(R;B|E')_{\varphi}
&=H(B|E')_{\varphi}-H(B|RE')_{\varphi}\\&=H(B|E')_{\varphi}+H(B|F)_{\varphi}  \label{altEsq} .
\end{align}
Thus, we can write the energy-constrained squashed entanglement of a channel as
\begin{equation}
 E_{\sq}(\mathcal{N},G,P) 
=
\sup\limits_{\rho_A :
\Tr\{G\rho_{A}\}\leq P
} E_{\sq}(\rho_A, \mathcal{N}_{A \to B}),
\label{altEsqN} 
\end{equation} 
where
\begin{align}
E_{\sq}(\rho_A, \mathcal{N}_{A \to B})
& \equiv 
\inf_{\mathcal{V}^{\mathcal{S}}_{E\to E'F}}
\frac{1}{2}[
H(B|E')_{\omega}+H(B|F)_{\omega}]\\
\omega_{BE'F} & = \big(\mathcal{V}^{\mathcal{S}}_{E\to E'F}\circ\mathcal{U}^{\mathcal{N}}_{A\to BE}\big)(\rho_{A}),
\end{align}
and we take advantage of the representation in \eqref{altEsqN} in our paper.

\subsection{Entanglement Monotones and Squashed Entanglement}

In this section, we review the notion of an entanglement monotone \cite{H42007} and how  squashed entanglement \cite{christ} and its extended definition in \cite{maks} satisfies the requirements of being an entanglement monotone.
Let $E(A;B)_\omega$ be a function of an arbitrary bipartite state $\omega_{AB}$. Then $E(A;B)_\omega$ is an entanglement monotone if it satisfies the following conditions:

	1) $E(A;B)_\omega=0$ if and only if $\omega_{AB}$ is separable. 
	
	2) $E$ is monotone under selective unilocal operations. That is,
\begin{equation}
E(A;B)_\omega
\geq\sum\limits_{k}
p_{k}E(A;B)_{\omega^k} ,
\end{equation}
where
\begin{equation}
p_{k}  =\Tr(\mathcal{N}^k_{A}(\omega_{AB})),\qquad 
\omega_{AB}^{k}  =p_{k}^{-1}\mathcal{N}^k_{A}(\omega_{AB})
\end{equation}
for any state $\omega_{AB}$ and any collection $\{\mathcal{N}^k_{A}\}$ of unilocal completely positive maps such that the sum map $\sum\limits_{k} \mathcal{N}^k_{A}$ is a channel. 

	3) $E$ is convex, in the sense that for states $\rho^0_{AB}$, $\rho^1_{AB}$, and $\rho^\lambda_{AB}
	= (1-\lambda) \rho^0_{AB} + \lambda \rho^1_{AB}$, where $\lambda \in [0,1]$, 
\begin{equation}	
	E(A;B)_{\rho^\lambda}\leq(1-\lambda) E(A;B)_{\rho^0}+\lambda E(A;B)_{\rho^1} .
\end{equation}

When the condition in 3) holds, then the condition in 2) is equivalent to monotonicity under LOCC. 

\bigskip
	 An entanglement monotone is additionally considered an entanglement measure if, for any pure state $\psi_{AB}$, it is equal to the quantum entropy of a marginal state:
\begin{equation}
E(A;B)_{\psi}
=H(A)_{\psi} = H(B)_{\psi}. \label{eq:marginal-von-neumann}
\end{equation}

	Other desirable properties for an entanglement monotone include
\begin{itemize}	
	 \item additivity for a product state $\omega_{AB}\otimes\theta_{A'B'}$:
	 \begin{equation}
	 E(AA';BB')_{\omega \otimes\theta} =E(A;B)_{\omega}+E(A';B')_{\theta} ,
	 \label{item:additivity}	
	 \end{equation}
	 \item subadditivity for a product state $\omega_{AB}\otimes\theta_{A'B'}$:
	 \begin{equation}
	 E(AA';BB')_{\omega \otimes\theta} \leq E(A;B)_{\omega}+E(A';B')_{\theta} ,
\label{item:subadditivity}	
	 \end{equation}
	 \item strong superadditivity for a state $\omega_{AA'BB'}$:
	 \begin{equation}
	 E(AA';BB')_{\omega}
	 \geq E(A;B)_{\omega}
	 +E(A';B')_{\omega} ,
	 \label{item:strongSuperadditivity}
	 \end{equation}
	 \item monogamy for a state $\omega_{ABC}$:
	 \begin{equation}
	 E(A;BC)_{\omega}
	 \geq E(A;B)_{\omega}
	 +E(A;C)_{\omega} ,\label{item:monogamy}	
	 \end{equation}
	\item asymptotic continuity:
	\begin{equation}
	\lim\limits_{n \to \infty} \frac{E(\rho_{AB}^{n})-E(\sigma_{AB}^{n})}{1+\log_2(\dim\mathcal{H}_{AB}^{n})}=0 , \label{item:asymptoticContinuity}
\end{equation}
which should hold for any sequences $\{ \rho_{AB}^{n}\}_n$ and $\{ \sigma_{AB}^{n}\}_n$ of states such that $\Vert \rho_{AB}^{n} - \sigma_{AB}^{n} \Vert_1$ converges to zero as $n \to \infty$.
	\end{itemize}
As discussed in \cite{maks}, for states in infinite-dimensional Hilbert spaces, global asymptotic continuity is too restrictive. For example, the discontinuity of the quantum entropy means that any entanglement monotone that possesses property \eqref{eq:marginal-von-neumann} is necessarily discontinuous. It is therefore reasonable to require   instead that $E$ be lower semi-continuous \cite{maks}:
\begin{equation}\liminf\limits_{n\to\infty} E(\omega_{AB}^{n})\geq E(\omega_{AB}^{0})
\end{equation}
 for any sequence $\{\omega_{AB}^{n}\}$ of states converging to the state~$\omega_{AB}^{0}$.


The squashed entanglement, as defined in \eqref{eq:squashede}, obeys all of the above properties \cite{christ,AF04,KWin04,C06,BCY11,LW14,maks}.
Regarding the last property, the squashed entanglement defined in \eqref{eq:squashede} has been proved to be lower semicontinuous on the set of states having at least one finite marginal entropy \cite{maks}.
It additionally satisfies the following uniform continuity inequality: Given states $\rho_{AB}$ and $\sigma_{AB}$ satisfying  $\frac{1}{2}\left\|\rho_{AB}-\sigma_{AB}\right\|_{1}\leq\varepsilon$ for $\varepsilon \in [0,1]$ then
\begin{multline}
|E_{\sq}(A;B)_{\rho}-E_{\sq}(A;B)_{\sigma}|\\ \leq \sqrt{2\varepsilon} \log_2\min[\dim(\mathcal{H}_{A}),\dim(\mathcal{H}_{B})]
+g(\sqrt{2\varepsilon})
\label{EsqContinuity}
\end{multline}
where
\begin{equation}
g(x)\equiv (1+x)\log_2(1+x)-x\log_2(x).
\label{eq:g-func}
\end{equation}
This follows by combining the well known Fuchs--van de Graaf inequalities \cite{FG98}, Uhlmann's theorem for fidelity \cite{uhlmann}, and the continuity bound from \cite[Corollary~1]{Shirokov2016} for conditional mutual information.

\subsection{Private States}

The main goal of any key distillation protocol is for two parties Alice and Bob to distill a tripartite state as close as possible to an ideal tripartite secret-key state, which is protected against a third-party Eve.
An ideal tripartite secret-key state
$\gamma_{ABE}$
is such that local projective measurements $\mathcal{M}_A$ and
$\mathcal{M}_B$ on it,
in the respective orthonormal bases $\{|i\rangle_{A}\}_{i}$ and $\{|i\rangle_{B}\}_{i}$, 
lead to the following form:
\begin{equation}
(\mathcal{M}_A \otimes \mathcal{M}_B)
(\gamma_{ABE})=\frac{1}{K}\sum\limits_{i=1}^{K}|i\rangle\langle i|_{A}\otimes|i\rangle\langle i|_{B}\otimes\sigma_{E}.
\label{secretKeyState}
\end{equation}
The key systems are finite-dimensional, but the  eavesdropper's system $E$ could be described by an infinite-dimensional, separable Hilbert space. The tripartite key state $\gamma_{ABE}$  contains $\log_2 K$ bits of secret key. By inspecting the right-hand side of \eqref{secretKeyState}, we see that the key value is uniformly random and perfectly correlated between systems $A$ and $B$, as well as being in tensor product with the state of system $E$, implying that the results of any experiment on the $AB$ systems will be independent of those given by an experiment conducted on the $E$ system. While a perfect ideal tripartite key state may be difficult to achieve in practice, a state that is nearly indistinguishable from the ideal case is good enough for practical purposes. If a state $\rho_{ABE}$ satisfies the following inequality:
\begin{equation}
F(\gamma_{ABE},\rho_{ABE})\geq 1-\varepsilon,
\end{equation}
for some $\varepsilon \in [0,1]$ and $\gamma_{ABE}$ an ideal tripartite key state,
then $\rho_{ABE}$ is called an $\varepsilon$-approximate tripartite key state \cite{HHHO05,HHHO09,WTB17}.

By purifying a tripartite secret-key state $\gamma_{ABE}$ with ``shield systems'' $A'$ and $B'$ and then tracing over the system $E$, the resulting state is called  a bipartite  private state, which takes the following form \cite{HHHO05,HHHO09}:
\begin{equation}
\gamma_{ABA'B'}=U_{ABA'B'}(|\Phi\rangle\langle\Phi|_{AB}\otimes\sigma_{A'B'})U^{\dagger}_{ABA'B'}, \label{privateState}
\end{equation}
where
\begin{equation}
|\Phi\rangle_{AB}=
\frac{1}{\sqrt{K}}
\sum_{i=1}^K |i\rangle_{A}|i\rangle_{B}
\end{equation}
is a maximally entangled state with Schmidt rank $K$ and $\sigma_{A'B'}$ is an arbitrary state of the shield systems $A'B'$. Due to the fact that the system $E$ of the tripartite key state $\gamma_{ABE}$ corresponds generally to an infinite-dimensional, separable Hilbert space, the same is true for the shield systems $A'B'$ of $\gamma_{ABA'B'}$. The unitary operator $U_{ABA^{\prime}B^{\prime}}$ is called a ``twisting" unitary and has the following form:
\begin{equation}
U_{ABA^{\prime}B^{\prime}}=\sum_{i,j=1}^{K} |i\rangle\langle i|_{A}\otimes |j\rangle\langle j|_{B}\otimes U^{ij}_{A^{\prime}B^{\prime}}, \label{twistingUnitary}
\end{equation}
where each $U^{ij}_{A^{\prime}B^{\prime}}$ is a unitary operator.
Note that, due to the correlation between the $A$ and $B$ systems in the state $\Phi_{AB}$, only the diagonal terms $U^{ii}_{A'B'}$ of the twisting unitary are relevant when measuring the systems $A$ and $B$ in the orthonormal bases $\{|i\rangle_{A}\}_{i}$ and $\{|i\rangle_{B}\}_{i}$, respectively \cite{HHHO05,HHHO09}. If a state $\rho_{ABA'B'}$ satisfies
\begin{equation}
F(\gamma_{ABA'B'},\rho_{ABA'B'})\geq 1-\varepsilon,
\end{equation}
for some $\varepsilon \in [0,1]$ and $\gamma_{ABA'B'}$ an ideal bipartite private state,
then $\rho_{ABA'B'}$ is called an $\varepsilon$-approximate bipartite private state \cite{HHHO05,HHHO09,WTB17}.

The converse of the above statement holds as well \cite{HHHO05,HHHO09}, and the fact that it does is one of the main reasons that the above notions are useful in applications. That is, given a bipartite private state of the form in \eqref{privateState}, we can then purify it by an $E$ system, and tracing over the shield systems $A'B'$ leads to a tripartite key state of the form in \eqref{secretKeyState}. These relations extend to the approximate case as well, by an application of Uhlmann's theorem for fidelity \cite{uhlmann}: purifying an $\varepsilon$-approximate tripartite key state $\rho_{ABE}$ with shield systems $A'B'$ and tracing over system $E$ leads to an $\varepsilon$-approximate bipartite private state, and vice versa.

The squashed entanglement of a bipartite private state of $\log_2 K$ bits is normalized such that \cite{C06} \begin{equation}  E_{\sq}(AA';BB')_{\gamma} \geq \log_2{K}. \label{privateStateEsqNormalized} \end{equation}
This result has recently been extended to the approximate case: \cite[Theorem~2]{appxpriv}  establishes that, for an $\varepsilon$-approximate bipartite private state $\rho_{ABA'B'}$, the following inequality holds 
\begin{equation}
 E_{\sq}(AA^{\prime};BB^{\prime})_{\rho}+2\sqrt{\varepsilon} \log_2 K +2g(\sqrt{\varepsilon}) 
 \geq \log_2{K}.
 \label{eq:appxpriv}
 \end{equation} 

\section{Properties of Conditional Quantum Mutual Information}
\label{s:CQMIinInfDim}

In this section, we establish a number of simple properties of conditional quantum mutual information (CQMI) for states of infinite-dimensional, separable Hilbert spaces. These properties will be useful in later sections of our paper.

\subsection{CQMI and Duality under a Finite-Entropy Assumption}

\begin{lemma}[Duality] \label{lem:dualitylemma}
	Let $\psi_{ABED}$ be a pure state such that $H(B)_{\psi}<\infty$. Then the conditional quantum mutual information 
	$I(A;B|E)_{\psi}$ 
	can be written as
	\begin{equation}
		I(A;B|E)_{\psi}=H(B|E)_{\psi}+H(B|D)_{\psi}. \label{eq:dualitylemma}
	\end{equation}
\end{lemma}
\begin{proof}
	Begin with the definition of CQMI from~\eqref{exconmutualiB}:
	\begin{align}
		I(A;B|E)_{\psi}=\sup_{P_{B}}\big[I(B;AE)&_{Q_{B}\psi Q_{B}}-I(B;E)_{Q_{B}\psi Q_{B}} \notag \\
		& : Q_{B}=P_{B}\otimes I_{AE}\big], \label{proofstep1}
	\end{align}
	where  we have exploited the symmetry of mutual information as recalled in \eqref{symmetricMutualInformation}.
	The assumption $H(B)_{\psi}<\infty$ is strong, implying that $I(B;AE)_{\psi}, I(B;E)_{\psi}< \infty$, so that we can write $I(A;B|E)_{\psi} = I(B;AE)_{\psi}-I(B;E)_{\psi}$ \cite{S15}. Then we find that
	\begin{align}
		& I(A;B|E)_{\psi} \nonumber \\
		& =H(B)_{\psi}-H(B)_{\psi}+I(B;AE)_{\psi}-I(B;E)_{\psi}\notag \\ 
		& =[H(B)_{\psi}-I(B;E)_{\psi}]-[H(B)_{\psi}-I(B;AE)_{\psi}] .\label{proofstep23}
	\end{align}
	From the definition in \eqref{eq:conditionale}, it is clear that the last line is equal to  a difference of conditional entropies, leading to
	\begin{equation}
		I(A;B|E)_{\psi}=H(B|E)_{\psi}-H(B|AE)_{\psi} .\label{CQMIisConditionalEntropies}
	\end{equation}
	Finally, we invoke the duality of conditional entropy from \eqref{conditionalEntropyDuality}
 in order to arrive at the statement of the lemma.
\end{proof}

\subsection{Subadditivity Lemma for Conditional Quantum Mutual Information}

In this section, we prove a lemma that generalizes one of the main technical results of \cite{tgwB,tgw} to the infinite-dimensional setting of interest here. This lemma was the main tool used in \cite{tgwB,tgw} to prove that the squashed entanglement of a quantum channel is an upper bound on its secret-key-agreement capacity. After \cite{tgwB,tgw} appeared, this lemma
was later interpreted as implying that amortization does not increase the squashed entanglement of a channel \cite{KW17a,RKBKMA17,BW17}.

\begin{lemma}
\label{lemma:subadditivity-CQMI}
Let $\phi_{A^{\prime}A B^{\prime}E^{\prime\prime}F^{\prime\prime}}$ be a pure state, and let $\mathcal{U}_{A \to B E' F'}$ be an isometric quantum channel. Set
\begin{equation}
\psi_{A^{\prime}BB^{\prime}E^{\prime}E^{\prime\prime}F^{\prime}F^{\prime\prime}} \equiv
\mathcal{U}_{A \to B E' F'} (\phi_{A^{\prime}A B^{\prime}E^{\prime\prime}F^{\prime\prime}}),
\end{equation}
and suppose that $H(B)_{\psi} < \infty$.
Then the following inequality holds
	\begin{multline}
		I(A^{\prime};BB^{\prime}|E^{\prime}E^{\prime\prime})_{\psi}\leq H(B|E^{\prime})_{\psi}+H(B|F^{\prime})_{\psi}\\
		+I(A^{\prime}A;B^{\prime}|E^{\prime\prime})_{\phi} \label{CQMIbound} .
	\end{multline}
Note that both sides of the inequality in \eqref{CQMIbound} could be equal to $+\infty$.
\end{lemma}
\begin{proof}
Let $\{P_{B^{\prime}}^{k}\}_{k}$ be a sequence of finite-rank projectors acting on the space $\mathcal{H}_{B^{\prime}}$, which strongly converges to the identity $I_{B^{\prime}}$. Define the sequence 
$\{\phi^{k}_{A^{\prime} A \overline{B^{\prime}_{k}}E^{\prime\prime}F^{\prime\prime}}\}_k$
of projected states as 
\begin{multline}
	\phi^{k}_{A^{\prime} A \overline{B^{\prime}_{k}}E^{\prime\prime}F^{\prime\prime}}=\\
	\lambda^{-1}_{k} [(P_{B^{\prime}}^{k}\otimes \overline{I})\phi_{A^{\prime}AB^{\prime}E^{\prime\prime}F^{\prime\prime}}(P_{B^{\prime}}^{k}\otimes \overline{I})] \label{psiProjection} ,
\end{multline}
where
\begin{align}
\overline{I} & \equiv I_{A^{\prime}}\otimes
I_{A}\otimes
I_{E^{\prime\prime}}\otimes I_{F^{\prime\prime}}, \\
\lambda_{k} & \equiv \Tr\{(P_{B^{\prime}}^{k}\otimes \overline{I})\phi_{A^{\prime} A B^{\prime}E^{\prime\prime}F^{\prime\prime}}\}, \\
\lim\limits_{k\to \infty}\lambda_{k} & = 1.
\end{align}
This then leads to the following sequence of projected states:
\begin{equation}
\psi^k_{A^{\prime}B\overline{B^{\prime}_k} E^{\prime}E^{\prime\prime}F^{\prime}F^{\prime\prime}} \equiv
\mathcal{U}_{A \to B E' F'}
(\phi^{k}_{A^{\prime} A \overline{B^{\prime}_{k}}E^{\prime\prime}F^{\prime\prime}}).
\end{equation}
Note that each state $\psi^{k}_{A^{\prime}B\overline{B^{\prime}_{k}}E^{\prime}E^{\prime\prime}F^{\prime}F^{\prime\prime}}$ is pure for all $k \geq 1$.
Then the conditional entropy and the conditional mutual information of the sequence converge to those of the original state \cite{S15,K11}:\begin{align}
	\lim\limits_{k\to \infty} H(B|E^{\prime})_{\psi^{k}}&=H(B|E^{\prime})_{\psi} \label{convergineCondiitonalEntropy},
	\\
		\lim\limits_{k\to \infty}
	H(B|F^{\prime})_{\psi^{k}} & = 
	H(B|F^{\prime})_{\psi} \label{convergineCondiitonalEntropy-1},
	\\ \lim\limits_{k\to \infty} I(A^{\prime};B\overline{B^{\prime}_{k}}|E^{\prime}E^{\prime\prime})_{\psi^{k}}&=I(A^{\prime};BB^{\prime}|E^{\prime}E^{\prime\prime})_{\psi} , 
	\label{convergineCondiitonalMutualInformation-1}
	\\
\lim\limits_{k\to \infty} I(A^{\prime}A;\overline{B^{\prime}_{k}}|E^{\prime\prime})_{\phi^{k}}
& = I(A^{\prime}A;B^{\prime}|E^{\prime\prime})_{\phi} .
\label{convergineCondiitonalMutualInformation}
\end{align}
The limits in \eqref{convergineCondiitonalEntropy}--\eqref{convergineCondiitonalEntropy-1} follow because
$\lim\limits_{k\to \infty}
	H(B)_{\psi^{k}}  = 
	H(B)_{\psi} < \infty$, by applying \cite[Lemma~2]{S15} and \cite[Proposition~2]{K11}.
The limits in \eqref{convergineCondiitonalMutualInformation-1}--\eqref{convergineCondiitonalMutualInformation} follow, with possible $+\infty$ on the right-hand side,  from the lower semicontinuity of conditional quantum mutual information 
and its monotonicity under local operations \cite[Theorem~2]{S15}.

Due to the fact that
$H(B\overline{B'_k})_{\psi^{k}} < \infty$ for all $k\geq 1$, 
we can  write the CQMI of the state $\psi^{k}_{A^{\prime}B\overline{B^{\prime}_{k}}E^{\prime}E^{\prime\prime}F^{\prime}F^{\prime\prime}}$ in terms of conditional entropies as in \eqref{CQMIisConditionalEntropies} and then use the duality of conditional entropy as in \eqref{conditionalEntropyDuality} to find that
	\begin{align}	
	& I(A^{\prime};B\overline{B^{\prime}_{k}}|E^{\prime}E^{\prime\prime})_{\psi^{k}}
	\nonumber \\
&=H(B\overline{B^{\prime}_{k}}|E^{\prime}E^{\prime\prime})_{\psi^{k}}-H(B\overline{B^{\prime}_{k}}|A^{\prime}E^{\prime}E^{\prime\prime})_{\psi^{k}}\\ &=H(B\overline{B^{\prime}_{k}}|E^{\prime}E^{\prime\prime})_{\psi^{k}}+H(B\overline{B^{\prime}_{k}}|F^{\prime}F^{\prime\prime})_{\psi^{k}}.  \label{dualityProjected}
	\end{align}
We then employ the subadditivity of conditional entropy from  \eqref{subadditivityOfConditionalEntropy} to split up each of these two terms and regroup the resulting terms:
	\begin{align} 
& H(B\overline{B^{\prime}_{k}}|E^{\prime}E^{\prime\prime})_{\psi^{k}}+H(B\overline{B^{\prime}_{k}}|F^{\prime}F^{\prime\prime})_{\psi^{k}} \nonumber 
\\
&\leq H(B|E^{\prime})_{\psi^{k}}+H(\overline{B^{\prime}_{k}}|E^{\prime\prime})_{\psi^{k}}+H(B|F^{\prime})_{\psi^{k}}
\nonumber
\\
& \qquad +H(\overline{B^{\prime}_{k}}|F^{\prime\prime})_{\psi^{k}}
\\&= H(B|E^{\prime})_{\psi^{k}}+H(B|F^{\prime})_{\psi^{k}}  \nonumber\\
& \qquad + H(\overline{B^{\prime}_{k}}|E^{\prime\prime})_{\psi^{k}}+H(\overline{B^{\prime}_{k}}|F^{\prime\prime})_{\psi^{k}}. \label{subadditivitySplit} 
	\end{align}
	This is then recognizable as two conditional entropies from after the channel use added to the conditional mutual information from before the channel use:
	\begin{multline}
I(A^{\prime};B\overline{B^{\prime}_{k}}|E^{\prime}E^{\prime\prime})_{\psi^{k}} 
		\leq H(B|E^{\prime})_{\psi^{k}}+H(B|F^{\prime})_{\psi^{k}}\\
		+I(A^{\prime}A;\overline{B^{\prime}_{k}}|E^{\prime\prime})_{\phi^{k}} \label{regroupedCQMI}
	\end{multline}
	Taking the limit $k\to\infty$ of this expression and applying \eqref{convergineCondiitonalEntropy}--\eqref{convergineCondiitonalMutualInformation} gives the inequality stated in \eqref{CQMIbound}:
	\begin{align} 
		& I(A^{\prime};BB^{\prime}|E^{\prime}E^{\prime\prime})_{\psi} \nonumber \\
	& 	=\lim\limits_{k\to \infty}I(A^{\prime};B\overline{B^{\prime}_{k}}|E^{\prime}E^{\prime\prime})_{\psi^{k}}
	\nonumber \\
	&\leq \lim\limits_{k\to \infty}\big[H(B|E^{\prime})_{\psi^{k}}+H(B|F^{\prime})_{\psi^{k}}+I(A^{\prime}A;\overline{B^{\prime}_{k}}|E^{\prime\prime})_{\phi^{k}}\big]
	\nonumber \\
	&=H(B|E^{\prime})_{\psi}+H(B|F^{\prime})_{\psi}+I(A^{\prime}A;B^{\prime}|E^{\prime\prime})_{\phi} \label{limits4lemmas} .
	\end{align}
This concludes the proof.	
\end{proof}

\section{Energy-Constrained
Secret-Key-Agreement Capacity}

\label{s:ECSKAC}

We now outline a protocol for energy-constrained secret key agreement between two parties Alice and Bob. The resources available to Alice and Bob in such a protocol are $n$ uses of a quantum channel $\mathcal{N}$ interleaved by rounds of LOCC. The energy constraint is such that the average energy of the $n$ states input to each channel use should be bounded from above by a fixed positive real number, where the energy is with respect to a given energy observable.
It is sensible to consider an energy constraint $P$ for any such protocol in light of the fact that any real transmitter is necessarily power limited.
A third party Eve has access to all of the classical information exchanged between Alice and Bob, as well as the environment of each of the $n$ uses of the channel $\mathcal{N}$. For a
photon-loss channel, the physical meaning of the latter assumption is that Eve retains all of the light that is lost along the way from Alice to Bob.

\subsection{Secret-Key-Agreement Protocol with an Average Energy Constraint}

\label{ss:LOCCcapacityWithAvgEnergyConstraint}

\begin{figure}
	\begin{center}
		\includegraphics[width=3.4in]{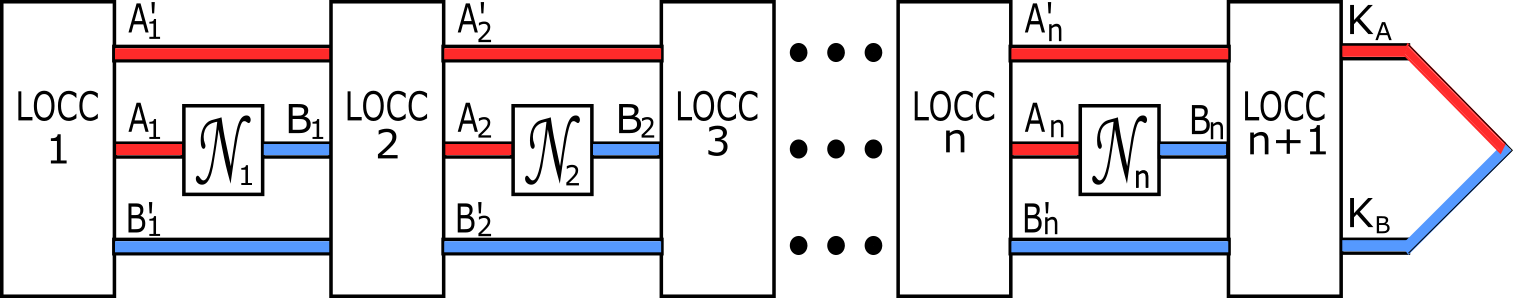}
		\caption[Protocol for 2 way assisted communication]{A secret-key-agreement protocol begins with Alice and Bob preparing a separable state of systems $A_1' A_1 B_1'$ using LOCC. Alice then feeds the $A_{1}$ system into the first channel in order to generate the $B_{1}$ system. After repeating this procedure $n$ times, with rounds of LOCC interleaved between every channel use, Alice and Bob perform a final round of LOCC, which yields the key systems $K_{A}$ and $K_{B}$.}
		\label{LOCCProtocol}
	\end{center}
\end{figure}

We first recall the notion of an energy observable:
\begin{definition}[Energy Observable]
For a Hilbert space $\mathcal{H}$,
let $G\in\mathcal{L}_{+}(\mathcal{H})$ denote 
 a positive semi-definite operator, defined in terms of its action on a vector $|\psi\rangle$ as
 \begin{equation}
 G|\psi\rangle=\sum_{j=1}^{\infty}g_{j}|e_{j}\rangle\langle e_{j}|\psi\rangle \label{eq:energyObservable} ,\end{equation}
 for $|\psi\rangle$ such that $\sum_{j=1}^{\infty}g_{j}|\langle e_{j}|\psi\rangle |^{2}<\infty$.
In the above, $\{|e_{j}\rangle\}_{j}$ is an orthonormal basis  and
$\{g_{j}\}_{j}$ is
a sequence of non-negative, real numbers. Then $\{|e_{j}\rangle\}_{j}$ is an eigenbasis for $G$ with corresponding eigenvalues $\{g_{j}\}_{j}$.
   We also follow the convention that \begin{equation} \Tr\{G\rho\}=\sup_{n} \Tr\{\Pi_{n}G\Pi_{n}\rho\} \label{eq:projectorConvention} ,\end{equation} where $\Pi_{n}$ is a spectral projector for $G$ corresponding to the interval $[0,n]$ \cite{H12,HS12}.
\end{definition}

We now formally define an energy-constrained secret-key-agreement protocol. Fix $n,K \in \mathbb{N}$, an energy observable $G$, a positive real $P\in[0,\infty)$, and $\varepsilon \in [0,1]$. An $(n,K,G,P,\varepsilon)$ secret-key-agreement protocol invokes $n$ uses of  a quantum channel $\mathcal{N}$, with each channel use interleaved by a countably decomposable LOCC channel. Such a protocol generates an $\varepsilon$-approximate tripartite key state of dimension $K$. Furthermore, the average energy of the channel input states, with respect to the energy observable $G$, is no larger than~$P$. Such a protocol is depicted in Figure~\ref{LOCCProtocol}.

In more detail, such a protocol begins with Alice and Bob performing an LOCC channel $\mathcal{L}^{(1)}_{\emptyset\to A^{\prime}_{1}A_{1}B^{\prime}_{1}}$ to generate a state $\rho^{(1)}_{A^{\prime}_{1}A_{1}B^{\prime}_{1}}$ that is separable with respect to the cut
$A^{\prime}_{1}A_{1}  | B^{\prime}_{1}$. 
Since the channel is a countably decomposable LOCC channel, the state 
$\rho^{(1)}_{A^{\prime}_{1}A_{1}B^{\prime}_{1}}$ is a countably decomposable separable state, as considered in \cite[Definition~1]{maks}.
Alice then inputs the system $A_{1}$  to the first channel use, resulting in the state
\begin{equation}
\sigma^{(1)}_{A^{\prime}_{1}B_{1}B^{\prime}_{1}}
\equiv
\mathcal{N}_{A_{1}\to B_{1}}(\rho^{(1)}_{A^{\prime}_{1}A_{1}B^{\prime}_{1}}).
\end{equation}
For now, we do not describe the systems that the eavesdropper obtains, and we only do so in the next subsection.
Alice and Bob then perform a second LOCC channel, producing
the state
\begin{equation} \rho^{(2)}_{A^{\prime}_{2}A_{2}B^{\prime}_{2}}
\equiv
\mathcal{L}^{(2)}_{A^{\prime}_{1}B_{1}B^{\prime}_{1}\to A^{\prime}_{2}A_{2}B^{\prime}_{2}}(\sigma^{(1)}_{A^{\prime}_{1}B_{1}B^{\prime}_{1}}). \label{secondLOCC} \end{equation} Next, Alice feeds system $A_{2}$ into the second channel use, which leads to the state
\begin{equation}
\sigma^{(2)}_{A^{\prime}_{2}B_{2}B^{\prime}_{2}}
\equiv
\mathcal{N}_{A_{2}\to B_{2}}(\rho^{(2)}_{A^{\prime}_{2}A_{2}B^{\prime}_{2}}).
\end{equation}
The procedure continues in this manner with a total of $n$ rounds of LOCC interleaved with $n$ uses of the channel as follows. For $i\in \{2,\ldots,n\}$, the relevant states of the protocol are
as follows:
\begin{align}
\rho^{(i)}_{A^{\prime}_{i}A_{i}B^{\prime}_{i}}&
\equiv
\mathcal{L}^{(i)}_{A^{\prime}_{i-1}B_{i-1}B^{\prime}_{i-1}\to A^{\prime}_{i}A_{i}B^{\prime}_{i}}(\sigma^{(i-1)}_{A^{\prime}_{i-1}B_{i-1}B^{\prime}_{i-1}}), \label{generalLOCC}
 \\
\sigma^{(i)}_{A^{\prime}_{i}B_{i}B^{\prime}_{i}}
&
\equiv
\mathcal{N}_{A_{i}\to B_{i}}(\rho^{(i)}_{A^{\prime}_{i}A_{i}B^{\prime}_{i}}). \label{generalChannel}
\end{align}
The primed systems correspond to separable Hilbert spaces.
After the $n$th channel use, a final LOCC channel is performed to produce key systems $K_{A}$ and $K_{B}$ for Alice and Bob, respectively, such that the final state is as follows: 
\begin{equation}
\omega_{K_{A}K_{B}}\equiv \mathcal{L}^{(n+1)}_{A^{\prime}_{n}B_{n}B^{\prime}_{n}\to K_{A}K_{B}}(\sigma^{(n)}_{A^{\prime}_{n}B_{n}B^{\prime}_{n}}).\label{finalLOCC}
\end{equation}

The average energy of the
$n$ channel input states
with respect to the
energy observable $G$  is constrained by $P$ as follows:
\begin{equation}
\frac{1}{n}\sum_{i=1}^{n}\Tr\{G\rho^{(i)}_{A_{i}}\}\leq P .\label{eq:energyConstraint}
\end{equation}
In the above,
$\rho^{(i)}_{A_{i}}$ is the marginal of the channel input states defined in \eqref{generalLOCC}.

One could alternatively demand a uniform bound on each channel input state, rather than a bound on the average energy. That is, one could demand that
\begin{equation}
\forall i \in \{1, \ldots, n\} : \Tr\{G\rho^{(i)}_{A_{i}}\}\leq P .
\end{equation}
Such an energy constraint would lead to a slightly different notion of capacity, and we return to this point later in Section~\ref{sec:energy-constr-cap}.

\subsection{The Purified Protocol} \label{ss:purified}

We now consider the role of a third party Eve in a secret-key-agreement protocol. The initial state $\rho^{(1)}_{A^{\prime}_{1}A_{1}B^{\prime}_{1}}$ is a separable state of the following form:
\begin{equation}
\rho_{A_{1}^{\prime}A_{1}B_{1}^{\prime}}^{(1)}\equiv\sum_{y_{1}}p_{Y_{1}}(y_{1})\tau_{A_{1}^{\prime}A_{1}}^{y_{1}}\otimes\zeta_{B_{1}}^{y_{1}},\label{separableRho}
\end{equation}
where $Y_{1}$ is a classical random variable corresponding to the message exchanged between Alice and Bob, which is needed to establish this state. The state
$\rho_{A_{1}^{\prime}A_{1}B_{1}^{\prime}}^{(1)}$
can be purified as
\begin{multline}
|\rho^{(1)} \rangle_{A_{1}^{\prime}A_{1}S_{A_{1}}B_{1}^{\prime}S_{B_{1}}Y_{1}}\equiv\\ \sum_{y_{1}}\sqrt{p_{Y_{1}}(y_{1})}|\tau^{y_{1}}\rangle_{A_{1}^{\prime}A_{1}S_{A_{1}}}\otimes|\zeta^{y_{1}}\rangle_{B_{1}S_{B_{1}}}\otimes|y_{1}\rangle_{Y_{1}}, \label{purifiedRho1}
\end{multline}
where the local shield systems
$S_{A_{1}}$ and $S_{B_{1}}$ are described by separable Hilbert spaces and in principle could be held by Alice and Bob, respectively, $|\tau^{y_{1}}\rangle_{A_{1}^{\prime}A_{1}S_{A_{1}}}$ and $|\zeta^{y_{1}}\rangle_{B_{1}S_{B_{1}}}$ purify $\tau_{A_{1}^{\prime}A_{1}}^{y_{1}}$ and $\zeta_{B_{1}}^{y_{1}}$, respectively, and Eve possesses system $Y_{1}$, which contains a coherent classical copy of the classical data exchanged.

\begin{figure}
	\begin{center}
		\includegraphics[width=3.4in]{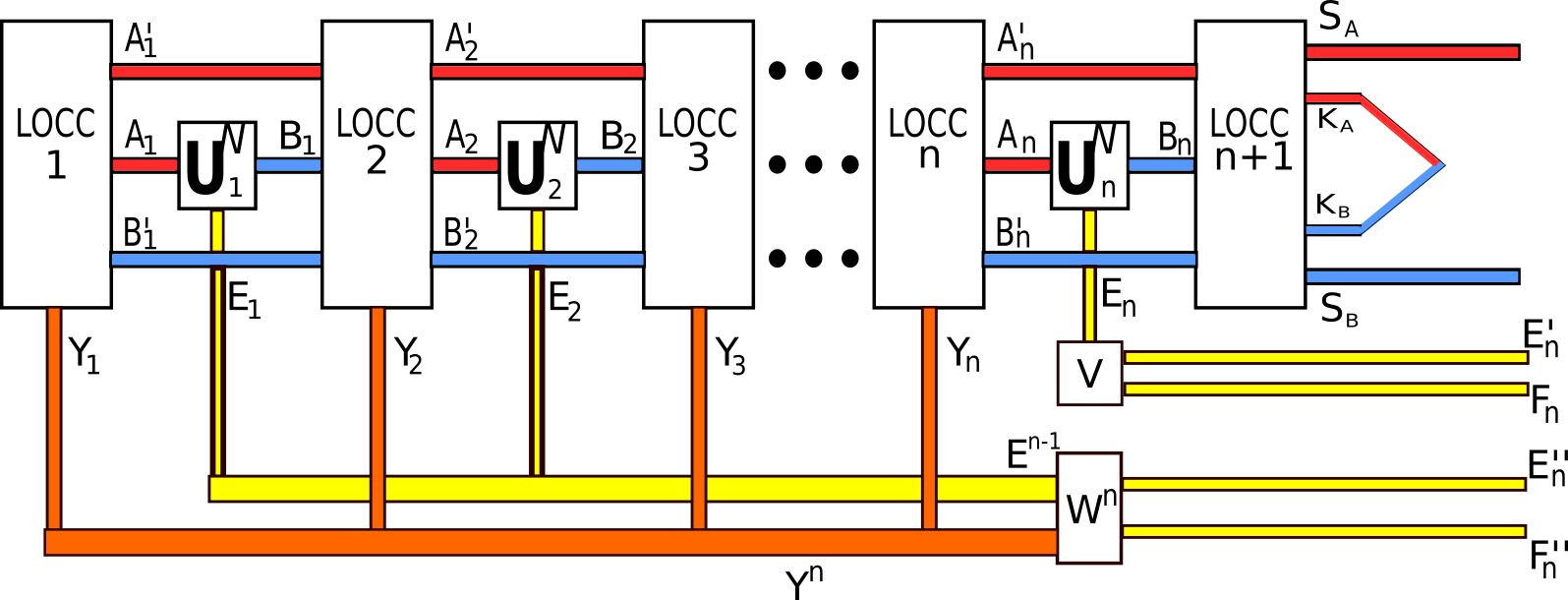}
		\caption[Purified 2 way assisted communication protocol]{Alice and Bob alternate rounds of LOCC and channel uses, just as in Figure~\ref{LOCCProtocol}. Each channel use is now purified, which yields outputs to Eve, the environment. Classical data is also collected by Eve from the LOCC. Eve's squashing channels are also purified and depicted above for the $n$th channel use.}
		\label{LOCCPurifiedProtocol}
	\end{center}
\end{figure}

Each LOCC channel $\mathcal{L}^{(i)}_{A^{\prime}_{i-1}B_{i-1}B^{\prime}_{i-1}\to A^{\prime}_{i}A_{i}B^{\prime}_{i}}$ for $i\in\{2,\ldots,n\}$ is of the form in \eqref{LOCCchannel} as
\begin{multline}
\mathcal{L}^{(i)}_{A^{\prime}_{i-1}B_{i-1}B^{\prime}_{i-1}\to A^{\prime}_{i}A_{i}B^{\prime}_{i}}
=\\
\sum_{y_{i}}
\mathcal{E}^{y_{i}}_{A_{i-1}^{\prime}\rightarrow A_{i}^{\prime}A_{i}}\otimes
\mathcal{F}^{y_{i}}_{B_{i-1}B_{i-1}^{\prime}\rightarrow B_{i}^{\prime}} ,
\end{multline}
and can be purified to an isometry in the following way:
\begin{multline}
U_{A_{i-1}^{\prime}B_{i-1}B_{i-1}^{\prime}\rightarrow A_{i}^{\prime}A_{i}S_{A_{i}}B_{i}^{\prime}S_{B_{i}}Y_{i}}^{\mathcal{L}^{(i)}}\\ \equiv\sum_{y_{i}}U_{A_{i-1}^{\prime}\rightarrow A_{i}^{\prime}A_{i}S_{A_{i}}}^{\mathcal{E}^{y_{i}}}\otimes U_{B_{i-1}B_{i-1}^{\prime}\rightarrow B_{i}^{\prime}S_{B_{i}}}^{\mathcal{F}^{y_{i}}}\otimes|y_{i}\rangle_{Y_{i}},
\label{eq:LOCC-as-separable-iso-ext}
\end{multline}
where $\{U_{A_{i-1}^{\prime}\rightarrow A_{i}^{\prime}A_{i}S_{A_{i}}}^{\mathcal{E}^{y_{i}}}\}_{y_{i}}$ and $\{U_{B_{i-1}B_{i-1}^{\prime}\rightarrow B_{i}^{\prime}S_{B_{i}}}^{\mathcal{F}^{y_{i}}}\}_{y_{i}}$ are collections of linear operators (each of which is a contraction, that is, $\Vert U_{A_{i-1}^{\prime}\rightarrow A_{i}^{\prime}A_{i}S_{A_{i}}}^{\mathcal{E}^{y_{i}}}\Vert _{\infty},\ \Vert U_{B_{i-1}B_{i-1}^{\prime}\rightarrow B_{i}^{\prime}S_{B_{i}}}^{\mathcal{F}^{y_{i}}}\Vert _{\infty}\leq1$) such that the linear operator in \eqref{eq:LOCC-as-separable-iso-ext} is an isometry. The systems $S_{A_i}$ and $S_{B_i}$ are shield systems belonging to Alice and Bob, respectively,  and $Y_{i}$ is a system held by Eve, containing a coherent classical copy of the classical data exchanged in this round. So a purification of the state
$\rho^{(i)}_{A^{\prime}_{i}A_{i}B^{\prime}_{i}}$
after each LOCC channel is as follows:
\begin{multline}
\vert \rho^{(i)}
\rangle _{A^{\prime}_{i}A_{i}
S_{A_1^{i}}
B^{\prime}_{i} S_{B_1^{i}}
E_1^{i-1} Y_1^{i}} 
\equiv\\
U_{A_{i-1}^{\prime}B_{i-1}B_{i-1}^{\prime}\rightarrow A_{i}^{\prime}A_{i}S_{A_{i}}B_{i}^{\prime}S_{B_{i}}Y_{i}}^{\mathcal{L}^{(i)}} \times \\ 
\vert \sigma^{(i-1)}\rangle _{A^{\prime}_{i-1}B_{i-1}B^{\prime}_{i-1} S_{A_1^{i-1}} S_{B_1^{i-1}} E_1^{i-1} Y_1^{i-1}},
\end{multline}
where we have employed the shorthands
$S_{A_1^{i}} \equiv S_{A_1} \cdots 
S_{A_i}$ and $S_{B_1^{i}} \equiv S_{B_1} \cdots 
S_{B_i}$, with a similar shorthand for 
$E_1^{i-1}$ and $ Y_1^{i}$.
A purification of the state
$\sigma^{(i)}_{A^{\prime}_{i}B_{i}B^{\prime}_{i}}$
after each use of the channel $\mathcal{N}_{A \to B}$ is
\begin{multline}
\vert \sigma ^{(i)}\rangle_
{A^{\prime}_{i}B_{i}
S_{A_1^{i}}
B^{\prime}_{i} S_{B_1^{i}}
E_1^{i} Y_1^{i}}
\equiv\\
U^\mathcal{N}_{A_{i}\to B_{i} E_{i}} \vert \rho^{(i)}
\rangle _{A^{\prime}_{i}A_{i}
S_{A_1^{i}}
B^{\prime}_{i} S_{B_1^{i}}
E_1^{i-1} Y_1^{i}} ,
\end{multline}
where $U^\mathcal{N}_{A_{i}\to B_{i} E_{i}} $ is an isometric extension of $i$th  channel use $\mathcal{N}_{A_{i}\to B_{i}} $.
The final LOCC channel also takes the form in \eqref{LOCCchannel}
\begin{equation}
\mathcal{L}_{A_{n}^{\prime}B_{n}B_{n}^{\prime}\rightarrow K_{A}K_{B}}^{(n+1)}=
\sum_{y_{n+1}}\mathcal{E}_{A_{n}^{\prime}\rightarrow K_{A}}^{y_{n+1}}\otimes\mathcal{F}_{B_{n}B_{n}^{\prime}\rightarrow K_{B}}^{y_{n+1}}, \label{finalLOCCchannel}
\end{equation}
and it can be purified to an isometry similarly as
\begin{multline}
U_{A_{n}^{\prime}B_{n}B_{n}^{\prime}\rightarrow K_{A}S_{A_{n+1}}K_{B}S_{B_{n+1}} Y_{n+1}}^{\mathcal{L}^{(n+1)}}\\ 
\equiv\sum_{y_{n+1}}U_{A_{n}^{\prime}\rightarrow K_{A}S_{A_{n+1}}}^{\mathcal{E}^{y_{n+1}}}\otimes U_{B_{n}B_{n}^{\prime}\rightarrow K_{B}S_{B_{n+1}}}^{\mathcal{F}^{y_{n+1}}}\otimes|y_{n+1}\rangle_{Y_{n+1}}. \label{finalLOCCchannelPurified} 
\end{multline}
The systems $S_{A_{n+1}}$ and $S_{B_{n+1}}$ are again shield systems belonging to Alice and Bob, respectively,  and $Y_{n+1}$ is a system held by Eve, containing a coherent classical copy of the classical data exchanged in this round.
As written above, each channel use $\mathcal{N}_{A_{i}\to B_{i}}$ can be purified, as in \eqref{eq:extendedChannel} and \eqref{eq:extendedKraus},
to an isometric channel $\mathcal{U}^{\mathcal{N}}_{A_{i}\to B_{i}E_{i}}$ such that Eve possesses system $E_{i}$ for all $i\in\{1,\ldots,n\}$.

The final state at the end of the purified protocol is a pure state $|\omega\rangle_{K_{A}S_{A}K_{B}S_{B}E^{n}Y^{n+1}}$, given by
\begin{multline}
|\omega\rangle_{K_{A}S_{A}K_{B}S_{B}E^{n}Y^{n+1}} = \\
U_{A_{n}^{\prime}B_{n}B_{n}^{\prime}\rightarrow K_{A}S_{A_{n+1}}K_{B}S_{B_{n+1}} Y_{n+1}}^{\mathcal{L}^{(n+1)}} \times \\
\vert \sigma ^{(n)}\rangle_
{A^{\prime}_{n}B_{n}
S_{A_1^{n}}
B^{\prime}_{n} S_{B_1^{n}}
E_1^{n} Y_1^{n}}.
\end{multline}
Alice is in possession of the key system $K_{A}$ and the shield systems $S_{A}\equiv S_{A_{1}}\ldots S_{A_{n+1}}$, Bob possesses the key system $K_{B}$ and the shield systems $S_{B}\equiv S_{B_{1}}\ldots S_{B_{n+1}}$, and Eve holds the environment systems $E^{n}\equiv E_{1}\ldots E_{n}$. The $S_{A}$, $S_{B}$, and $E^{n}$ systems all correspond to separable Hilbert spaces of generally infinite dimensions. Additionally, Eve has coherent copies
$Y^{n+1}\equiv Y_{1}\ldots Y_{n+1}$
 of all the classical data exchanged. By tracing over the systems $E^n$ and $Y^{n+1}$, it is clear that the protocol is an LOCC-assisted protocol whose aim is to generate an approximate bipartite private state on the systems $K_A S_A K_B S_B$.

For a fixed $n,K\in \mathbb{N}$ and $\varepsilon\in [0,1]$, the protocol is an $(n,K,G,P,\varepsilon)$ secret-key-agreement protocol if the final state 
$\omega_{K_{A}S_{A}K_{B}S_{B}}$ satisfies
\begin{equation}
F(\omega_{K_{A}S_{A}K_{B}S_{B}},\gamma_{K_{A}S_{A}K_{B}S_{B}})\geq 1-\varepsilon, \label{fidelityCriterion}
\end{equation}
where $\gamma_{K_{A}S_{A}K_{B}S_{B}}$ is a bipartite private state as in \eqref{privateState}. Alternatively (and equivalently), the criterion is that the final state $\omega_{K_{A}K_{B}E^n Y^{n+1}}$ satisfies
\begin{equation}
F(\omega_{K_{A}K_{B}E^n Y^{n+1}},\gamma_{K_{A}K_{B}E^n Y^{n+1}})\geq 1-\varepsilon, \label{fidelityCriterionTri}
\end{equation}
where $\gamma_{K_{A}K_{B}E^n Y^{n+1}}$ is a tripartite key state as in \eqref{secretKeyState}.

\subsection{Achievable Rates and Energy-Constrained Secret-Key-Agreement Capacity}

\label{sec:energy-constr-cap}

The rate $R=\frac{\log_2 K}{n}$ is a measure of the efficiency of the protocol, measured in secret key bits communicated per channel use. We say that the rate $R$ is achievable if, for all $\varepsilon\in (0,1)$, $\delta >0$, and for sufficiently large $n$, there exists an $(n,2^{n(R-\delta)},G,P,\varepsilon)$ secret-key-agreement protocol.

We call $P_{2}(\mathcal{N},G,P)$ the energy-constrained secret-key-agreement capacity of the channel $\mathcal{N}$, and it is equal to the supremum of all achievable rates subject to the energy constraint $P$ with respect to the energy observable~$G$.

As discussed previously in Section~\ref{ss:LOCCcapacityWithAvgEnergyConstraint}, one could have a modified notion of energy-constrained communication based on a uniform energy constraint, and this would lead to a different definition of capacity. However, it is clear from the definitions that for the same parameters $n$, $G$, $P$, and $\varepsilon$, the number of secret key values $K$ can only be the same or larger for a protocol having an average energy constraint, when compared to one that has a uniform constraint (simply because meeting the average energy constraint implies that the uniform energy constraint is met). Accordingly, the capacity with a uniform energy constraint can never exceed that with an average energy constraint. Since one of the main results of our paper is to obtain upper bounds on the (average) energy-constrained capacities, our results are much stronger than they would be had we only reported upper bounds on the uniform energy-constrained capacities.

\subsection{Energy-Constrained LOCC-assisted Quantum Communication}

We define the energy-constrained
LOCC-assisted quantum capacity $Q_{2}(\mathcal{N},G,P)$ of a channel $\mathcal{N}$ similarly. In this case, an $(n,K,G,P,\varepsilon)$ energy-constrained LOCC-assisted quantum communication protocol is defined similarly as in Section~\ref{ss:LOCCcapacityWithAvgEnergyConstraint}, but the main difference is that the final state $\omega_{K_{A}K_{B}}$ should satisfy the following inequality:
\begin{equation} F(\omega_{K_{A}K_{B}},\Phi_{AB})\geq 1-\varepsilon,\label{quantumFidelityCriterion} \end{equation} 
where $\Phi_{AB}$ is a maximally entangled state. Achievable rates are defined similarly as in the previous subsection, and the energy-constrained
LOCC-assisted quantum capacity $Q_{2}(\mathcal{N},G,P)$ of the channel $\mathcal{N}$ is defined to be equal to the supremum of all achievable rates.

It is worthwhile to note that the end goal of an LOCC-assisted quantum communication protocol is more difficult to achieve than a secret-key-agreement protocol for the same channel $\mathcal{N}$, energy observable $G$, energy constraint $P$, number $n$ of channel uses, and error parameter $\varepsilon$. This is because a maximally entangled state $\Phi_{K_A K_B}$ is a very particular kind of bipartite private state $\gamma_{K_A S_A K_B S_B}$, as observed in \cite{HHHO05,HHHO09}. Given this observation, it immediately follows that the energy-constrained LOCC-assisted quantum capacity is bounded from above by the energy-constrained secret-key-agreement capacity:
\begin{equation} Q_{2}(\mathcal{N},G,P)\leq P_{2}(\mathcal{N},G,P). \label{quantumCapLeqPrivateCap} \end{equation}

\section{Energy-Constrained Squashed Entanglement is an Upper Bound on Energy-Constrained Secret-Key-Agreement Capacity} 

\label{s:EsqBound}

The main goal of this section is to prove that the energy-constrained squashed entanglement of a quantum channel is an upper bound on its energy-constrained secret-key-agreement capacity. Before doing so, we recall the notion of a Gibbs observable \cite{H03,H04,HS06,Holevo2010,HS12,H12} and the finite output-entropy condition \cite{H03,H04,H12} for quantum channels.

\begin{definition}[Gibbs Observable]
\label{def:Gibbs-obs}
An energy observable $G$ is a Gibbs observable if \begin{equation} \Tr\{\exp(-\beta G)\}<\infty \label{gibbsObservable} \end{equation} for all $\beta>0$.
\end{definition}

This condition implies that there exists a well defined thermal state for $G$, having the following form for all $\beta > 0$ \cite{W78} (see also \cite{H03,HS06}):
\begin{equation}
	e^{-\beta G}/\Tr\{e^{-\beta G}\}. \label{eq:thermalState}
\end{equation}

\begin{condition}[Finite Output Entropy]
\label{finiteOutputEntropy} 
Let $G$ be  a Gibbs observable as in Definition~\ref{def:Gibbs-obs},
 and let $P\in [0,\infty)$ be an energy constraint. A quantum channel $\mathcal{N}$ satisfies the finite output-entropy condition with respect to $G$ and $P$ if
 \cite{H03,H04,H12} 
 \begin{equation} \sup_{\rho:\Tr\{G\rho\}\leq P} H(\mathcal{N}(\rho))<\infty .\end{equation}
\end{condition}

If a channel $\mathcal{N}$ satisfies the finite output-entropy condition
with respect to $G$ and $P$, then any complementary channel $\hat{\mathcal{N}}$ of $\mathcal{N}$ also satisfies the condition \cite{qi}:
\begin{equation} \sup_{\rho:\Tr\{G\rho\}\leq P} H(\hat{\mathcal{N}}(\rho))<\infty \label{finiteOutputEntropyComp} .\end{equation}

\begin{lemma}
\label{lem:finiteHimpliesfiniteEsq}
Finiteness of the output entropy of a channel $\mathcal{N}$ implies finiteness of the energy-constrained squashed entanglement of that channel. That is, if
\begin{equation}
\sup_{\rho:\Tr\{G\rho\}\leq P} H(\mathcal{N}(\rho))<\infty 
\end{equation}
holds, then
\begin{equation}
E_{\sq}(\mathcal{N}, G, P)<\infty \label{finiteSquashedEntanglementChannel} .\end{equation}
\end{lemma}
\begin{proof}
The statement is a consequence of
\eqref{mutualInfoBound}. Indeed,
applying the definition of squashed entanglement and picking the extension system $E$ to be trivial, we then get that 
\begin{equation}
E_{\sq}(A;B)_{\omega}\leq  \frac{1}{2} I(A;B)_{\omega}.
\label{ifQMIthenEsq}
\end{equation}
Applying Condition~\ref{finiteOutputEntropy} to \eqref{mutualInfoBound} and combining \eqref{ifQMIthenEsq} with the definition in \eqref{eq:constrainedSquashedeN} yields the statement of the lemma.
\end{proof}

\bigskip
We now establish the following weak-converse bound that applies to an arbitrary $(n,K,G,P,\varepsilon)$ energy-constrained secret-key-agreement protocol.

\begin{proposition}
\label{prop:EC-SE-upper-bnd}
Let $\mathcal{N}$ be a quantum channel satisfying the finite output-entropy condition (Condition~\ref{finiteOutputEntropy}), let $G$ be a Gibbs observable as in Definition~\ref{def:Gibbs-obs}, and let $P\in[0,\infty)$ be an energy constraint.
Fix $n,K\in \mathbb{N}$ and $\varepsilon \in (0,1)$. Then an 
$(n,K,G,P,\varepsilon)$ energy-constrained secret-key-agreement protocol for $\mathcal{N}$ is subject to the following upper bound in terms of the energy-constrained squashed entanglement of the channel $\mathcal{N}$:
\begin{equation}
\frac{1-2\sqrt{\varepsilon}}{n}
\log_2 K \leq 
E_{\operatorname{sq}}(\mathcal{N},G,P) + 
\frac{2}{n} g(\sqrt{\varepsilon}),
\end{equation}
where $g(\cdot)$ is defined in \eqref{eq:g-func}.
\end{proposition}

\begin{proof}
By assumption,
the final state $\omega_{K_A S_A K_B S_B}$ of any $(n,K,G,P,\varepsilon)$ secret-key-agreement protocol is an
$\varepsilon$-approximate bipartite private state, as given in \eqref{fidelityCriterion}.
Thus, the bound in \eqref{eq:appxpriv} applies, leading to  the following bound:
\begin{multline} 
	\log_2 K
	\leq E_{\sq}(K_{A}S_{A};K_{B}S_{B})_{\omega}\\
	+
	2\sqrt{\varepsilon} \log_2 K +2g(\sqrt{\varepsilon})
	 .\label{eq:continuityAndNormalized}
\end{multline} 

Let $\mathcal{U}^{\mathcal{N}}_{A\to BE}$
be an isometric channel extending the original channel ${\mathcal{N}}_{A\to B}$.
Let $V^{\mathcal{S}}_{E\to E^{\prime}F}$ denote an isometric extension of a squashing channel that can act on the environment system $E$ of the isometric channel 
$\mathcal{U}^{\mathcal{N}}_{A\to BE}$, and let $W^{n}_{E_1^{n-1} Y^n\to E^{\prime\prime}_{n}F^{\prime\prime}_{n}}$ denote an isometric extension of a squashing channel that can act on the systems $E^{n-1} Y^n$. Then we define the states
\begin{multline}
\vert \tau^{(n)} \rangle_{A^{\prime}_{n}B_{n}
S_{A_1^{n}}
B^{\prime}_{n} S_{B_1^{n}}
E^{\prime}_n F_n E_n^{\prime \prime}
F_n^{\prime \prime}} 
\equiv
\\
(V^{\mathcal{S}}_{E_n\to E_n^{\prime}F_n} \otimes W^{n}_{E_1^{n-1} Y^n\to E^{\prime\prime}_{n}F^{\prime\prime}_{n}}) \times \\ 
\vert \sigma ^{(n)}\rangle_
{A^{\prime}_{n}B_{n}
S_{A_1^{n}}
B^{\prime}_{n} S_{B_1^{n}}
E_1^{n} Y_1^{n}} ,
\end{multline}
and
\begin{multline}
\vert \zeta^{(n)} \rangle_{A^{\prime}_{n}A_{n}
S_{A_1^{n}}
B^{\prime}_{n} S_{B_1^{n}}
 E_n^{\prime \prime}
F_n^{\prime \prime}} 
\equiv
 W^{n}_{E_1^{n-1} Y^n\to E^{\prime\prime}_{n}F^{\prime\prime}_{n}} \times \\ 
\vert \rho ^{(n)}\rangle_
{A^{\prime}_{n}A_{n}
S_{A_1^{n}}
B^{\prime}_{n} S_{B_1^{n}}
E_1^{n-1} Y_1^{n}}.
\end{multline}
We invoke the LOCC monotonicity of squashed entanglement and the definition of squashed entanglement from \eqref{eq:squashede}, as well as Lemma~\ref{lemma:subadditivity-CQMI}, to find that 
\begin{align}
	& 2 E_{\sq}(K_{A}S_{A};K_{B}S_{B})_{\omega}
	\nonumber \\
	& \leq 2 E_{\sq}(A^{\prime}_{n}
	S_{A_1^n};B_{n}S_{B_1^n}B^{\prime}_{n})_{\sigma^{(n)}}\\
	&\leq I(A^{\prime}_{n}
	S_{A_1^n};B_{n}B^{\prime}_{n}
	S_{B_1^n}|E^{\prime}_{n}E^{\prime\prime}_{n})_{\tau^{(n)}}\\
	&\leq H(B_{n}|E^{\prime}_{n})_{\tau^{(n)}}+H(B_{n}|F_{n})_{\tau^{(n)}} \nonumber \\
	& \qquad 
	+ I(A^{\prime}_{n}
	S_{A_1^n} A_n;B^{\prime}_{n}
	S_{B_1^n}|E^{\prime\prime}_{n})_{\zeta^{(n)}}.
	\label{squashedeMonConAltConEntDual}
\end{align}
The conditions needed to apply Lemma~\ref{lemma:subadditivity-CQMI} indeed hold, following by hypothesis from \eqref{eq:energyConstraint}
and the finite output-entropy condition.
Since the isometric extension  $W^{n}_{E^{n-1} Y^n\to E^{\prime\prime}_{n}F^{\prime\prime}_{n}}$ of a 
squashing channel
is an arbitrary choice, the inequality above holds for the infimum over all such squashing channel extensions, and we find that
\begin{multline}
E_{\sq}(A^{\prime}_{n}
	S_{A_1^n};B_{n}S_{B_1^n}B^{\prime}_{n})_{\sigma^{(n)}} \leq \\
	\frac{1}{2} [H(B_{n}|E^{\prime}_{n})_{\tau^{(n)}}+H(B_{n}|F_{n})_{\tau^{(n)}}]  \\
	 \qquad 
	+ E_{\sq}(A^{\prime}_{n}
	S_{A_1^n} A_n;B^{\prime}_{n}
	S_{B_1^n})_{\rho^{(n)}}.
\end{multline}
We can then again invoke the LOCC monotonicity of squashed entanglement to find that
\begin{multline}
E_{\sq}(A^{\prime}_{n}
	S_{A_1^n} A_n;B^{\prime}_{n}
	S_{B_1^n})_{\rho^{(n)}} \leq \\
E_{\sq}(A^{\prime}_{n-1}
	S_{A_1^{n-1}} ;B_{n-1} B^{\prime}_{n-1}
	S_{B_1^{n-1}})_{\sigma^{(n-1)}}.	
\end{multline}
Now repeating the above reasoning $n-1$ more times (applying Lemma~\ref{lemma:subadditivity-CQMI} and LOCC monotonicity of squashed entanglement iteratively), we find that
\begin{align}
& 2 E_{\sq}(K_{A}S_{A};K_{B}S_{B})_{\omega} \nonumber \\
&\leq \sum_{i=1}^{n} [H(B_{i}|E^{\prime}_{i})_{\tau^{(i)}}+H(B_{i}|F_{i})_{\tau^{(i)}}] \nonumber \\
& \qquad	+ 2 E_{\sq}(A^{\prime}_{1}A_{1};B^{\prime}_{1})_{\rho^{(1)}} \\
&= \sum_{i=1}^{n} [H(B_{i}|E^{\prime}_{i})_{\tau^{(i)}}+H(B_{i}|F_{i})_{\tau^{(i)}}] \\
&= n \frac{1}{n} \sum_{i=1}^{n} [H(B_{i}|E^{\prime}_{i})_{\tau^{(i)}}+H(B_{i}|F_{i})_{\tau^{(i)}}] \\
& \leq n[H(B|E^{\prime})_{\overline{\tau}}+H(B|F)_{\overline{\tau}}].
\end{align}
The first equality follows because the state $\rho^{(1)}_{A^{\prime}_{1}A_{1}B^{\prime}_{1}}$ is separable, being the result of the initial LOCC, and so
$E_{\sq}(A^{\prime}_{1}A_{1};B^{\prime}_{1})_{\rho^{(1)}}=0$.
Note here that we are invoking the assumption that the protocol begins with a countably decomposable separable state \cite[Definition~1]{maks} and the fact that $E_{\sq} = 0$ for such states
\cite[Proposition~2]{maks}.
The last inequality follows from the
concavity of conditional entropy \cite{K11}, defining $\overline{\tau}_{B E' F}$ as the average output state of the channel:
\begin{equation}
\overline{\tau}_{B E' F} \equiv \frac{1}{n}
\sum_{i=1}^n \mathcal{V}^{\mathcal{S}}_{E_i \to E_i' F_i}(\sigma^{(i)}_{B_i E_i}).
\end{equation}
Since the inequality above holds for an arbitrary choice of the isometric channel 
$\mathcal{V}^{\mathcal{S}}_{E \to E' F}$ extending a squashing channel, and the average channel input state for the protocol satisfies the energy constraint in \eqref{eq:energyConstraint} by assumption, we find that
\begin{align}
&  E_{\sq}(K_{A}S_{A};K_{B}S_{B})_{\omega} \nonumber \\
& \leq 
 n \inf_{V^{S}_{E\to E^{\prime}F}}
 \frac{1}{2}
 [H(B|E^{\prime})_{\overline{\tau}}+H(B|F)_{\overline{\tau}}] \\
& \leq n  E_{\sq}(\mathcal{N},G,P),
\label{eq:almost-done-se-bnd}
\end{align}
where we have employed the alternative representation of squashed entanglement from \eqref{altEsqN}. Now combining \eqref{eq:continuityAndNormalized} and \eqref{eq:almost-done-se-bnd}, we conclude the proof.
\end{proof}

By applying Proposition~\ref{prop:EC-SE-upper-bnd} and taking the limit as $n \to \infty$ and then as $\varepsilon \to 0$, we arrive at the following theorem:

\begin{theorem}
\label{thm:EC-SE-upper-bnd}
Let $\mathcal{N}$ be a quantum channel satisfying the finite output-entropy condition (Condition~\ref{finiteOutputEntropy}), let $G$ be a Gibbs observable as in Definition~\ref{def:Gibbs-obs}, and let $P\in[0,\infty)$ be an energy constraint.
Then the energy-constrained squashed entanglement of the channel $\mathcal{N}$ is an upper bound on its energy-constrained secret-key-agreement capacity:
\begin{equation} P_{2}(\mathcal{N},G,P)\leq E_{\operatorname{sq}}(\mathcal{N},G,P). \label{EsqUpperBound} \end{equation}
\end{theorem}

Immediate consequences of Proposition~\ref{prop:EC-SE-upper-bnd}
and Theorem~\ref{thm:EC-SE-upper-bnd}
are bounds for rates of LOCC-assisted quantum communication. Indeed, 
let $\mathcal{N}$ be a quantum channel satisfying the finite output-entropy condition (Condition~\ref{finiteOutputEntropy}), let $G$ be a Gibbs observable as in Definition~\ref{def:Gibbs-obs}, and let $P\in[0,\infty)$ be an energy constraint.
Fix $n,K\in \mathbb{N}$ and $\varepsilon \in (0,1)$. Then an 
$(n,K,G,P,\varepsilon)$ energy-constrained LOCC-assisted quantum communication protocol for
$\mathcal{N}$
is subject to the following upper bound in terms of its energy-constrained squashed entanglement:
\begin{equation}
\frac{1-2\sqrt{\varepsilon}}{n}
\log_2 K \leq 
E_{\operatorname{sq}}(\mathcal{N},G,P) + 
\frac{2}{n} g(\sqrt{\varepsilon}).
\end{equation}
Then this implies that
\begin{equation} Q_{2}(\mathcal{N},G,P)\leq E_{\operatorname{sq}}(\mathcal{N},G,P). \label{qfinalBound} \end{equation}


\section{Bounds on Energy-Constrained Secret-Key-Agreement Capacities of Phase-Insensitive Quantum Gaussian Channels} 
\label{s:realizing}

The main result of Section~\ref{s:EsqBound} is that the energy-constrained squashed entanglement is an upper bound on the energy-constrained secret-key-agreement capacity of quantum channels that satisfy the finite output-entropy condition with respect to a given Gibbs observable. In this section, we specialize this result to particular phase-insensitive bosonic Gaussian channels that accept as input a single mode and output multiple modes. We prove here that a relaxation of the energy-constrained squashed entanglement of these channels is optimized by a thermal state input (when the squashed entanglement is written with respect to the representation in \eqref{altEsqN}).
Our results in this section thus generalize statements from prior works in \cite{tgw,tgwB,goodEW16}.

We also note the following point here before proceeding with the technical development. The prior works \cite{tgw,tgwB,goodEW16} argued that a thermal-state input should be the optimal choice for a particular relaxation of the energy-constrained squashed entanglement. However, it appears that these works have not given a full justification of these claims. In particular,  \cite{tgw,tgwB} appealed only to the extremality of Gaussian states \cite{WGC06} to argue that a thermal state should be optimal. However, it is necessary to argue that, among all Gaussian states, the thermal state is optimal. In \cite{goodEW16}, arguments about covariance of single-mode phase-insensitive Gaussian channels with respect to displacements and squeezing unitaries were given, but there was not an explicit proof of the latter covariance with respect to the squeezers, and furthermore, the squeezing unitaries can change the energy of the input state. Thus, in light of these questionable aspects, it seems worthwhile to provide a clear proof of the optimality of the thermal-state input, and our development in this section accomplishes this goal. The approach taken here is strongly related to that given in Section~5.2 and Remark~21 of \cite{SWAT17}.

\subsection{Single-Mode, Phase-Insensitive Bosonic Gaussian Channels
and Their Properties}

We begin in what follows by considering the argument for the particular case of phase-insensitive single-mode bosonic Gaussian channels. Three classes of channels of primary interest are thermal, amplifier, and additive-noise channels.

A thermal channel can be described succinctly in terms of the following Heisenberg-picture evolution:
	\begin{equation}
	\hat{b} =\sqrt{\eta}\hat{a}+\sqrt{1-\eta}\hat{e}  \label{eq:PLHeis-reduced},
	\end{equation} 
where $\hat{a}$, $\hat{b}$, and $\hat{e}$ represent respective bosonic annihilation operators for the sender, receiver, and environment. The parameter $\eta \in (0,1)$ represents the transmissivity of the channel, and the state of the environment is a bosonic thermal state
$\theta(N_B)$
of the following form:
\begin{equation}
\theta(N_B) \equiv \frac{1}{N_B+1}
\sum_{n=0}^\infty
\left(\frac{N_B}{N_B+1}\right)^n \vert n \rangle \langle n \vert,
\end{equation}
where $N_B \geq 0$ is the mean photon number of the above thermal state. So a thermal channel is characterized by two parameters: $\eta \in (0,1)$ and $N_B \geq 0$. If $N_B = 0$, then the channel is called a pure-loss channel because
 the environment state is prepared in a vacuum state and
the only corruption of the input signal is due to loss. An alternate description of a thermal channel in terms of its Kraus operators is available in \cite{ISS11}, and in what follows, we denote it by $\mathcal{L}_{\eta,N_B}$.

It is helpful to consider a unitary extension of a thermal channel. That is, we can consider a thermal channel arising as the result of a beamsplitter interaction between the input mode and the thermal-state environment mode, followed by a partial trace over the output environment mode. We can represent this interaction in the Heisenberg picture as follows:
\begin{align}
	\hat{b}&=\sqrt{\eta}\hat{a}+\sqrt{1-\eta}\hat{e}  \nonumber , \\
	\hat{e}^{\prime}&= - \sqrt{1-\eta}\hat{a}+\sqrt{\eta}\hat{e},
	\label{eq:PLHeis}
	\end{align}
where $\hat{e}^{\prime}$ denotes the output environment mode. Let $U^{\mathcal{L}_{\eta,N_B}}$ denote the Schr\"odinger-picture, two-mode unitary describing this interaction. It is well known that this unitary obeys the following phase covariance symmetry for all $\phi \in \mathbb{R}$:
\begin{equation}
	U^{\mathcal{L}_{\eta,N_B}}
	e^{i\hat{n}_{AE}\phi}=
	e^{i\hat{n}_{BE'}\phi} U^{\mathcal{L}_{\eta,N_B}} 
	,
	\end{equation}
where $\hat{n}_{AE} = \hat{n}_{A}+ \hat{n}_{E}$ is the total photon number operator for the input mode $A$ and environment mode $E$, while $\hat{n}_{BE'} = \hat{n}_{B}+ \hat{n}_{E'}$ is that for the output mode $B$ and the output environment mode $E'$. Thus, we can equivalently write the above phase covariance symmetry as
\begin{equation}
	U^{\mathcal{L}_{\eta,N_B}}
	(e^{i\hat{n}_{A}\phi} \otimes 
	e^{i\hat{n}_{E}\phi}) =
	(e^{i\hat{n}_{B}\phi} \otimes
	e^{i\hat{n}_{E'}\phi} ) U^{\mathcal{L}_{\eta,N_B}} \label{eq:phasePL}.
	\end{equation}
Due to this relation, the fact that a thermal state is phase invariant (i.e.,
$e^{i\hat{n}_{E}\phi} \theta(N_B)
e^{-i\hat{n}_{E}\phi} = \theta(N_B)$), and the fact that the thermal channel results from a partial trace after the unitary transformation $U^{\mathcal{L}_{\eta,N_B}}$, it follows that the thermal channel is phase covariant in the following sense:
\begin{equation}
\mathcal{L}_{\eta,N_B}(
e^{i\hat{n}_{A}\phi} \rho_A e^{-i\hat{n}_{A}\phi}) = 
e^{i\hat{n}_{B}\phi} \mathcal{L}_{\eta,N_B}(
 \rho_A ) e^{-i\hat{n}_{B}\phi},
\end{equation}
where $\rho_A$ is an arbitrary input state. This is the reason that thermal channels are called phase-insensitive.

Another class of channels to consider is the class of amplifier channels. 
An amplifier  channel can also be described succinctly in terms of the following Heisenberg-picture evolution:
	\begin{equation}
	\hat{b}=\sqrt{\mathscr{G}}\hat{a}+\sqrt{\mathscr{G}-1}\hat{e}^{\dagger} 	
	\label{eq:TMSHeis-reduced}
,
	\end{equation} 
where $\hat{a}$, $\hat{b}$, and $\hat{e}$ again represent respective bosonic annihilation operators for the sender, receiver, and environment. The parameter $\mathscr{G} \in (1,\infty)$ represents the gain of the channel, and the state of the environment is a bosonic thermal state
$\theta(N_B)$ with $N_B \geq 0$.
So an amplifier channel is characterized by two parameters: $\mathscr{G} \in (1,\infty)$ and $N_B \geq 0$. If $N_B = 0$, then the channel is called a pure-amplifier channel because the environment state is prepared in a vacuum state and the only corruption of the input signal is due to amplification, which inevitably introduces noise due to the no-cloning theorem \cite{Park1970,nat1982}. An alternate description of an amplifier channel in terms of its Kraus operators is available in \cite{ISS11}, and in what follows, we denote it by $\mathcal{A}_{\mathscr{G},N_B}$.

It is again helpful to consider a unitary extension of an amplifier channel. That is, we can consider an amplifier channel arising as the result of a two-mode squeezer interaction between the input mode and the thermal-state environment mode, followed by a partial trace over the output environment mode. We can represent this interaction in the Heisenberg picture as follows:
\begin{align}
	\hat{b}&=\sqrt{\mathscr{G}}\hat{a}+\sqrt{\mathscr{G}-1}\hat{e}^{\dagger} ,
	\nonumber \\
	\hat{e}^{\prime}&=\sqrt{\mathscr{G}-1}\hat{a}^{\dagger}+\sqrt{\mathscr{G}}\hat{e},\label{eq:TMSHeis}
	\end{align}
where $\hat{e}^{\prime}$ denotes the output environment mode. Let $U^{\mathcal{A}_{\mathscr{G},N_B}}$ denote the Schr\"odinger-picture, two-mode unitary describing this interaction. It is well known that this unitary obeys the following phase covariance symmetry for all $\phi \in \mathbb{R}$
\begin{equation}
	U^{\mathcal{A}_{\mathscr{G},N_B}}
	(e^{i\hat{n}_{A}\phi} \otimes 
	e^{-i\hat{n}_{E}\phi}) =
	(e^{i\hat{n}_{B}\phi} \otimes
	e^{-i\hat{n}_{E'}\phi} ) U^{\mathcal{A}_{\mathscr{G},N_B}} \label{eq:phaseTMS}.
	\end{equation}
Due to this relation, the fact that a thermal state is phase invariant, and the fact that the amplifier channel results from a partial trace of the unitary transformation $U^{\mathcal{A}_{\mathscr{G},N_B}}$, it follows that the amplifier channel is phase covariant in the following sense:
\begin{equation}
\mathcal{A}_{\mathscr{G},N_B}(
e^{i\hat{n}_{A}\phi} \rho_A e^{-i\hat{n}_{A}\phi}) = 
e^{i\hat{n}_{B}\phi} \mathcal{A}_{\mathscr{G},N_B}(
 \rho_A ) e^{-i\hat{n}_{B}\phi},
\end{equation}
where $\rho_A$ is an arbitrary input state. So amplifier channels  are also called phase-insensitive.

Another class of single-mode, phase-insensitive bosonic Gaussian channels are called additive-noise channels. These channels are easily described in the Schr\"odinger picture and are characterized by a single parameter $\xi \geq 0$, which is the variance of the channel. Additive-noise channels can be written as the following transformation:
\begin{equation}
\rho_A \to \int d^2 \alpha \  
\frac{\exp(-|\alpha|^2/\xi)}{\pi \xi} 
D(\alpha) \rho_A D(-\alpha),
\end{equation}
and can be interpreted as applying a unitary displacement operator $D(\alpha)$ randomly chosen according to a complex, isotropic  Gaussian distribution 
$\frac{\exp(-|\alpha|^2/\xi)}{\pi \xi} $ of variance $\xi$. These channels are phase-covariant as well and are thus phase-insensitive.

A well known theorem from \cite{CGH06,PhysRevLett.108.110505} establishes that any single-mode, phase-insensitive bosonic Gaussian channel $\mathcal{N}$ can be written as the serial concatenation of a pure-loss channel $\mathcal{L}_{T,0}$ of transmissivity $T \in [0,1]$ followed by a pure-amplifier channel $\mathcal{A}_{\mathscr{G},0}$ of gain $\mathscr{G}>1$:
\begin{equation}
\mathcal{N} = \mathcal{A}_{\mathscr{G},0} \circ \mathcal{L}_{T,0}.
\label{eq:phase-ins-decomp}
\end{equation}
This theorem has been extremely helpful in obtaining good upper bounds on various capacities of single-mode, phase-insensitive bosonic Gaussian channels \cite{KS13,tgw,tgwB,BW14,BPWW14,goodEW16,SWAT17,NAJ18}.

\subsection{Bounds for Single-Mode, Phase-Insensitive Bosonic Gaussian Channels}

In the following theorem, we prove that a thermal input state is the optimal state for a relaxation of the energy-constrained squashed entanglement of a single-mode, phase-insensitive bosonic Gaussian channel. This in turn gives an upper bound on the energy-constrained secret-key-agreement capacities of these channels, which has already been claimed in \cite{tgw,tgwB,goodEW16}.

\begin{theorem}
\label{thm:sq-ent-ph-ins-bound}
Let $\mathcal{N}$ be a single-mode,  phase-insensitive bosonic Gaussian channel as in \eqref{eq:phase-ins-decomp}. Then its energy-constrained squashed entanglement is bounded as
\begin{equation}
E_{\sq}(\mathcal{N},\hat{n},N_S)
\leq 
\frac{1}{2} [ H(B|E_1'E_2')_{\omega} +
H(B|F_1'F_2')_{\omega}
] , \label{eq:sq-ent-upper-bnd-ph-ins}
\end{equation}
where $\hat{n}$ is the photon number operator acting on the channel input mode, $N_S\geq 0$ is an energy constraint, $\omega_{BE_1'E_2'F_1'F_2'}$ is the following state:
\begin{equation}
\omega_{BE_1'E_2'F_1'F_2'} =
\mathcal{W}_{A\to BE^{\prime}_{1}E^{\prime}_{2}F^{\prime}_{1}F^{\prime}_{2}} 
(\theta(N_S)),
\end{equation}
and $\mathcal{W}$ is an isometric channel of the form
\begin{multline}
	\mathcal{W}_{A\to BE^{\prime}_{1}E^{\prime}_{2}F^{\prime}_{1}F^{\prime}_{2}}
	= 
	(\mathcal{V}^{\mathcal{A}}_{E_{2}\to E^{\prime}_{2}F^{\prime}_{2}}\circ 
	\mathcal{U}^{\mathcal{A}_{\mathscr{G},0}}_{B_{1}\to BE_{2}})
	\\
	\circ
	(\mathcal{V}^{\mathcal{L}}_{E_{1}\to E^{\prime}_{1}F^{\prime}_{1}}\circ \mathcal{U}^{\mathcal{L}_{T,0}}_{A\to B_{1}E_{1}}).
\end{multline}
In the above, $\mathcal{U}^{\mathcal{L}_{T,0}}$ is an isometric channel extending the pure-loss channel $\mathcal{L}_{T,0}$ and realized from \eqref{eq:PLHeis}.
Also, $\mathcal{U}^{\mathcal{A}_{\mathscr{G},0}}$ is an isometric channel extending the pure-amplifier channel $\mathcal{A}_{\mathscr{G},0}$ and realized from \eqref{eq:TMSHeis}.
Both $\mathcal{V}^{\mathcal{L}}_{E_{1}\to E^{\prime}_{1}F^{\prime}_{1}}$ and
$ \mathcal{V}^{\mathcal{A}}_{E_{2}\to E^{\prime}_{2}F^{\prime}_{2}}$ are bosonic Gaussian isometric channels that are phase covariant.
Figure~\ref{fig:4bs} depicts an
 example of the isometric channel $\mathcal{W}_{A\to BE^{\prime}_{1}E^{\prime}_{2}F^{\prime}_{1}F^{\prime}_{2}}$.

\end{theorem}

An immediate consequence of Theorems~\ref{thm:EC-SE-upper-bnd} and \ref{thm:sq-ent-ph-ins-bound} is the following corollary:
\begin{corollary}
\label{cor:p-cap-bnd-ph-ins}
With the same notation as in Theorem~\ref{thm:sq-ent-ph-ins-bound}, the energy-constrained secret-key-agreement capacity of the channel $\mathcal{N}$ is bounded as
\begin{equation}
P_{2}(\mathcal{N},\hat{n},N_{S})
\leq \frac{1}{2} [ H(B|E_1'E_2')_{\omega} +
H(B|F_1'F_2')_{\omega}
] .
\end{equation}
\end{corollary}

\begin{figure}
\begin{center}
\includegraphics[width=3.4in]{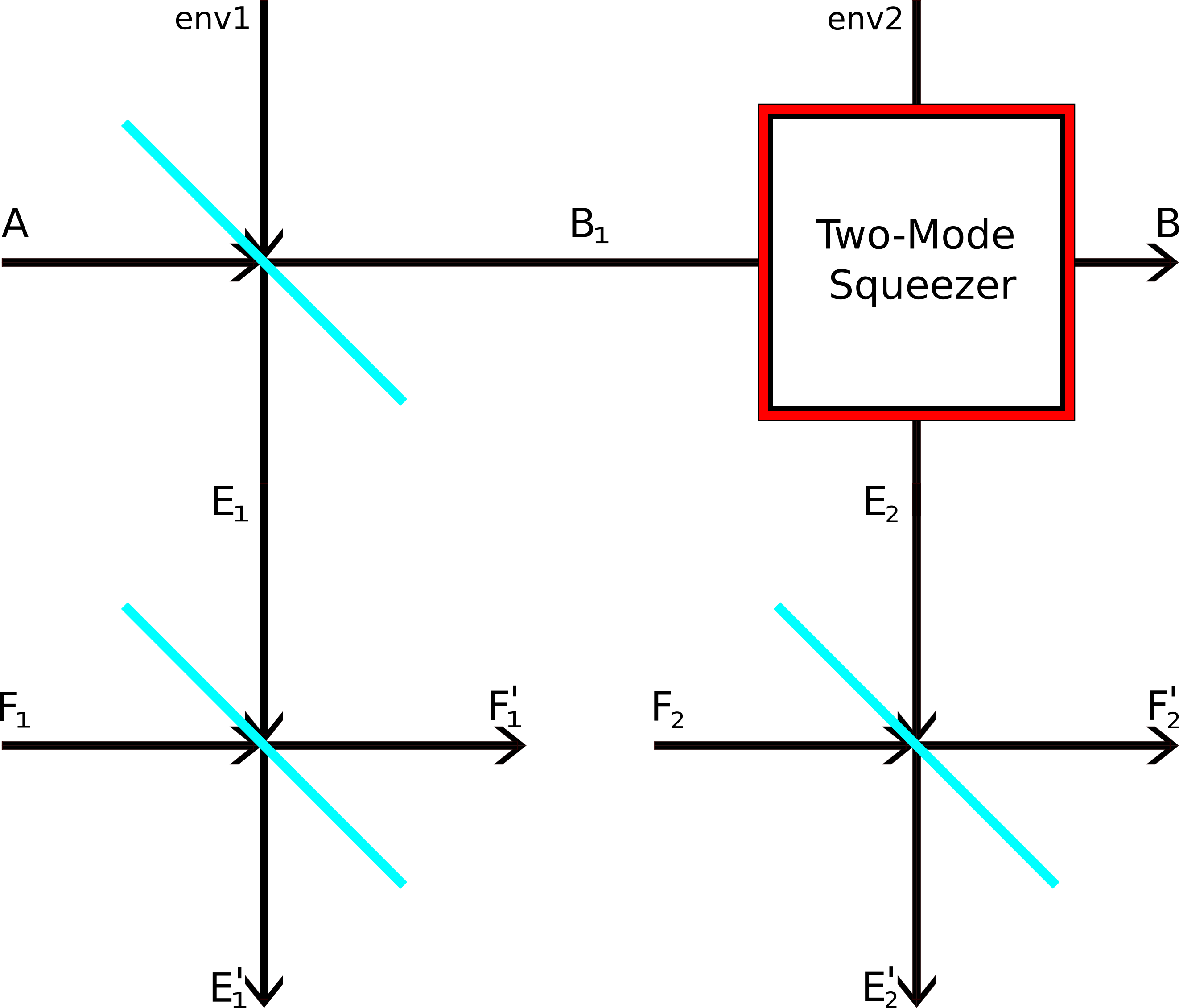}
\caption{A depiction of the isometric channel $\mathcal{W}_{A\to BE^{\prime}_{1}E^{\prime}_{2}F^{\prime}_{1}F^{\prime}_{2}}$ from Theorem~\ref{thm:sq-ent-ph-ins-bound}.
Note that this is the precise construction used in \cite{goodEW16}.
As stated in Theorem~\ref{thm:sq-ent-ph-ins-bound}, the isometric channel
$\mathcal{W}_{A\to BE^{\prime}_{1}E^{\prime}_{2}F^{\prime}_{1}F^{\prime}_{2}}$ is equal to $(\mathcal{V}^{\mathcal{A}}_{E_{2}\to E^{\prime}_{2}F^{\prime}_{2}}\circ 	\mathcal{U}^{\mathcal{A}_{\mathscr{G},0}}_{B_{1}\to BE_{2}}) \circ	(\mathcal{V}^{\mathcal{L}}_{E_{1}\to E^{\prime}_{1}F^{\prime}_{1}}\circ \mathcal{U}^{\mathcal{L}_{T,0}}_{A\to B_{1}E_{1}})$.
	The modes labeled ``env1'' and ``env2'' are respective environmental modes for the isometric channels
$\mathcal{U}^{\mathcal{L}_{T,0}}$ (top left) and $\mathcal{U}^{\mathcal{A}_{\mathscr{G},0}}$ (top right) and are prepared in the pure vacuum state.
The other isometric channels $\mathcal{V}^{\mathcal{L}}_{E_{1}\to E^{\prime}_{1}F^{\prime}_{1}}$ (bottom left) and
$\mathcal{V}^{\mathcal{A}}_{E_{2}\to E^{\prime}_{2}F^{\prime}_{2}}$ (bottom right) are chosen here to be 50-50 beamsplitters, following \cite{goodEW16}. 
The modes $F_1$ and $F_2$ are also prepared in the pure vacuum state.
Given this setup, Theorem~\ref{thm:sq-ent-ph-ins-bound} states that, among all possible input states with mean photon number $\leq N_S$, the thermal state $\theta(N_S)$ maximizes the entropy function $H(B|E_1'E_2') +
H(B|F_1'F_2')$.
}
\label{fig:4bs}
\end{center}
\end{figure}

\begin{proof}[Proof of Theorem~\ref{thm:sq-ent-ph-ins-bound}]
For convenience, we summarize the main steps of the proof here. We note that certain aspects of the proof bear some similarities to related approaches given in the literature \cite{H12,qi,NAJ18}, and the strongest overlap is with Remark~21 and Section~5.2 in \cite{SWAT17}.
\begin{enumerate}
\item
First, we employ the representation of a channel's squashed entanglement in \eqref{altEsqN}, and set  $\mathcal{U}^{\mathcal{A}_{\mathscr{G},0}}_{B_1 \to B E_2}\circ \mathcal{U}^{\mathcal{L}_{T,0}}_{A \to B_1 E_1}$ to be the isometric extension of $\mathcal{N} = \mathcal{A}_{\mathscr{G},0} \circ \mathcal{L}_{T,0}$.
\item
Then, we relax the infimum over all squashing isometries by setting it to be equal to
$\mathcal{V}^{\mathcal{L}}_{E_{1}\to E^{\prime}_{1}F^{\prime}_{1}}
\otimes \mathcal{V}^{\mathcal{A}}_{E_{2}\to E^{\prime}_{2}F^{\prime}_{2}}$. This leads to the isometric channel $\mathcal{W}_{A\to BE^{\prime}_{1}E^{\prime}_{2}F^{\prime}_{1}F^{\prime}_{2}}$ described in the theorem statement.
\item
Next, we employ the extremality of Gaussian states \cite{WGC06} to conclude that the entropy objective function $H(B|E_1'E_2') +
H(B|F_1'F_2')$ is maximized when the input state to mode $A$ is Gaussian.
\item We then employ the phase covariance of $\mathcal{W}_{A\to BE^{\prime}_{1}E^{\prime}_{2}F^{\prime}_{1}F^{\prime}_{2}}$ and concavity of conditional entropy to conclude that, for input states having a fixed mean photon number $N_S$,
 the entropy objective function $H(B|E_1'E_2') +
H(B|F_1'F_2')$
is maximized when the input state to mode $A$ is phase invariant.
\item Steps 3 and 4 imply that, for input states having a fixed mean photon number $N_S$, the optimal input state to mode $A$ should be a thermal state $\theta(N_S)$. This follows because $\theta(N_S)$ is the unique single-mode state of fixed mean photon number $N_S$ that is both Gaussian and phase invariant.
\item Finally, we use the displacement covariance of $\mathcal{W}_{A\to BE^{\prime}_{1}E^{\prime}_{2}F^{\prime}_{1}F^{\prime}_{2}}$ and concavity of conditional entropy to conclude that the 
entropy objective function $H(B|E_1'E_2') +
H(B|F_1'F_2')$ is monotone with respect to $N_S$. This finally implies that $\theta(N_S)$ is the optimal input state among all those having mean photon number $\leq N_S$.
\end{enumerate}

Steps one through three do not require any further justification, and so we proceed to step four. In what follows, we take the isometric channels $\mathcal{V}^{\mathcal{L}}_{E_{1}\to E^{\prime}_{1}F^{\prime}_{1}}$  and
$\mathcal{V}^{\mathcal{A}}_{E_{2}\to E^{\prime}_{2}F^{\prime}_{2}}$  to be 50-50 beamsplitters, following the heuristic from \cite{goodEW16} (based on numerical evidence that these are the best choices among all local phase-insensitive Gaussian channels). Thus, the isometries are manifestly phase covariant. However, note that our argument applies to arbitrary phase-covariant, bosonic Gaussian isometries
$\mathcal{V}^{\mathcal{L}}_{E_{1}\to E^{\prime}_{1}F^{\prime}_{1}}$  and
$\mathcal{V}^{\mathcal{A}}_{E_{2}\to E^{\prime}_{2}F^{\prime}_{2}}$.

Let $\rho_A$ denote an arbitrary input state of mean photon number $N_S$. The state $\rho_A$ can be input to the isometric channel 
$\mathcal{W}_{A\to BE^{\prime}_{1}E^{\prime}_{2}F^{\prime}_{1}F^{\prime}_{2}}$.
The entropy objective function 
$H(B|E_1'E_2')_{\mathcal{W}(\rho)} +
H(B|F_1'F_2')_{\mathcal{W}(\rho)}$	
is equal to a sum of conditional entropies and so we make use of two properties of conditional entropy: its invariance under local unitaries and concavity. Set
\begin{equation}
\hat{N} \equiv
\hat{n}_{B}
+ \hat{n}_{E^{\prime}_{1}}
- \hat{n}_{E^{\prime}_{2}}
+\hat{n}_{F^{\prime}_{1}}
-\hat{n}_{F^{\prime}_{2}},
\end{equation}
and consider the following phase shift unitary, depending on a phase $\phi\in\mathbb{R}$:
	\begin{equation}
	e^{i\hat{N}\phi}
	\equiv e^{i\hat{n}_{B}\phi}
\otimes e^{i\hat{n}_{E^{\prime}_{1}}\phi}
\otimes e^{-i\hat{n}_{E^{\prime}_{2}}\phi}
\otimes 
e^{i\hat{n}_{F^{\prime}_{1}}\phi}
\otimes e^{-i\hat{n}_{F^{\prime}_{2}}\phi} .
	\end{equation}
Then it follows from the invariance of conditional entropy under local unitaries that
\begin{multline}
	H(B|E^{\prime}_{1}E^{\prime}_{2})_{\mathcal{W}(\rho)}
	+H(B|F^{\prime}_{1}F^{\prime}_{2})_{\mathcal{W}(\rho)}
	\\
	=H(B|E^{\prime}_{1}E^{\prime}_{2})_{e^{i\hat{N}\phi}\mathcal{W}(\rho)e^{-i\hat{N}\phi}}
	\\+H(B|F^{\prime}_{1}F^{\prime}_{2})_{e^{i\hat{N}\phi}\mathcal{W}(\rho)e^{-i\hat{N}\phi}}. \label{eq:isometryPreserved}
	\end{multline}
Now exploiting the phase covariance of all of the isometric channels involved in 
$\mathcal{W}_{A\to BE^{\prime}_{1}E^{\prime}_{2}F^{\prime}_{1}F^{\prime}_{2}}$ (see 
\eqref{eq:phasePL} and \eqref{eq:phaseTMS}), we find that the last line above is equal to
	\begin{equation}
H(B|E^{\prime}_{1}E^{\prime}_{2})_{\mathcal{W}(e^{i\hat{n} \phi} \rho e^{-i\hat{n}\phi})} + 
H(B|F^{\prime}_{1}F^{\prime}_{2})_{\mathcal{W}(e^{i\hat{n} \phi}
\rho e^{-i\hat{n}\phi})} . 
	\end{equation}
	These equalities hold for any phase $\phi$ on the input, and so we can average over the input phase $\phi$ without changing the entropy objective function:
	\begin{multline}
	H(B|E^{\prime}_{1}E^{\prime}_{2})_{\mathcal{W}(\rho)}
	+H(B|F^{\prime}_{1}F^{\prime}_{2})_{\mathcal{W}(\rho)}
	\\
	=\frac{1}{2\pi}\int_{0}^{2\pi} d\phi \Bigg[ H(B|E^{\prime}_{1}E^{\prime}_{2})_{\mathcal{W}(e^{i \hat{n}\phi}
	\rho e^{-i \hat{n}\phi})}
	\\+H(B|F^{\prime}_{1}F^{\prime}_{2})_{\mathcal{W}(e^{i \hat{n}\phi}
	\rho e^{-i \hat{n}\phi})} \Bigg]. 
	\label{eq:phase-avg-no-change}
	\end{multline}
Let us define the phase-invariant state
$\overline{\rho}_A$ as
\begin{equation}
\overline{\rho}_A
\equiv
\frac{1}{2\pi}\int_{0}^{2\pi} d\phi \
e^{i \hat{n}\phi}\rho_{A} e^{-i \hat{n}\phi},
\end{equation}	
and note that the mean photon number of 
$\overline{\rho}_A$ is equal to $N_S$, which follows from the assumption that $\rho_A$ has mean photon number $N_S$ and the fact that phase averaging does not change the mean photon number.
Now exploiting the concavity of conditional entropy and the equality in \eqref{eq:phase-avg-no-change}, we find that	
\begin{multline}
H(B|E^{\prime}_{1}E^{\prime}_{2})_{\mathcal{W}(\rho)}
	+H(B|F^{\prime}_{1}F^{\prime}_{2})_{\mathcal{W}(\rho)}
	\\ \leq
H(B|E^{\prime}_{1}E^{\prime}_{2})_{\mathcal{W}(\overline{\rho})}
	+H(B|F^{\prime}_{1}F^{\prime}_{2})_{\mathcal{W}(\overline{\rho})}.	
\end{multline}
By combining with step three (extremality of Gaussian states), we conclude that, for an arbitrary state $\rho_A$ of mean photon number $N_S$, there exists a Gaussian, phase-invariant state that achieves the same or higher	value of the entropy objective function $H(B|E^{\prime}_{1}E^{\prime}_{2})
	+H(B|F^{\prime}_{1}F^{\prime}_{2})$.
So this completes step four, and step five is the next conclusion, which is that the thermal state $\theta(N_S)$ maximizes the entropy objective function with respect to all input states with mean photon number equal to $N_S$.
	

We now move on to the final step six.	
	In order to prove that the entropy objective function monotonically increases as a function of the mean photon number $N_S$ of an input thermal state, we repeat steps similar to those above that we used for step four. Recall again that conditional entropy is invariant under local unitaries, and so we can apply arbitrary displacements without changing the entropy objective function. In particular, since the local displacements can be arbitrary, we take advantage of the specific covariances of beam splitters and two-mode squeezers from \eqref{eq:PLHeis} and \eqref{eq:TMSHeis} when choosing the local displacements. We employ the following shorthand for the local displacements acting on the output modes of $\mathcal{W}$:
\begin{multline}
D^{\alpha}_{\operatorname{out}}
\equiv D_{B}(\sqrt{T \mathscr{G}}\alpha) \otimes D_{E^{\prime}_{1}}(\sqrt{\eta_{2}(1-T)}\alpha)
\\
\otimes 		D_{F^{\prime}_{1}}(\sqrt{(1-\eta_{2})(1-T)}\alpha)
\\
\otimes D_{E^{\prime}_{2}}(\sqrt{\eta_{3}T(\mathscr{G}-1)}\alpha^{\ast}) \\
\otimes D_{F^{\prime}_{2}}(\sqrt{(1-\eta_{3})T(\mathscr{G}-1)}\alpha^{\ast}), 
\end{multline}
where $\eta_2$ and $\eta_3$ are the transmissivities of the beamsplitters
$\mathcal{V}^{\mathcal{L}}_{E_{1}\to E^{\prime}_{1}F^{\prime}_{1}}$  and
$\mathcal{V}^{\mathcal{A}}_{E_{2}\to E^{\prime}_{2}F^{\prime}_{2}}$, respectively (here, however just set to $1/2$ for both). Let $\theta(N_1)$ be a thermal state of mean photon number $N_1 \geq 0$.
Then we find that
	\begin{multline}
	H(B|E^{\prime}_{1}E^{\prime}_{2})_{\mathcal{W}(\theta(N_{1}))}
	+
	H(B|F^{\prime}_{1}F^{\prime}_{2})_{\mathcal{W}(\theta(N_{1}))}
	\\=H(B|E^{\prime}_{1}E^{\prime}_{2})_{D^{\alpha}_{\operatorname{out}}\mathcal{W}(\theta(N_{1}))D^{\alpha\dagger}_{\operatorname{out}}}
	\\+
H(B|F^{\prime}_{1}F^{\prime}_{2})_{D^{\alpha}_{\operatorname{out}}\mathcal{W}(\theta(N_{1}))D^{\alpha\dagger}_{\operatorname{out}}}	. \label{eq:displacementIsometry}
	\end{multline}
	Employing the displacement covariance of the isometric Gaussian channel $\mathcal{W}$, we recast the local displacements on the outputs as a displacement of the input state:
	\begin{equation}
D^{\alpha}_{\operatorname{out}} \mathcal{W}(\theta(N_{1})) D^{\alpha\dagger}_{\operatorname{out}}
=
	\mathcal{W}\big(D_{A}(\alpha)\theta(N_{1}) D^{\dagger}_{A}(\alpha)\big) . \label{eq:bigDisplacement}
	\end{equation}
	Since this is true for any displacement $\alpha$, an expectation with respect to a probability distribution over $\alpha$ does not change the quantity, and by combining with \eqref{eq:displacementIsometry}, we find that
	\begin{multline}
		H(B|E^{\prime}_{1}E^{\prime}_{2})_{\mathcal{W}(\theta(N_{1}) )}
+H(B|F^{\prime}_{1}F^{\prime}_{2})_{\mathcal{W}(\theta(N_{1}) )}	
\\
		= \int d^2\alpha \  p^{N_{2}}(\alpha) \Bigg[ H(B|E^{\prime}_{1}E^{\prime}_{2})_{\mathcal{W}(D(\alpha)\theta(N_{1}) D^{\dagger}(\alpha))} 
		\\+
H(B|F^{\prime}_{1}F^{\prime}_{2})_{\mathcal{W}(D(\alpha)\theta(N_{1}) D^{\dagger}(\alpha))} \Bigg].\label{eq:avgOverAlpha} 
	\end{multline}
In the above, we choose the distribution $p^{N_{2}}(\alpha)$ to be a complex, isotropic Gaussian with variance  $N_{2} \geq 0$. Now recall the well known fact that Gaussian random displacements of a thermal state produce a thermal state of higher mean photon number:
\begin{equation}
	\int d^2 \alpha \  p^{N_{2}}(\alpha)
	\
	D(\alpha)\theta(N_{1}) D^{\dagger}(\alpha) = \theta(N_{1}+N_{2}) .\label{eq:thermalStateSum}
	\end{equation}
The concavity of conditional entropy and the equality in \eqref{eq:thermalStateSum} then imply that
	\begin{multline}
	H(B|E^{\prime}_{1}E^{\prime}_{2})_{\mathcal{W}(\theta(N_{1}))} 
+	H(B|F^{\prime}_{1}F^{\prime}_{2})_{\mathcal{W}(\theta(N_{1}))} \\
	\leq H(B|E^{\prime}_{1}E^{\prime}_{2})_{\mathcal{W}(\theta(N_1+N_2))}
\\	+
	H(B|F^{\prime}_{1}F^{\prime}_{2})_{\mathcal{W}(\theta(N_1+N_2))}. \label{eq:engConstThermal}
	\end{multline}
Since $N_1,N_2 \geq 0$ are arbitrary, we conclude that the entropy  objective function $H(B|E^{\prime}_{1}E^{\prime}_{2}) + H(B|F^{\prime}_{1}F^{\prime}_{2})$ is monotone increasing with respect to the mean photon number of the input thermal state. This now completes step six, and as such, we conclude the proof.
\end{proof}

\begin{remark}
We note here that \cite[Section~C.2]{goodEW16} provided an alternative way to handle step six in the above proof.
\end{remark} 

\begin{remark}
\label{rem:gen-sit-ph-ins-cond-ent}
Following Remark~21 of \cite{SWAT17},
the method used in the proof of Theorem~\ref{thm:sq-ent-ph-ins-bound} to establish the upper bound in \eqref{eq:sq-ent-upper-bnd-ph-ins} on $E_{\sq}(\mathcal{N},\hat{n},N_S)$ can be applied in far more general situations. Suppose that $\mathcal{N}$ is a single-mode input and multi-mode output channel. Suppose that $\mathcal{N}$ is phase covariant, such that a phase rotation on the input state is equivalent to
a product of local phase rotations on the output. Suppose that $\mathcal{N}$ is covariant with respect to displacement operators, such that a displacement
operator acting on the input state is equivalent to a product of local displacement operators on the
output. Then by relaxing the energy-constrained squashed entanglement in such a way that the squashing isometry has the same general phase and displacement covariances, it follows that, among all input states with mean photon number $\leq N_S$, the resulting objective function is maximized by a thermal state input with mean photon number equal to $N_S$.
\end{remark}

\begin{remark}
We can apply Theorem~\ref{thm:sq-ent-ph-ins-bound} and Corollary~\ref{cor:p-cap-bnd-ph-ins} to the pure-loss channel in order to recover one of the main claims of \cite{tgw,tgwB}. That is, the energy-constrained secret-key-agreement capacity of the pure-loss channel
$\mathcal{L}_{\eta,0}$ is bounded from above as
	\begin{equation}
	P_2(\mathcal{L}_{\eta,0},\hat{n},N_S)\leq
	g(N_{S}(1+\eta)/2)-g(N_{S}(1-\eta)/2).
	\end{equation}
Also, the following bound holds for the pure-amplifier channel $\mathcal{A}_{\mathscr{G},0}$, as a special case of a more general result stated in \cite{goodEW16}:
\begin{multline}
	P_2(\mathcal{A}_{\mathscr{G},0},\hat{n},N_S)\leq
	g(N_{S}[\mathscr{G}+1]/2+[\mathscr{G}-1]/2)
	\\
	-g([N_{S}+1][\mathscr{G}-1]/2).\label{eq:GEW16amp}
	\end{multline}
Since the bound in \eqref{eq:GEW16amp} was not explicitly stated in \cite{goodEW16}, for convenience, the arXiv posting of this paper includes a Mathematica file that can be used to derive \eqref{eq:GEW16amp}.
	Furthermore, other bounds on energy-constrained secret-key-agreement capacities of more general phase-insensitive channels are stated in \cite{goodEW16}.
\end{remark}

\subsection{Improved Bounds for Energy-Constrained Secret-Key-Agreement
Capacities of Bosonic Thermal Channels}

In this section, we discuss a variation of the method from \cite{goodEW16}
that leads to improvements of the bounds reported there. To begin with, we
note that any single-mode phase-insensitive channel $\mathcal{M}$, which is
not entanglement breaking \cite{HSR03}, can be decomposed as a pure-amplifier channel of
gain $\mathscr{G}>1$ followed by a pure-loss channel of transmissivity $T
\in(0,1]$:%
\begin{equation}
\mathcal{M}=\mathcal{L}_{T,0}\circ\mathcal{A}_{\mathscr{G},0}.
\label{eq:amp-then-loss-decomp}
\end{equation}
This result was found independently in
\cite[Theorem~30]{SWAT17} and \cite{RMG18,NAJ18} (see also \cite{notesSWAT17}). It has been used in
\cite{RMG18} to bound the unconstrained (and unassisted) quantum capacity of a thermal
channel, and it has been used in
\cite{SWAT17}
to bound the energy-constrained (and
unassisted) quantum and private capacities of an amplifier channel. After
\cite{RMG18} appeared, it was subsequently used in
\cite{SWAT17}
to bound the energy-constrained (and
unassisted) quantum and private capacities of a thermal channel. It has also been used most recently in \cite{NAJ18} in similar contexts.

For a thermal channel $\mathcal{L}_{\eta,N_{B}}$\ of transmissivity $\eta
\in\lbrack0,1]$ and thermal photon number $N_{B}\geq0$, the decomposition is
as above with
\begin{align}
T& =\eta-\left(  1-\eta\right)  N_{B} , \\
\mathscr{G} & =\eta/T.
\end{align}
Thus,
given that a thermal channel is entanglement breaking when $\eta\leq\left(
1-\eta\right)  N_{B}$ \cite{Holevo2008}, it is clear that the decomposition is only valid (i.e.,
$T\in\left(  0,1\right]  $) whenever the thermal channel is not entanglement
breaking. However, this is no matter when bounding secret-key-agreement or
LOCC-assisted quantum capacities, due to the fact that they vanish for any
entanglement-breaking channel.

\begin{figure}
\begin{center}
\includegraphics[width=3.2in]{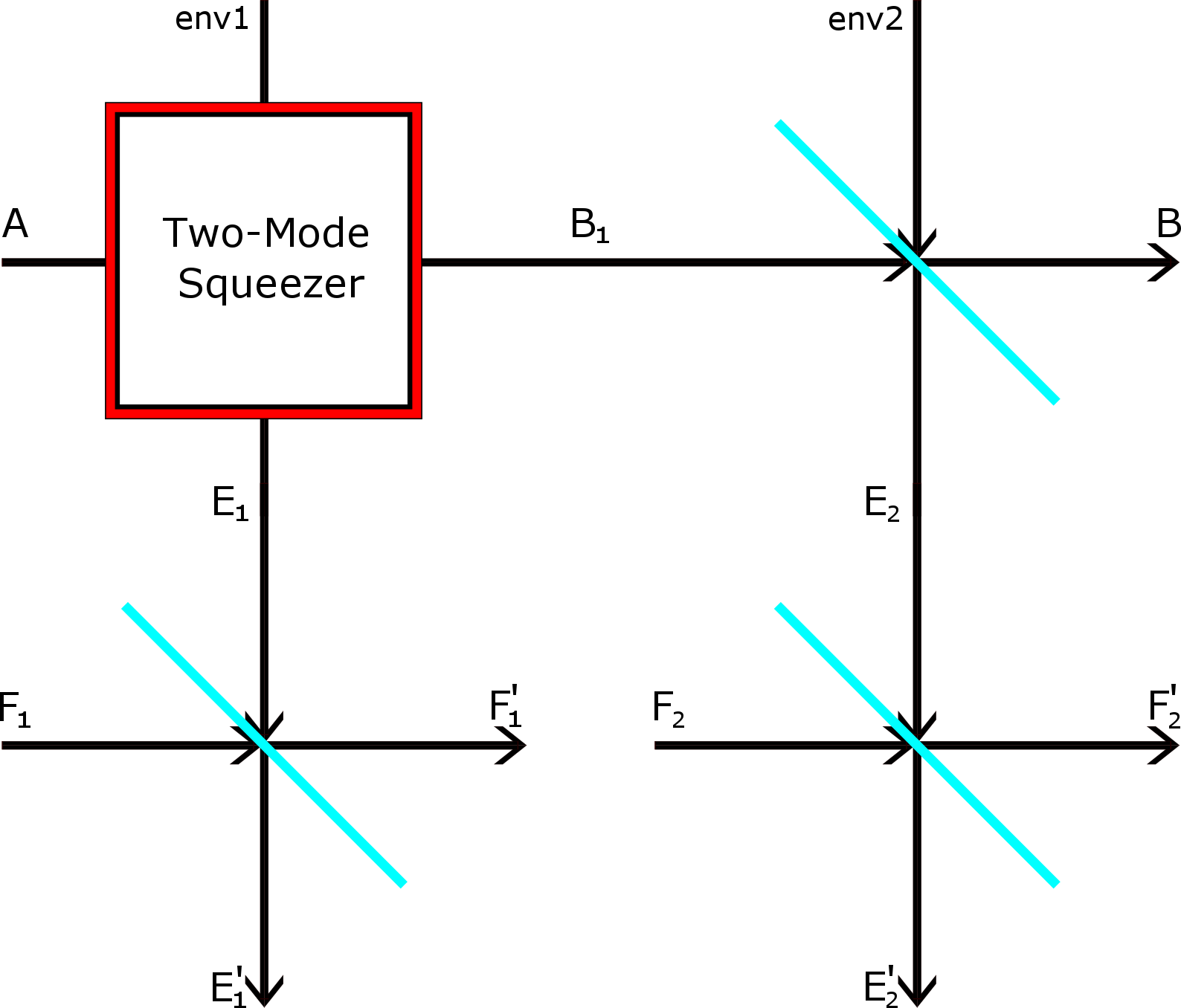}
\caption{A depiction of the isometric channel $\mathcal{W}_{A\to BE^{\prime}_{1}E^{\prime}_{2}F^{\prime}_{1}F^{\prime}_{2}}$ needed for the bound in Proposition~\ref{prop:new-finite-en-bound}.
This construction swaps the top-left beamsplitter and top-right two-mode squeezer from Figure~\ref{fig:4bs} and corresponds to the channel decomposition in \eqref{eq:amp-then-loss-decomp}. This construction leads to an improvement of the bound from \cite{goodEW16}.
}
\label{fig:4bs-other}
\end{center}
\end{figure}

Now, the main idea that leads to an improved energy-constrained bound is
simply 
to
employ the decomposition in \eqref{eq:amp-then-loss-decomp} and the same squashing isometries used in
\cite{goodEW16}. In other words, we are just swapping the top-left
beamsplitter with the top-right two-mode squeezer in Figure~\ref{fig:4bs}. For concreteness, we have depicted this change in
Figure~\ref{fig:4bs-other}. Let $\mathcal{W}$ denote the overall isometry taking the input mode
$A$ to the output modes $BE_{1}^{\prime}E_{2}^{\prime}F_{1}^{\prime}%
F_{2}^{\prime}$, as depicted in Figure~\ref{fig:4bs-other}. Then by the same reasoning as in the proof of Theorem~\ref{thm:sq-ent-ph-ins-bound} and subsequently given in
Remark~\ref{rem:gen-sit-ph-ins-cond-ent}, it follows that the thermal state $\theta(N_{S})$ of mean photon
number $N_{S}\geq0$ optimizes a relaxation of the energy-constrained squashed
entanglement corresponding to $\mathcal{W}$. This relaxation evaluates to%
\begin{multline}
\frac{1}{2}\left[  H(B|E_{1}^{\prime}E_{2}^{\prime})_{\mathcal{W}(\theta
(N_{S}))}+H(B|F_{1}^{\prime}F_{2}^{\prime})_{\mathcal{W}(\theta(N_{S}%
))}\right]  \\
=H(B|E_{1}^{\prime}E_{2}^{\prime})_{\mathcal{W}(\theta(N_{S}))},
\end{multline}
with the latter equality following due the symmetry resulting from choosing
each squashing isometry to be a 50-50 beamsplitter. This in turn implies the
following:

\begin{proposition}
\label{prop:new-finite-en-bound}
For a thermal channel $\mathcal{L}_{\eta,N_{B}}$ of transmissivity $\eta
\in\lbrack0,1]$ and thermal photon number $N_{B}\geq0$ such that $\eta
>\left(  1-\eta\right)  N_{B}$, its energy-constrained secret-key-agreement
capacity is bounded as%
\begin{equation}
P_{2}(\mathcal{L}_{\eta,N_{B}},\hat{n},N_{S})\leq H(B|E_{1}^{\prime}%
E_{2}^{\prime})_{\mathcal{W}(\theta(N_{S}))},\label{eq:DSW18-bnd}
\end{equation}
where $\mathcal{W}$ is the isometry depicted in Figure~\ref{fig:4bs-other}.
\end{proposition}

Now consider a general phase-insensitive single-mode bosonic Gaussian channel $\mathcal{M}$ that is not entanglement-breaking.
By applying Proposition~\ref{prop:new-finite-en-bound} and step six in the proof of Theorem~\ref{thm:sq-ent-ph-ins-bound}, we find that the quantity $H(B|E_{1}^{\prime}
E_{2}^{\prime})_{\mathcal{W}(\theta(N_{S}))}$ is monotone increasing with $N_S$, with $\mathcal{W}$ the corresponding isometry in Figure~\ref{fig:4bs-other}. Furthermore, the limit exists for all $T \in (0,1)$ and $\mathscr{G} > 1$ and converges to the same expression as given in \cite[Eq.~(29)]{goodEW16}:
\begin{multline}
\lim_{N_S \to \infty} H(B|E_{1}^{\prime}
E_{2}^{\prime})_{\mathcal{W}(\theta(N_{S}))} \\=
\frac{(1-T^2)\mathscr{G}\log_2\!\left(\frac{1+T}{1-T}\right) - (\mathscr{G}^2-1)T\log_2\!\left(\frac{\mathscr{G}+1}{\mathscr{G}-1}\right)}{1- \mathscr{G}^2 T^2}.
\label{eq:lim-inf-energy}
\end{multline}
We evaluated the latter limit with the aid of Mathematica and note here that the source files are available for download with the arXiv posting of this paper.

The fact that the expression in \eqref{eq:lim-inf-energy} is no different from that found in \cite[Eq.~(29)]{goodEW16} can be intuitively explained in the following way: Given that the input state to $\mathcal{W}$ is a thermal state, the limit $N_S \to \infty$ in some sense is like a classical limit, and in this limit, the commutation of the pure-loss channel and the pure-amplifier channel in \eqref{eq:amp-then-loss-decomp} makes no difference for the resulting expression. However, the values for $T$ and $\mathscr{G}$ for a thermal channel
$\mathcal{L}_{\eta,N_{B}}$
for the decomposition in \eqref{eq:amp-then-loss-decomp} are quite different from the values that $T$ and $\mathscr{G}$ would take in the decomposition in \eqref{eq:phase-ins-decomp}, and this is part of the reason that the decomposition in \eqref{eq:amp-then-loss-decomp} leads to an improved bound for a thermal channel $\mathcal{L}_{\eta,N_{B}}$.

In particular, for a thermal channel $\mathcal{L}_{\eta,N_{B}}$,
the  expression in \eqref{eq:lim-inf-energy} converges to zero in the entanglement-breaking limit $\eta \to N_B/(N_B+1)$ (or, equivalently,
$N_B \to \eta / (1-\eta)$; this limit calculation is included in our Mathematica files also). Due to this fact and the monotonicity of $H(B|E_{1}^{\prime}
E_{2}^{\prime})_{\mathcal{W}(\theta(N_{S}))}$ with $N_S$, \textbf{\textit{we conclude that the bound from Proposition~\ref{prop:new-finite-en-bound} converges to zero in the entanglement-breaking limit for any finite photon number $N_S$}}. This explains the improved behavior of the  bound in \eqref{eq:DSW18-bnd}, as compared to that from \cite{goodEW16}, as we discuss in what follows.

\subsubsection{Comparison of bounds on energy-constrained secret-key-agreement capacity of a thermal channel}

We have evaluated the bound in \eqref{eq:DSW18-bnd}\ numerically, and we found
strong numerical evidence that it outperforms the bound from \cite{goodEW16} for any values of
$N_{S}\geq0$, $\eta\in\left[  0,1\right]  $, and $N_{B}\geq0$ such that
$\eta>\left(  1-\eta\right)  N_{B}$.

It is also interesting to compare the bound in \eqref{eq:DSW18-bnd} with the
bounds from \cite{goodEW16}\ and \cite{PLOB17,WTB17}, for particular parameter
regimes. In \cite{PLOB17,WTB17}, the following photon-number-independent bound
was established:%
\begin{equation}
P_{2}(\mathcal{L}_{\eta,N_{B}},\hat{n},N_{S})\leq-\log_{2}(\left[
1-\eta\right]  \eta^{N_{B}})-g(N_{B}).
\label{eq:PLOBWTB}
\end{equation}

\begin{figure}
[ptb]
\begin{center}
\includegraphics[
width=3.7126in
]%
{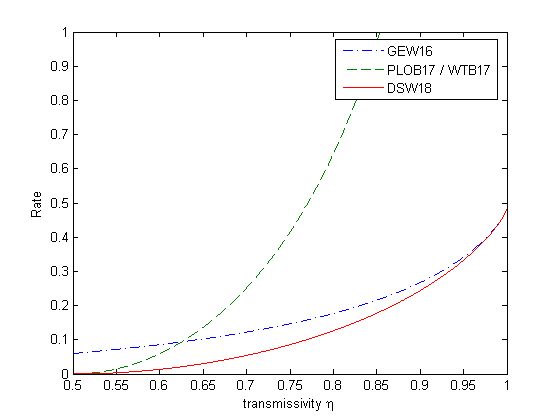}%
\caption{Comparison of the ``DSW18 bound'' from \eqref{eq:DSW18-bnd} with prior bounds
from \cite{goodEW16}\ and \cite{PLOB17,WTB17}, with $\eta\in[0.5,1]$, $N_{S}=0.1$ and $N_{B}=1$.
The plot shows that the bound in \eqref{eq:DSW18-bnd} converges to zero as the
channel becomes entanglement breaking.
}%
\label{fig:compare-1}%
\end{center}
\end{figure}

Figure~\ref{fig:compare-1}\ plots the three different bounds for a fixed
photon number $N_{S}=0.1$ and thermal photon number $N_{B}=1$. Therein, we see
that the bound in \eqref{eq:DSW18-bnd} improves upon the bounds from
\cite{goodEW16}\ and \cite{PLOB17,WTB17} for all transmissivities $\eta
\in\left[  1/2,1\right]  $. At $\eta=1/2$, the channel becomes entanglement
breaking for the aforementioned choice $N_{B}=1$, and we see that the
bound in \eqref{eq:DSW18-bnd} is converging to zero in the entanglement-breaking limit
$\eta\rightarrow1/2$, for fixed $N_B = 1$. The bound in \eqref{eq:DSW18-bnd} is also tighter than the one in \eqref{eq:PLOBWTB} for all values depicted in the plot.

Figure~\ref{fig:compare-another} plots the three different bounds for other parameter regimes, now
with $N_{S}\in\lbrack0,1]$, $\eta=0.1$, and $N_{B}$ set to $3\times10^{-7}$,
$1\times10^{-3}$, and $0.1$. These choices correspond to values expected in
a variety of experimental scenarios, as first discussed in \cite{RGRKVRHWE17} and subsequently
considered in \cite{KW17}. The bound in \eqref{eq:DSW18-bnd} is essentially
indistinguishable from that in \cite{goodEW16} for $N_{B}=3\times10^{-7}$, but
then the bound in \eqref{eq:DSW18-bnd} performs better as $N_{B}$ increases.

The Matlab files used to generate
Figures~\ref{fig:compare-1} and \ref{fig:compare-another} are available for download with the arXiv posting of this paper.

\begin{figure}
[ptb]
\begin{center}
\includegraphics[
width=3.3226in
]%
{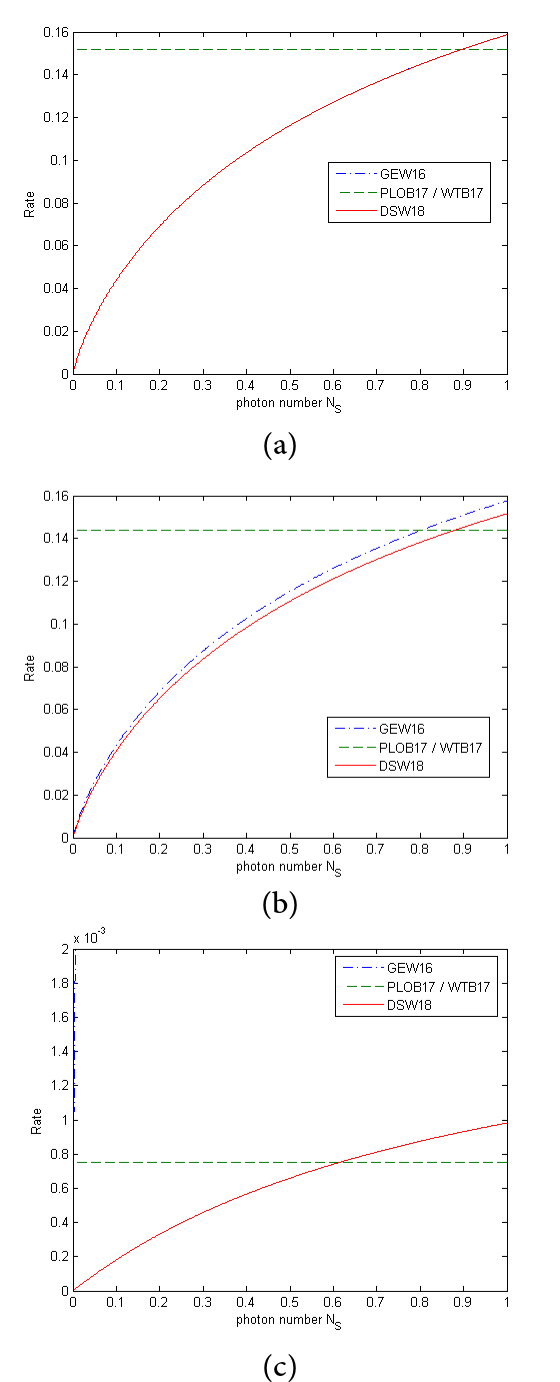}%
\caption{Comparison of the ``DSW18 bound''
 from \eqref{eq:DSW18-bnd} with prior bounds
from \cite{goodEW16}\ and \cite{PLOB17,WTB17}, with $N_{S} \in [0,1]$, $\eta = 0.1$, and $N_{B}\in \{3\times 10^{-7}, 1 \times 10^{-3}, 0.1\}$ (respectively, panels (a), (b), (c), above).
The DSW18 bound from \eqref{eq:DSW18-bnd} is indistinguishable from the bound from \cite{goodEW16} for small $N_B$, but then the bounds are
very different for higher~$N_B$.
In (a), GEW16 is not visible because it overlaps with DSW18.}%
\label{fig:compare-another}%
\end{center}
\end{figure}

\section{Multipartite Conditional Mutual Informations and Squashed Entanglement}

In this section, we review two different definitions of multipartite conditional mutual information from \cite{watanabe,H75,H78,CMS02,AHS08,YHHHOS}, and we prove that they satisfy a duality relation that generalizes the following well known duality relation for conditional mutual information:
\begin{equation}
I(A;B|C)_\psi = I(A;B|D)_\psi,
\label{eq:CQMI-duality}
\end{equation}
which holds for an arbitrary four-party pure state $\psi_{ABCD}$. This duality relation was established in \cite{DY08} and interpreted operationally therein in terms of the quantum state redistribution protocol \cite{DY08,YD09}, and it was recently generalized to the infinite-dimensional case in \cite{S15}, by employing the definition of conditional mutual information from \eqref{exconmutualiA}--\eqref{exconmutualiB}.

After establishing the multipartite generalization of the duality relation in \eqref{eq:CQMI-duality}, we prove that it implies that two definitions of multipartite squashed entanglement \cite{YHHHOS,AHS08} that were previously thought to be different are in fact equal to each other.

We finally then recall various properties of multipartite squashed entanglement, including how to evaluate it for multipartite GHZ and private states.

\subsection{Multipartite Conditional Quantum Mutual Informations}

\label{s:mcqmi}

We now recall two different multipartite generalizations of conditional mutual information \cite{watanabe,H75,H78,CMS02,AHS08,YHHHOS}. Consider an $m$-party state $\rho_{A_{1}\cdots A_{m}}$ acting on a tensor product of infinite-dimensional, separable Hilbert spaces. Let
$\rho_{A_{1}\cdots A_{m}E}$
denote an extension of this state,   which in turn can be purified to $\phi^{\rho}_{A_{1}\cdots A_{m}EF}$. The two generalizations of conditional quantum mutual information are known as the conditional total correlation and the conditional dual total correlation:
\begin{definition}[\cite{watanabe,AHS08,YHHHOS,S15}]
 The conditional total correlation 
of a state $\rho_{A_{1}\cdots A_{m}E}$
is defined as
\begin{equation}
	\label{eq:condTotalCorrel}
	I(A_{1};\cdots;A_{m}|E)_{\rho}
	\equiv \sum_{i=2}^{m} I(A_{i};A_{1}^{i-1}|E)_{\rho} .
\end{equation}
The notation $A_{1}^{i-1}$ refers to all the systems $A_{1}\cdots A_{i-1}$.
\end{definition}
\begin{definition}[\cite{H75,H78,CMS02,S15}]
 The conditional dual total correlation 
 of a state $\rho_{A_{1}\cdots A_{m}E}$
 is defined as
\begin{equation}
	\label{eq:condDualTotalCorrel}
	\widetilde{I}(A_{1};\cdots;A_{m}|E)_{\rho}
	\equiv \sum_{i=2}^{m} I(A_{i};A_{1}^{i-1}|A_{i+1}^{m}E)_{\rho},	
\end{equation}
where $A_{i+1}^{m} \equiv A_{i+1} \cdots A_m$. 
\end{definition}

Many years after the dual total correlation was defined and analyzed in \cite{H75,H78}, the conditional version of it was called ``secrecy monotone'' in \cite{CMS02} and analyzed there.

Note that the above quantities are invariant with respect to permutations of the systems $A_1$, \ldots, $A_m$.
This is more easily seen in the finite-dimensional case. That is, if the state $\rho_{A_1 \cdots A_m E}$ is finite-dimensional, then we have the following identities:
\begin{multline}
I(A_{1};\cdots;A_{m}|E)_{\rho}  \\
		 =\sum_{i=1}^{m}H(A_{i}|E)_{\rho}-H(A_{1}\cdots A_{m}|E)_{\rho}
\end{multline}
and
\begin{align}
 & \widetilde{I}(A_{1};\cdots;A_{m}|E)_{\rho}\notag\\
	& =\sum_{i=1}^{m} H(A_{[m]\backslash \{i\}}|E)_{\rho} -(m-1)H(A_{1}\cdots A_{m}|E)_{\rho} \notag \\
	& = H(A_{1}\cdots A_{m}|E)_{\rho}-\sum_{i=1}^{m} H(A_{i}|A_{[m]\backslash \{i\}}E)_{\rho} .
\end{align}

%

Although the two generalizations of CQMI in \eqref{eq:condTotalCorrel} and \eqref{eq:condDualTotalCorrel} are generally incomparable, they are related  by the following identity \cite{YHHHOS}:
\begin{multline}
	I(A_{1};\cdots;A_{m}|E)_{\rho}+\tilde I(A_{1};\cdots;A_{m}|E)_{\rho}
	\\
	=\sum_{i=1}^{m} I(A_{i};A_{[m]\backslash \{i\}}|E)_{\rho}. \label{eq:correlationRelation}
\end{multline}
The invariance of the above quantities with respect to permutations of the subsystems, as well as the validity of the identity in \eqref{eq:correlationRelation} in the general infinite-dimensional case, are consequences of Propositions~5 and 7
in~\cite{S15}.

\subsection{Duality for the Conditional Total Correlation and the Conditional Dual Total Correlation}

We now generalize the duality of CQMI
in \eqref{eq:CQMI-duality}
to the multipartite setting:

\begin{theorem}
\label{thm:duality-mult-CQMI}
	For a multipartite pure state $\phi_{A_{1}\cdots A_{m}EF}^{\rho}$, the
	following equality holds
	\begin{equation}
		I(A_{1};\cdots;A_{m}|E)_{\phi^{\rho}}=\widetilde{I}(A_{1};\cdots
		;A_{m}|F)_{\phi^{\rho}}. \label{eq:mcqmiDuality}
	\end{equation}
\end{theorem}

\begin{proof}
	There are at least two ways to see this. For the general infinite-dimensional case, we can simply apply definitions and the duality relation in \eqref{eq:CQMI-duality}. We find that
	\begin{align}
	I(A_{1};\cdots;A_{m}|E)_{\phi^{\rho}}&=\sum_{i=2}^{m}I(A_{i};A_{1}^{i-1}|E)_{\phi^{\rho}}\\
	&  =\sum_{i=2}^{m}I(A_{i};A_{1}^{i-1}|A_{i+1}^{m}F)_{\phi^{\rho}}\\
	&  =\widetilde{I}(A_{1};\cdots;A_{m}|F)_{\phi^{\rho}}.
	\end{align}
	In the less general case in which   conditional entropies are finite, we can apply a slightly different, but related method.  Recall that
	conditional entropy obeys a duality property: for a pure state $\psi_{ABC}$,
	we have that $H(A|B)_{\psi}=-H(A|C)_{\psi}$. Using the identities given above
	and this duality, we find that
	\begin{align}
		& I(A_{1};\cdots;A_{m}|E)_{\phi^{\rho}} \nonumber \\
		& =\sum_{i=1}^{m}H(A_{i}|E)_{\phi^{\rho}}-H(A_{1}\cdots A_{m}|E)_{\phi^{\rho}}\\
		&  =-\sum_{i=1}^{m}H(A_{i}|A_{[m]\backslash\{i\}}F)_{\phi^{\rho}}
		+H(A_{1}\cdots A_{m}|F)_{\phi^{\rho}}\\
		&  =\widetilde{I}(A_{1};\cdots;A_{m}|F)_{\phi^{\rho}}.
	\end{align}
	This concludes the proof.
\end{proof}

\begin{remark}
It is interesting to compare the somewhat long route by which Han arrived at the conditional dual total correlation in \cite{H78}, versus the comparatively short route by which we arrive at it in Theorem~\ref{thm:duality-mult-CQMI}. This latter method of using purifications and related entropy identities is unique to quantum information theory. It is also pleasing to find that the conditional total correlation and the conditional dual total correlation are dual to each other in the entropic sense of Theorem~\ref{thm:duality-mult-CQMI}.
\end{remark}

\subsection{Equivalence of Multipartite Squashed Entanglements} \label{s:equivEsq}

Two multipartite generalizations of the squashed entanglement of a state
$\rho_{A_{1}\cdots A_{m}}$ are based on the conditional total correlation and the conditional dual total correlation \cite{AHS08,YHHHOS}:
\begin{multline}
	E_{\sq}(A_{1};\cdots;A_{m})_{\rho} \equiv
	\frac{1}{2}\inf_{\rho_{A_{1}\cdots A_{m}E}}\Big\{I(A_{1};\cdots;A_{m}|E)_{\rho}
	\\
	:\Tr_{E}\{\rho_{A_{1}\cdots A_{m}E}\}=\rho_{A_{1}\cdots A_{m}}\Big\}, \label{eq:mEsq}
\end{multline}
\begin{multline}
	\widetilde{E}_{\sq}(A_{1};\cdots;A_{m})_{\rho} \equiv
	\frac{1}{2}\inf_{\rho_{A_{1}\cdots A_{m}E}}\Big\{\widetilde{I}(A_{1};\cdots;A_{m}|E)_{\rho}
	\\
	:\Tr_{E}\{\rho_{A_{1}\cdots A_{m}E}\}=\rho_{A_{1}\cdots A_{m}}\Big\}. \label{eq:mEsqTilde}
	\end{multline}

By employing Theorem~\ref{thm:duality-mult-CQMI}, we find that these quantities are in fact
always equal to each other, so that there is no need to consider two separate definitions, as was previously done in \cite{YHHHOS,STW16}:

\begin{theorem}
\label{thm:m-sq-equal}
	For a multipartite state $\rho_{A_{1}\cdots A_{m}}$, the following equality
	holds
	\begin{equation}
	E_{\sq}(A_{1};\cdots;A_{m})_{\rho}=\widetilde{E}_{\sq}(A_{1};\cdots;A_{m})_{\rho}.
	\end{equation}
\end{theorem}

\begin{proof}
	Let $\rho_{A_{1}\cdots A_{m}E}$ be an extension of $\rho_{A_{1}\cdots A_{m}}%
	$,\ and let $\phi_{A_{1}\cdots A_{m}EF}^{\rho}$ be a purification of
	$\rho_{A_{1}\cdots A_{m}E}$. Then by  Theorem~\ref{thm:duality-mult-CQMI},%
	\begin{align}
	I(A_{1};\cdots;A_{m}|E)_{\rho}  &  =\widetilde{I}(A_{1};\cdots;A_{m}|F)_{\phi^{\rho}}\\
	&  \geq2\widetilde{E}_{\sq}(A_{1};\cdots;A_{m})_{\rho}.
	\end{align}
	The inequality holds because $\phi_{A_{1}\cdots A_{m}F}^{\rho}$ is a
	particular extension of $\rho_{A_{1}\cdots A_{m}}$, and the squashed
	entanglement involves an infimum over all extensions of $\rho_{A_{1}\cdots A_{m}}$. Since the inequality holds for all extensions of $\rho_{A_{1}\cdots A_{m}}$, we can conclude that
	\begin{equation}
	E_{\sq}(A_{1};\cdots;A_{m})_{\rho}\geq\widetilde{E}_{\sq}(A_{1};\cdots;A_{m})_{\rho}.
	\end{equation}
	A proof for the other inequality $\widetilde{E}_{\sq}(A_{1};\cdots;A_{m})_{\rho}\geq E_{\sq}(A_{1};\cdots;A_{m})_{\rho}$ goes similarly.
\end{proof}

\begin{remark}
One of the main results of \cite{STW16} was to establish bounds on the secret-key-agreement capacity region of a quantum broadcast channel in terms of multipartite squashed entanglements. Theorem~\ref{thm:m-sq-equal} demonstrates that essentially half of the upper bounds written down in \cite{STW16} were in fact redundant. The same is true for the key distillation bounds from \cite{YHHHOS}.
\end{remark}

\subsection{Partitions and multipartite squashed entanglement}

\label{s:multiPrimer}
In this brief section, we recall some notation from \cite[Section~2.7]{STW16},
which we use in what follows as a shorthand for describing various partitions
of a set of quantum systems and their corresponding multipartite squashed
entanglements. Given a set $\mathcal{W}$ of quantum systems, a partition
$\mathbb{G}=\{\chi_{1},\dots,\chi_{|\mathbb{G}|}\}$ is a set of non-empty
subsets of $\mathcal{W}$ such that
\begin{equation}
\bigcup\limits_{\chi_{i}\in\mathbb{G}}\chi_{i}=\mathcal{W},
\end{equation}
and for all $\chi_{i},\chi_{j}\in\mathbb{G}$, $i\neq j$,
\begin{equation}
\chi_{i}\cap\chi_{j}=\emptyset.
\end{equation}
For example, one possible partition of $\mathcal{W}=\{A,B,C\}$ is given by
$\mathbb{G}=\{\{AB\},\{C\}\}$. The power set $\mathcal{P}(\mathcal{W})$ is the
set of all subsets of $\mathcal{W}$. The  sets $\mathcal{P}_{\geq
1}(\mathcal{W})$ and $\mathcal{P}_{\geq2}(\mathcal{W})$ are the sets of all subsets
of $\mathcal{W}$ with greater than or equal to one and two members,
respectively. That is, for $\mathcal{W}=\{A,B,C\}$,%
\begin{align}
\mathcal{P}(\mathcal{W}) &  =\{\emptyset
,\{A\},\{B\},\{C\},\{A,B\},\{A,C\},\nonumber\\
&  \qquad\{B,C\},\{A,B,C\}\} ,\\
\mathcal{P}_{\geq1}(\mathcal{W}) &
=\{\{A\},\{B\},\{C\},\{A,B\},\{A,C\},\nonumber\\
&  \qquad\{B,C\},\{A,B,C\}\} ,
\\
\mathcal{P}_{\geq2}(\mathcal{W}) &
=\{\{A,B\},\{A,C\},\{B,C\},\{A,B,C\}\}.\label{eq:powerSet2}%
\end{align}
Given a set $\mathcal{Y}$, let $\omega_{\mathcal{Y}}$ denote a $|\mathcal{Y}%
|$-partite state shared by the parties specified by the elements of
$\mathcal{Y}$. If $\mathbb{G}$ is a partition of $\mathcal{Y}$, then the
notation
\begin{equation}
E_{\sq}(\mathbb{G})_{\omega}%
\end{equation}
refers to the multipartite squashed entanglement with parties grouped
according to partition $\mathbb{G}$. For example, if $\mathcal{Y}=\{A,B,C\}$,
$\omega_{\mathcal{Y}}=\omega_{ABC}$, $\mathbb{G}_{1}=\{\{A\},\{B\},\{C\}\}$,
and $\mathbb{G}_{2}=\{\{AB\},\{C\}\}$, then
\begin{align}
E_{\sq}(\mathbb{G}_{1})_{\omega} &  =E_{\sq}(A;B;C)_{\omega},\qquad
\text{and}\label{eq:eSqPart1}\\
E_{\sq}(\mathbb{G}_{2})_{\omega} &  =E_{\sq}(AB;C)_{\omega}%
.\label{eq:eSqPart2}%
\end{align}

\subsection{Multipartite Private States}

One multipartite generalization of the maximally entangled state in 
\eqref{eq:maximallyentangled} is the Greenberger-Horne-Zeilinger (GHZ) state.
A GHZ state of $\log_{2} K$ entangled bits of an $m$-party system $A_{1}$, \ldots, $ A_{m}$ takes the form
\begin{equation}
|\Phi\rangle_{A_{1}\cdots A_{m}}=\frac{1}{\sqrt{K}}\sum_{i=1}^{K}%
|i\rangle_{A_{1}}\otimes\cdots\otimes|i\rangle_{A_{m}}\label{eq:GHZ}%
\end{equation}
where $\{|i\rangle_{A_{1}}\},\ldots,\{|i\rangle_{A_{m}}\}$ are orthonormal
basis sets for their respective systems. The bipartite private states
from \eqref{privateState} 
are similarly
generalized to the multipartite case \cite{HA06}, so that a state of $\log
_{2}K$ private bits is as follows:%
\begin{multline}
\gamma_{A_{1}\cdots A_{m}A_{1}^{\prime}\cdots A_{m}^{\prime}}\\
=U_{A_{1}\cdots A_{m}A_{1}^{\prime}\cdots A_{m}^{\prime}}(|\Phi\rangle
\langle\Phi|_{A_{1}\cdots A_{m}}\otimes\rho_{A_{1}^{\prime}\cdots
A_{m}^{\prime}})\times\\
U_{A_{1}\cdots A_{m}A_{1}^{\prime}\cdots A_{m}^{\prime}}^{\dagger
},\label{eq:multiPrivateState}%
\end{multline}
with the GHZ state generalizing the role of the maximally entangled state, and the
twisting unitary from \eqref{twistingUnitary} is generalized as%
\begin{multline}
U_{A_{1}\cdots A_{m}A_{1}^{\prime}\cdots A_{m}^{\prime}}=\sum_{i_{1}%
,\ldots,i_{m}=1}^{K}|i_{1},\ldots,i_{m}\rangle\langle i_{1},\ldots
,i_{m}|_{A_{1}\cdots A_{m}}\\
\otimes U_{A_{1}^{\prime}\cdots A_{m}^{\prime}}^{i_{1},\ldots,i_{m}},
\end{multline}
where $U_{A_{1}^{\prime}\cdots A_{m}^{\prime}}^{i_{1},\ldots,i_{m}}$ are
unitary operators depending on the values $i_{1},\ldots,i_{m}$.

\subsection{Properties of Multipartite Squashed Entanglement}

\label{ss:multiEsqProps}Multipartite squashed entanglement possesses a number
of useful properties that have been proven separately in \cite{STW16} for
the quantities in \eqref{eq:mEsq} and \eqref{eq:mEsqTilde}. In light of
Theorem~\ref{thm:m-sq-equal}, we now know that these measures are equal. Since we
require these properties in what follows, we recall some of them here:

\begin{lemma}
[Subadditivity \cite{STW16}]\label{lem:multiSubadditivity}Given a pure state
$\phi_{RA_{1}\cdots A_{m}B_{1}\cdots B_{m}EF}$, the following inequality holds%
\begin{multline}
E_{\sq}(R;A_{1}B_{1};\cdots;A_{m}B_{m})_{\phi}\leq\\
E_{\sq}(RA^{m}E;B_{1};\cdots;B_{m})_{\phi}+\\
E_{\sq}(RB^{m}F;A_{1};\cdots;A_{m})_{\phi}%
\end{multline}
where the notation $A^{m}$ refers to all systems $A_{1}\cdots A_{m}$ and a
similar convention for $B^{m}$.
\end{lemma}

Technically speaking, \cite{STW16} did not establish the above statement in the general infinite-dimensional case, but we note here that the approach from \cite{S15} can be used to establish the lemma above.

\begin{lemma}
[Monotonicity for Groupings \cite{STW16}]Squashed entanglement is
non-increasing when subsystems are grouped. That is, given a state
$\rho_{A_{1}\cdots A_{m}}$, the following inequality holds%
\begin{equation}
E_{\sq}(A_{1};A_{2};\cdots;A_{m})_{\rho}\geq E_{\sq}(A_{1}A_{2};A_{3}%
;\cdots;A_{m})_{\rho}.
\end{equation}

\end{lemma}

\begin{lemma}
[Product States \cite{STW16}]Let
\begin{equation}
\omega_{AB_{1}\cdots B_{m}%
}=\rho_{A}\otimes\sigma_{B_{1}\cdots B_{m}}
\end{equation}
 where  $\rho_{A}$ and $\sigma_{B_{1}\cdots B_{m}}$ are density operators.
Then%
\begin{equation}
E_{\sq}(A;B_{1};\cdots;B_{m})_{\omega}=E_{\sq}(B_{1};\cdots;B_{m})_{\sigma}.
\end{equation}

\end{lemma}

We also have the following alternative representation of multipartite squashed
entanglement, which was employed implicitly in
\cite{STW16}:

\begin{lemma}
\label{lem:alt-rep-mult-sq-e}Let $\rho_{A_{1}\cdots A_{m}}$ be a multipartite
density operator such that the entropy $H(A_i)_{\rho} < \infty$ for all $i \in \{2,\ldots, m\}$. Then its multipartite squashed entanglement can be written
as%
\begin{multline}
E_{\sq}(A_{1};A_{2};\cdots;A_{m})_{\rho}=\\
\frac{1}{2}\inf_{\mathcal{V}%
_{E\rightarrow E^{\prime}F}}\left[  \sum_{i=2}^{m}H(A_{i}|E^{\prime})_{\omega
}+H(A_{2}\cdots A_{m}|F)_{\omega}\right]  ,
\end{multline}
where the infimum is with respect to an isometric channel $\mathcal{V}%
_{E\rightarrow E^{\prime}F}$,%
\begin{equation}
\omega_{A_{1}\cdots A_{m}E^{\prime}F}\equiv\mathcal{V}_{E\rightarrow
E^{\prime}F}(\phi_{A_{1}\cdots A_{m}E}^{\rho}),
\end{equation}
and $\phi_{A_{1}\cdots A_{m}E}^{\rho}$ is a purification of $\rho_{A_{1}\cdots
A_{m}}$.
\end{lemma}

\begin{proof}
A proof follows easily from the definition of $E_{\sq}(A_{1};A_{2}%
;\cdots;A_{m})_{\rho}$ in \eqref{eq:mEsq}, rewriting it in terms of a squashing isometry as
has been done in the bipartite case, and employing duality of conditional entropy.
\end{proof}

\subsection{Multipartite Squashed Entanglement for GHZ and Private States}

The multipartite squashed entanglement of a maximally entangled state or a
private state scales linearly with the number of parties
\cite{YHHHOS,STW16}. That is, for $\Phi_{A_{1}\cdots A_{m}}$ a GHZ state as in
\eqref{eq:GHZ}  and $\gamma_{A_{1}\cdots A_{m}}$ a private state as
in \eqref{eq:multiPrivateState}, then the following
relations hold%
\begin{align}
 E_{\sq}(A_{1};\cdots;A_{m}%
)_{\Phi} &  =\frac{m}{2}\log_{2}K,\label{eq:eSqGHZ}\\
 E_{\sq}(A_{1};\cdots
;A_{m})_{\gamma} &  \geq\frac{m}{2}\log_{2}K.\label{eq:eSqPriv}%
\end{align}

Now consider a set
$\mathcal{W}=\{A,B,C\}$
of systems  and let $\Psi_{ABC}$ be
composed of maximally entangled states $\Phi$ and private states $\gamma$ over
the systems $A$, $B$, $\text{and }C$, according to the power set in
\eqref{eq:powerSet2} for two or more members:%
\begin{multline}
\Psi_{ABC}=\Phi_{A_{1}B_{1}}\otimes\Phi_{A_{2}C_{2}}\otimes\Phi_{B_{3}C_{3}%
}\otimes\Phi_{A_{4}B_{4}C_{4}}\\
\otimes\gamma_{A_{5}B_{5}}\otimes\gamma_{A_{6}C_{6}}\otimes\gamma_{B_{7}C_{7}%
}\otimes\gamma_{A_{8}B_{8}C_{8}}.
\label{eq:multiMaxPriv}
\end{multline}
In the above, we have subdivided the systems $A$, $B$, $\text{and }C$ for the
various correlations so that, in the given example,%
\begin{align}
A  & =A_{1}A_{2}A_{4}A_{5}A_{6}A_{8},\\
B  & =B_{1}B_{3}B_{4}B_{5}B_{7}B_{8},\\
C  & =C_{2}C_{3}C_{4}C_{6}C_{7}C_{8}.
\end{align}
For each of the constituent states given in \eqref{eq:multiMaxPriv}, we denote
the number of entangled bits or private bits as $E$ or $K$, respectively, as
done in \cite{STW16}. For example,
\begin{align}
E_{AB} &  =H(A_{1})_{\Phi}=H(B_{1})_{\Phi}=\log_{2}K_{A_{1}}%
,\label{eq:entBits}\\
K_{ABC} &  =H(A_{8})_{\gamma}=H(B_{8})_{\gamma}=H(C_{8})_{\gamma}=\log
_{2}K_{A_{8}},\label{eq:privBits}%
\end{align}
and so the tuple
\[
(E_{AB},E_{AC},E_{BC},E_{ABC},K_{AB},K_{AC},K_{BC},K_{ABC})
\]
characterizes the entangled and private bit content of the state. By using
\eqref{eq:eSqGHZ} and \eqref{eq:eSqPriv}, along with the additivity of
squashed entanglement for tensor-product states and adopting the notation in
\eqref{eq:entBits} and \eqref{eq:privBits}, we find that
\begin{align}
& E_{\sq}(A;B;C)_{\Psi}\nonumber\\
& =E_{\sq}(A_{1};B_{1})_{\Phi}+E_{\sq}(A_{2};C_{2})_{\Phi}+E_{\sq}(B_{3}%
;C_{3})_{\Phi}\nonumber\\
& \quad+E_{\sq}(A_{4};B_{4};C_{4})_{\Phi}+E_{\sq}(A_{5};B_{5})_{\gamma
}+E_{\sq}(A_{6};C_{6})_{\gamma}\nonumber\\
& \quad+E_{\sq}(B_{7};C_{7})_{\gamma}+E_{\sq}(A_{8};B_{8};C_{8})_{\gamma}\\
& \geq E_{AB}+E_{AC}+E_{BC}+\frac{3}{2}E_{ABC}\nonumber\\
& \quad+K_{AB}+K_{AC}+K_{BC}+\frac{3}{2}K_{ABC}
\label{eq:mult-sq-lwr-part-all}
\end{align}
As in \eqref{eq:eSqPart1} and \eqref{eq:eSqPart2}, if $\Psi_{ABC}%
=\Psi_{\mathcal{Y}}$ for $\mathcal{Y}=\{A,B,C\}$, and for partitions
$\mathbb{G}_{1}=\{\{A\},\{B\},\{C\}\}$ and $\mathbb{G}_{2}=\{\{AB\},\{C\}\}$
then $E_{\sq}(\mathbb{G}_{1})=E_{\sq}(A;B;C)_{\Psi}$ as shown in
\eqref{eq:eSqPart1}. For $E_{\sq}(\mathbb{G}_{2})$, we have that
\begin{align}
& E_{\sq}({\mathbb{G}_{2}})\nonumber\\
& =E_{\sq}(AB;C)_{\Psi}\\
& =E_{\sq}(A_{2};C_{2})_{\Phi}+E_{\sq}(B_{3};C_{3})_{\Phi}+E_{\sq}(A_{4}%
B_{4};C_{4})_{\Phi}\nonumber\\
& \quad+E_{\sq}(A_{6};C_{6})_{\gamma}+E_{\sq}(B_{7};C_{7})_{\gamma}%
+E_{\sq}(A_{8}B_{8};C_{8})_{\gamma}\\
& \geq E_{AC}+E_{BC}+E_{ABC}+K_{AC}+K_{BC}+K_{ABC}%
\end{align}

\section{Quantum Broadcast Channels and Secret-Key-Agreement Capacity Regions}

\label{s:broadcast} A quantum broadcast channel is a channel as defined in
\eqref{eq:kraus}, except that it is a map from one sender to \textit{multiple}
receivers \cite{YHD11}. A protocol for energy-constrained, multipartite secret
key agreement is much the same as in the bipartite case outlined in
Section~\ref{s:ECSKAC}, with a constraint on the average energy of the channel input
states and with rounds of LOCC between channel uses. For demonstrative
purposes, in this section we focus exclusively on the case of a single sender
and two receivers. We make use of an energy observable $G$
 and energy constraint $P\in\lbrack0,\infty)$. A quantum
broadcast channel $\mathcal{N}_{A\rightarrow BC}$ satisfies the finite
output-entropy condition with respect to $G$ and $P$ if%
\begin{align}
\sup_{\rho_{A}:\Tr\{G\rho_{A}\}\leq P}H(\Tr_C\{\mathcal{N}_{A\rightarrow BC}(\rho
_{A})\})<\infty ,\\
\sup_{\rho_{A}:\Tr\{G\rho_{A}\}\leq P}H(\Tr_B\{\mathcal{N}_{A\rightarrow BC}(\rho
_{A})\})<\infty .\label{eq:multFiniteOutputEntropy}%
\end{align}
That is, the output entropy to each receiver should be finite.
In what follows, for example, we denote the rate of entanglement generation
between the sender $A$ and the receiver $B$ by $R_{AB}^{E}$
and
the rate of key generation by
$R_{AB}^{K}$. Generalizing
this, we have a vector $\vec{R}$ of rates, for which we employ the following shorthand:%
\begin{equation}
\vec{R}\equiv(R_{AB}^{E},R_{AC}^{E},R_{BC}^{E},R_{ABC}^{E},R_{AB}^{K}%
,R_{AC}^{K},R_{BC}^{K},R_{ABC}^{K}).
\end{equation}

In a general $(n,\vec{R},G,P,\varepsilon)$ protocol, the sender Alice and
the receivers Bob and Charlie are tasked to use a quantum broadcast channel
$\mathcal{N}_{A\rightarrow BC}\ n$ times to establish a shared state
$\Omega_{ABC}$ such that
\begin{equation}
F(\Omega_{ABC},\Psi_{ABC})\geq1-\varepsilon,\label{eq:omegaPsi}%
\end{equation}
with $\Psi$ defined in \eqref{eq:multiMaxPriv} and the elements of $\vec{R}$ are given by, e.g., \cite{STW16}
\begin{align}
R_{AB}^{E} &  =\frac{E_{AB}}{n}=\frac{1}{n}H(A_1)_{\Psi},\label{eq:ratedef}\\
R_{AB}^{K} &  =\frac{K_{AB}}{n}=\frac{1}{n}H(A_5)_{\Psi}.
\end{align}
In such a protocol, Alice, Bob, and Charlie begin by performing an LOCC
channel $\mathcal{L}_{\emptyset\rightarrow A_{1}^{\prime}A_{1}B_{1}^{\prime
}C_{1}^{\prime}}^{(1)}$ to create a  state $\rho_{A_{1}^{\prime
}A_{1}B_{1}^{\prime}C_{1}^{\prime}}^{(1)}$
that is separable with respect to the cut 
$A_{1}^{\prime
}A_{1} | B_{1}^{\prime} | C_{1}^{\prime}$,
and where  $A_{1}%
^{\prime}$, $B_{1}^{\prime}$, and $C_{1}^{\prime}$
are scratch systems. Alice then uses $A_{1}$ as
the input to the first channel use, resulting in the state%
\begin{equation}
\sigma_{A_{1}^{\prime}B_{1}B_{1}^{\prime}C_{1}C_{1}^{\prime}}^{(1)}%
\equiv\mathcal{N}_{A_{1}\rightarrow B_{1}C_{1}}(\rho_{A_{1}^{\prime}A_{1}%
B_{1}^{\prime}C_{1}^{\prime}}^{(1)}).
\end{equation}
Alice, Bob, and Charlie then perform a second LOCC channel, producing
\begin{multline}
\rho_{A_{2}^{\prime}A_{2}B_{2}^{\prime}C_{2}^{\prime}}^{(2)}\equiv
\label{secondLOCCmult}\\
\mathcal{L}_{A_{1}^{\prime}B_{1}B_{1}^{\prime}C_{1}C_{1}^{\prime}\rightarrow
A_{2}^{\prime}A_{2}B_{2}^{\prime}C_{2}^{\prime}}^{(2)}(\sigma_{A_{1}^{\prime
}B_{1}B_{1}^{\prime}C_{1}C_{1}^{\prime}}^{(1)}).
\end{multline}
The procedure continues in this manner, as in Section~\ref{s:ECSKAC}, with a total of
$n$ rounds of LOCC interleaved with $n$ uses of the channel as follows: for
$i\in\{2,\ldots,n\}$%
\begin{multline}
\rho_{A_{i}^{\prime}A_{i}B_{i}^{\prime}C_{i}^{\prime}}^{(i)}\equiv\nonumber\\
\mathcal{L}_{A_{i-1}^{\prime}B_{i-1}B_{i-1}^{\prime}C_{i-1}C_{i-1}^{\prime
}\rightarrow A_{i}^{\prime}A_{i}B_{i}^{\prime}C_{i}^{\prime}}^{(i)}%
(\sigma_{A_{i-1}^{\prime}B_{i-1}B_{i-1}^{\prime}C_{i-1}C_{i-1}^{\prime}%
}^{(i-1)}),
\end{multline}%
\begin{equation}
\sigma_{A_{i}^{\prime}B_{i}B_{i}^{\prime}C_{i}C_{i}^{\prime}}^{(i)}%
\equiv\mathcal{N}_{A_{i}\rightarrow B_{i}C_{i}}(\rho_{A_{i}^{\prime}A_{i}%
B_{i}^{\prime}C_{i}^{\prime}}^{(i)}).
\end{equation}
After the $n$th channel use, a final, ($n+1$)th LOCC channel is performed.
Going to the purified picture as before, tracing over the eavesdropper's systems while retaining the shield systems, the goal is to
establish the state $\Omega_{ABC}$ satisfying $F(\Omega_{ABC},\Psi_{ABC}%
)\geq1-\varepsilon$, where $\Psi_{ABC}$ is the ideal state from \eqref{eq:multiMaxPriv}. Finally, the same
average energy constraint for the channel input states, as in \eqref{eq:energyConstraint}, should be satisfied.

The rate tuple $\vec{R}$ is achievable if for all $\varepsilon\in(0,1)$, $\vec{\delta} \succeq 0$, and
 sufficiently large $n$, there exists an $(n,\vec{R}-\vec{\delta},G,P,\varepsilon)$
protocol as outlined above. The energy-constrained secret-key-agreement
capacity region of the channel $\mathcal{N}$ is the closure of the region
mapped out by all achievable rate tuples subject to the energy constraint $P$.

\subsection{Energy-Constrained Squashed Entanglement Upper Bound for the
LOCC-Assisted Capacity Region of a Quantum Broadcast Channel}

\label{s:EsqBroad}The main result of this section is a generalization of the
result in Section~\ref{s:EsqBound}, as well as a generalization of the main result in
\cite{STW16}. In particular, we prove that the energy-constrained,
multipartite squashed entanglement is a key tool in bounding the LOCC-assisted
capacity region of a quantum broadcast channel.

\begin{theorem}
\label{thm:2-receiver}
Let $G$ be a Gibbs observable, and let $P\in\lbrack
0,\infty)$ be an energy constraint.
Let $\mathcal{N}_{A\rightarrow BC}$ be a quantum
broadcast channel satisfying the finite-output entropy condition in \eqref{eq:multFiniteOutputEntropy} with respect to $G$ and $P$. Suppose that $\vec{R}$ is an achievable
rate tuple for LOCC-assisted private and quantum communication. Then the elements of
the rate tuple $\vec{R}$ are bounded in terms of multipartite squashed
entanglement as%
\begin{align}
R_{AC}^{E}+R_{AC}^{K}+R &  _{BC}^{E}+R_{BC}^{K}+R_{ABC}^{E}+R_{ABC}%
^{K}\nonumber\\
&  \leq E_{\sq}(SB;C)_{\omega}\label{eq:bBoundABc}\\
R_{AB}^{E}+R_{AB}^{K}+R &  _{BC}^{E}+R_{BC}^{K}+R_{ABC}^{E}+R_{ABC}%
^{K}\nonumber\\
&  \leq E_{\sq}(SC;B)_{\omega}\label{eq:bBoundACb}\\
R_{AB}^{E}+R_{AB}^{K}+R &  _{AC}^{E}+R_{AC}^{K}+R_{ABC}^{E}+R_{ABC}%
^{K}\nonumber\\
&  \leq E_{\sq}(S;BC)_{\omega}\label{eq:bBoundAbc}\\
R_{AB}^{E}+R_{AB}^{K}+R &  _{AC}^{E}+R_{AC}^{K}+R_{BC}^{E}+R_{BC}%
^{K}\nonumber\\
+\frac{3}{2}\big(R_{ABC}^{E}+R &  _{ABC}^{K}\big)\nonumber\\
&  \leq E_{\sq}(S;B;C)_{\omega},\label{eq:bBoundABC}%
\end{align}
for some pure state $\psi_{SA}$ satisfying $\operatorname{Tr}\{G\psi_{A}\}\leq
P$, with the state $\omega_{SBC}$ defined in terms of it as%
\begin{equation}
\omega_{SBC}=\mathcal{N}_{A\rightarrow BC}(\psi_{SA}).\label{eq:pureoutput}%
\end{equation}

\end{theorem}

\begin{proof}
The proof of this bound follows that of Proposition~\ref{prop:EC-SE-upper-bnd} and \cite[Theorem~12]%
{STW16}, working backward through the communication protocol one channel use
at a time in order to demonstrate the inequalities. For this reason, we keep
the proof brief. Let us begin by considering the partition $\mathbb{G}%
_{1}=\{\{A\},\{B\},\{C\}\}$. From reasoning as in \eqref{eq:mult-sq-lwr-part-all} but instead
applying an estimate in \cite[Theorem~6]{appxpriv} to the condition $F(\Omega_{ABC}%
,\Psi_{ABC})\geq1-\varepsilon$, we find that%
\begin{multline}
n\Big(R_{AC}^{E}+R_{AC}^{K}+R_{BC}^{E}+R_{BC}^{K}+R_{AB}^{E}+R_{AB}^{K}\\
+\frac{3}{2}(R_{ABC}^{E}+R_{ABC}^{K})\Big)\leq E_{\sq}(A;B;C)_{\Omega}%
+f_{2}(n,\varepsilon) ,
\label{eq:multicont}
\end{multline}
where $f_{2}(n,\varepsilon)$ is a function such that
$f_{2}(n,\varepsilon)  / n$
tends to zero as
$n\rightarrow\infty$ and as $\varepsilon\rightarrow0$.

If we look at just the squashed entanglement term of \eqref{eq:multicont}, we
can split it and group terms, working backward through the $n$ channel uses of
the protocol:%
\begin{align}
&  E_{\sq}(A;B;C)_{\Omega}\nonumber\\
&  \leq E_{\sq}(A_{n}^{\prime};B_{n}B_{n}^{\prime};C_{n}C_{n}^{\prime
})_{\sigma^{(n)}}\\
&  \leq E_{\sq}(A_{n}^{\prime}B_{n}C_{n}E_{n};B_{n}^{\prime};C_{n}^{\prime
})_{\sigma^{(n)}}\nonumber\\
&  \qquad+E_{\sq}(A_{n}^{\prime}B_{n}^{\prime}C_{n}^{\prime}R_{n};B_{n}%
;C_{n})_{\sigma^{(n)}}\\
&  =E_{\sq}(A_{n}^{\prime}A_{n};B_{n}^{\prime};C_{n}^{\prime})_{\rho^{(n)}%
}\nonumber\\
&  \qquad+E_{\sq}(A_{n}^{\prime}B_{n}^{\prime}C_{n}^{\prime}R_{n};B_{n}%
;C_{n})_{\sigma^{(n)}}\\
&  \leq E_{\sq}(A_{n-1}^{\prime};B_{n-1}B_{n-1}^{\prime};C_{n-1}%
C_{n-1}^{\prime})_{\sigma^{(n-1)}}\nonumber\\
&  \qquad+E_{\sq}(A_{n}^{\prime}B_{n}^{\prime}C_{n}^{\prime}R_{n};B_{n}%
;C_{n})_{\sigma^{(n)}}\\
&  \leq\sum_{i=1}^{n}E_{\sq}(A_{i}^{\prime}B_{i}^{\prime}C_{i}^{\prime}%
R_{i};B_{i};C_{i})_{\sigma^{(i)}}.\label{eq:multiIterate}%
\end{align}
The first inequality follows from the monotonicity of squashed entanglement
under LOCC. For the second inequality the quantity has been split using the
subadditivity property from Lemma~\ref{lem:multiSubadditivity} (there are also
some implicit purifying systems $R$ and $E$, which we have not explicitly
defined, but note that $E$ denotes an environment of the broadcast channel).
The equality is a result of the invariance of squashed entanglement under
isometries, because an isometric extension of $\mathcal{N}$ relates $A_{n}$ to
$B_{n}C_{n}E_{n}$. The third inequality is the beginning of the first
repetition of this procedure, in which we again apply the monotonicity of
squashed entanglement under LOCC. Iterating this reasoning $n$ times leads to
the final inequality in \eqref{eq:multiIterate}. Working backward another step
yields no additional terms, because the initial state is separable, having
been created through LOCC. However, with purifying systems $R_{i}$, we combine
\eqref{eq:multiIterate} with  \eqref{eq:multicont} to
conclude that there exists a state $\omega$, as defined in
\eqref{eq:pureoutput}, such that%
\begin{equation}
\sum_{i=1}^{n}E_{\sq}(A_{i}^{\prime}B_{i}^{\prime}C_{i}^{\prime}R_{i}%
;B_{i};C_{i})_{\sigma^{(i)}}
\leq
nE_{\sq}(S;B;C)_{\omega}%
\end{equation}
and
\begin{multline}
R_{AC}^{E}+R_{AC}^{K}+R_{BC}^{E}+R_{BC}^{K}+R_{AB}^{E}+R_{AB}^{K}\\
+\frac{3}{2}\big(R_{ABC}^{E}+R_{ABC}^{K}\big)\\
\leq E_{\sq}(S;B;C)_{\omega}+\frac{1}{n}f_{2}(n,\varepsilon).
\end{multline}
Taking the limit $n\rightarrow\infty$ and then $\varepsilon\rightarrow0$
yields \eqref{eq:bBoundABC}. A similar rationale can be applied to obtain the
other bounds, and key to the claim, as in the proof of \cite[Theorem~12]%
{STW16}, is that the same state $\omega$ can be used in all of the bounds.
\end{proof}

\begin{remark}
\label{rem:gen-to-mult-rec}Just as \cite[Theorem~12]{STW16} was generalized
from the single-sender, two-receiver case to the single sender, $m$-receiver
case in \cite[Theorem~13]{STW16}, our above bounds for the energy-constrained
capacity region of the quantum broadcast channel can be generalized to an
$m$-receiver case through the consideration of the many possible partitions,
as described in Section~\ref{s:multiPrimer}.
\end{remark}

\subsection{Upper Bounds on the Energy-Constrained LOCC-Assisted Capacity
Regions of a Pure-Loss Bosonic Broadcast Channel}

\label{s:BoundCalcs}In this section, we focus on a concrete quantum broadcast
channel, known as the pure-loss broadcast channel. The model for this channel
was introduced in \cite{G08thesis}\ and subsequently studied in
\cite{STW16,tsw}. It is equivalent to a linear sequence of beamsplitters, in
which the sender inputs into the first one, the vacuum state is injected into
all of the environment ports, the receivers each get one output from the
sequence of beamsplitters and one output of the beamsplitters is lost to the
environment (see Figure~3-13 of \cite{G08thesis} or Figure~1c of \cite{tsw}).
In what follows, we adopt the same strategy as before for the single-mode
pure-loss channel (and what was subsequently used in \cite{STW16}), and we
relax the squashing isometry for the environment mode to be a 50-50 beamsplitter.

Using this strategy, we now calculate bounds on rates of energy-constrained entanglement
generation and key distillation achievable between the sender and one of the
receivers. The same reasoning as in Remark~\ref{rem:gen-sit-ph-ins-cond-ent}, along with the representation of
multipartite squashed entanglement in Lemma~\ref{lem:alt-rep-mult-sq-e} and
the relaxation of it described above, allow us to conclude that, for a given
input mean photon number constraint $N_{S}\geq0$, a thermal state of that
photon number is optimal.

Before stating the theorem, we establish the following notation:

\begin{itemize}
\item The set of all receivers is denoted by
$\mathcal{B}=\{B_{1},\dots,B_{m}\}$. The total transmissivity for all
receivers is $\eta_{\mathcal{B}} \in [0,1]$.

\item In the theorem below, the set $\mathcal{T}$ denotes a subset of the
receivers ($\mathcal{T}\subseteq\mathcal{B}$), and its complement set is
denoted by $\overline{\mathcal{T}}=\mathcal{B}\backslash\mathcal{T}$. The
total transmissivity to the members of the set $\mathcal{T}$ is denoted by
$\eta_{\mathcal{T}}=\sum_{B_{i}\in\mathcal{T}}\eta_{B_{i}}$, and the total
transmissivity to the members of the complement set is denoted by
$\eta_{\overline{\mathcal{T}}}=\sum_{B_{i}\in\overline{\mathcal{T}}}%
\eta_{B_{i}}$, such that $\eta_{\mathcal{T}}+\eta_{\overline{\mathcal{T}}%
}=\eta_{\mathcal{B}}$.

\item The transmissivity to the adversary Eve is denoted by $\eta_{E}%
=1-\eta_{\mathcal{B}}=1-\eta_{\mathcal{T}}-\eta_{\overline{\mathcal{T}}}$.
\end{itemize}

With this notation, we can now establish the following theorem:

\begin{theorem}
\label{thm:mult-pure-loss}The energy-constrained LOCC-assisted capacity region
of a pure-loss quantum broadcast channel, for entanglement and key
distillation between the sender and each receiver, is bounded as%
\begin{multline}
\sum_{B_{i}\in\mathcal{T}}R_{AB_{i}}^{E}+R_{AB_{i}}^{K}\leq g(
N_{S}(1+\eta_{\mathcal{T}}-\eta_{\overline{\mathcal{T}}})/2)  \\
-g(  N_{S}(1-\eta_{\mathcal{T}}-\eta_{\overline{\mathcal{T}}})/2)
.\label{eq:thmRateBoundforGroupings}%
\end{multline}
for all non-empty $\mathcal{T}\subseteq\mathcal{B}$.
\end{theorem}

\begin{proof}
For the choices discussed above, it simply suffices to calculate various
relaxations of the multipartite squashed entanglements when the thermal state of mean photon
number $N_{S}$ is input.
As mentioned above, the same reasoning as in Remark~\ref{rem:gen-sit-ph-ins-cond-ent}, along with the representation of
multipartite squashed entanglement in Lemma~\ref{lem:alt-rep-mult-sq-e} and
the relaxation of it described above, allow us to conclude that, for a given
input mean photon number constraint $N_{S}\geq0$, a thermal state of that
photon number is optimal.
By applying Theorem~\ref{thm:2-receiver} and
Remark~\ref{rem:gen-to-mult-rec}, the following bounds apply%
\begin{align}
\sum_{B_{i}\in\mathcal{T}}R_{AB_{i}}^{E}+R_{AB_{i}}^{K}
& \leq E_{\sq}%
(R\overline{\mathcal{T}};\mathcal{T}),\label{eq:groupingBound}\\
& \leq 
\frac{1}{2}[H(\mathcal{T}|E_{1})+H(\mathcal{T}|E_{2}%
)]
\end{align}
where  the second inequality follows from relaxing the squashing
isometry to be a 50-50 beamsplitter as discussed above, with output systems
$E_{1}$ and $E_{2}$, and then it follows that the thermal state of mean photon number
$N_{S}$ into the pure-loss bosonic broadcast channel is optimal. Now employing  entropy identities, we find
that%
\begin{align}
&  \frac{1}{2}[H(\mathcal{T}|E_{1})+H(\mathcal{T}|E_{2}%
)]\notag \\
&  =\frac{1}{2}[H(\mathcal{T}E_{1})-H(E_{1})+H(\mathcal{T}E_{2})-H(E_{2}%
)]\label{eq:splitCondEnts}\\
&  =H(\mathcal{T}E_{1})-H(E_{1})\label{eq:combCondEnts}. %
\end{align}
The last line in \eqref{eq:combCondEnts} combines terms that are equal, due to
the fact that the transmissivity of the squashing channel is balanced
(coming from a 50-50 beamsplitter). We then use the $g$ function to represent the
entropies of the thermal states resulting from the use of the quantum
broadcast channel, giving that%
\begin{align}
&  H(\mathcal{T}E_{1})-H(E_{1})\nonumber\\
&  =g\left(  N_{S}(\eta_{\mathcal{T}}+\eta_{E}/2)\right)  -g\left(  N_{S}%
\eta_{E}/2\right)  \label{eq:gotoGfunction}\\
&  =g\left(  N_{S}(\eta_{\mathcal{T}}+(1-\eta_{\mathcal{T}%
}-\eta_{\overline{\mathcal{T}}}))/2\right)  \nonumber\\
&  \qquad-g\left(  N_{S}(1-\eta_{\mathcal{T}}-\eta_{\overline{\mathcal{T}}%
})/2\right)  \label{eq:gtrans}\\
&  =g\left(  N_{S}(1+\eta_{\mathcal{T}}-\eta_{\overline{\mathcal{T}}%
})/2\right)  \notag \\
& \qquad -g\left(  N_{S}(1-\eta_{\mathcal{T}}-\eta_{\overline{\mathcal{T}%
}})/2\right) . \label{eq:theBoundCalc}%
\end{align}
This concludes the proof.
\end{proof}

We conclude this section with a few brief remarks. In the limit of large
photon number $N_{S}\rightarrow\infty$, the bound in
Theorem~\ref{thm:mult-pure-loss}\ reduces to%
\begin{equation}
\sum_{B_{i}\in\mathcal{T}}R_{AB_{i}}^{E}+R_{AB_{i}}^{K}\leq\log_{2}\!\left(
\frac{1+\eta_{\mathcal{T}}-\eta_{\overline{\mathcal{T}}}}{1-\eta_{\mathcal{T}%
}-\eta_{\overline{\mathcal{T}}}}\right)  ,
\end{equation}
which is not as tight as the result of \cite{tsw}, in which the upper bound
was found to be $\log_{2}\!\left(  \frac{1-\eta_{\overline{\mathcal{T}}}}%
{1-\eta_{\mathcal{T}}-\eta_{\overline{\mathcal{T}}}}\right)  $. However, for
low photon number, the energy-constrained bounds of
Theorem~\ref{thm:mult-pure-loss} can be tighter.

Let us look at some particular examples of the bound. For the case of two
receivers, Bob and Charlie, the set $\mathcal{T}$ can take a few different
values. If $\mathcal{T}=\{B,C\}$ then $\overline{\mathcal{T}}=0$ and
\begin{multline}
R_{AB}^{E}+R_{AB}^{K}+R_{AC}^{E}+R_{AC}^{K}+R_{ABC}^{E}+R_{ABC}^{K}\\
\leq\log_{2}\!\left(  \frac{1+\eta_{B}+\eta_{C}}{1-\eta_{B}-\eta_{C}}\right)
\end{multline}
which has been discussed already in \cite{STW16}. For the case $\mathcal{T}%
=C$, then $\overline{\mathcal{T}}=B$, and so%
\begin{equation}
R_{AC}^{E}+R_{AC}^{K}\leq\log_{2}\!\left(  \frac{1+\eta_{C}-\eta_{B}}{1-\eta
_{B}-\eta_{C}}\right)  .
\end{equation}
Other permutations of the sets $\mathcal{T}$ and $\overline{\mathcal{T}}$ can
naturally be worked out for scenarios involving any number of receivers.

\section{Conclusion} \label{s:concl}

Knowing not only the achievable rates of current protocols but also  fundamental limitations of a channel for secret key agreement or LOCC-assisted quantum communication is important for the implementation of rapidly progressing quantum technologies. In this paper, we formally defined the task of energy-constrained secret key agreement and LOCC-assisted quantum communication. We  proved that the energy-constrained squashed entanglement is an upper bound on these capacities. We also proved that a thermal-state input is optimal for a relaxation of the energy-constrained squashed entanglement of a single-mode input, phase-insensitive bosonic Gaussian channel, generalizing results from prior work on this topic. After doing so, we proved that a variation of the method introduced in \cite{goodEW16} leads to improved upper bounds on the energy-constrained secret-key-agreement capacity of a bosonic thermal channel. In particular, these improved upper bounds have the property that they converge to zero in the limit as the thermal channel becomes entanglement breaking.

We then generalized the results to the multipartite setting, along the lines of \cite{STW16}. Here, we began by proving that two multipartite squashed entanglements are in fact equal even though they were previously thought to be different. We also proved that the energy-constrained multipartite squashed entanglement serves as an upper bound on the energy-constrained, secret key agreement and LOCC-assisted quantum capacity region of a quantum broadcast channel. We then applied the presented squashed entanglement bounds to the pure-loss bosonic broadcast channel with an arbitrary number of receivers, and the special case of communication between a sender and each of the individual receivers.

Since the squashed entanglement bounds presented here are independent of the physical examples given, we expect it to apply to other systems not discussed here. 

In the future, our bound should be examined in the context of a limited number of channel uses in addition to the energy constraint. It still remains an open question from \cite{tgw,tgwB,STW16} to determine whether the squashed entanglement bounds could serve as strong converse rates.
We also think it is clear that our formalism can be generalized to even more settings, such as those considered in \cite{AML16,BA17,RKBKMA17}.
An important technical question is whether the energy-constrained squashed entanglement bounds could apply
when the LOCC channels involved are not countably decomposable, and answering this question is directly related to the question discussed in \cite[Remark~1]{maks}. 
Finally, we think it would be  interesting to find physical examples
outside of the bosonic setting
to which our general theory could apply.

\begin{acknowledgments}
We are grateful to Kenneth Goodenough, Saikat Guha, Masahiro Takeoka, and Kaushik Seshadreesan for discussions regarding this research.
We are especially grateful to Kenneth Goodenough for many insightful discussions about his prior results in \cite{goodEW16} and for his suggestions regarding the bound in \eqref{eq:DSW18-bnd}.
ND acknowledges support from the Department of Physics and Astronomy at LSU and the National Science Foundation under Grant No.~1714215. MMW acknowledges support from the Office of Naval Research. 
\end{acknowledgments}

\bibliographystyle{alpha}
\bibliography{oprogref}

\end{document}